\documentclass{tlp}

\usepackage{aopmath}

\usepackage{lmodern}
\usepackage[T1]{fontenc}
\usepackage[latin9]{inputenc}
\usepackage{bm} 
\usepackage{amstext}
\usepackage{amsfonts}
\usepackage{amssymb}
\usepackage{xcolor}
\usepackage{framed}

\usepackage{url}
\usepackage{hyperref}
\usepackage{verbatim}
\usepackage[english]{babel}

\usepackage{longtable}
\usepackage{xr}

\newtheorem{definition}{Definition}[section]
\newtheorem{example}[theorem]{Example}
\newtheorem{claim}[theorem]{Claim}

\newtheorem{remark}{Remark}

\newenvironment{cenchop}{\renewcommand{\arraystretch}{1.3}\begin{center}\begin{tabular}{l}}{\end{tabular}\end{center}\renewcommand{\arraystretch}{1}}

\newcommand{\qsp}{\\[1.2ex]} 

\newcommand{\qed}{\hfill $\square$}

\newcommand{\langname}[1]{\text{#1}}  
\newcommand{\pred}[1]{\mathtt{#1}}  
\newcommand{\fname}[1]{\mathit{#1}}  
\newcommand{\sq}[1]{`{#1}'}  

\newcommand{\dedalus}{\langname{Dedalus}}
\newcommand{\datalogneg}{\langname{Datalog}^{\neg}}
\newcommand{\datalog}{\langname{Datalog}}
\newcommand{\webdamlog}{\langname{Webdamlog}}
\newcommand{\lups}{\langname{LUPS}}

\renewcommand{\qed}{$\hfill \square$}  
\newcommand{\tline}{\hrulefill}
\newcommand{\Nat}{\mathbb{N}}  
\newcommand{\len}[1]{|#1|} 
\newcommand{\rom}[1]{\text{\emph{(#1)}}} 
\newcommand{\romI}{\rom i}
\newcommand{\romII}{\rom{ii}}
\newcommand{\romIII}{\rom{iii}}
\newcommand{\romIV}{\rom{iv}}
\newcommand{\ded}{\mathcal{P}}  

\newcommand{\timecomplexity}[1]{{\scriptstyle \text{#1}}}

\newcommand{\ptime}{\timecomplexity{PTIME}}

\newcommand{\myset}[1]{\{#1\}}

\newcommand{\fc}{\boldsymbol{f}}
\newcommand{\fcB}{\boldsymbol{g}}
\newcommand{\rar}[2]{#1/#2}
\newcommand{\proj}[2]{#1|_{#2}}
\newcommand{\adom}[1]{\fname{adom}(#1)}
\newcommand{\sch}{\mathcal{D}}
\newcommand{\univ}{\mathbf{dom}}
\newcommand{\ntup}{(\,)}

\newcommand{\rl}{\varphi}
\newcommand{\head}[1]{\fname{head}_{#1}}
\newcommand{\bpos}[1]{\fname{pos}_{#1}}
\newcommand{\bneg}[1]{\fname{neg}_{#1}}
\newcommand{\idb}[1]{\fname{idb}(#1)}
\newcommand{\edb}[1]{\fname{edb}(#1)}
\newcommand{\schof}[1]{\fname{sch}(#1)}
\newcommand{\vars}[1]{\fname{vars}(#1)}
\newcommand{\uvar}{\mathbf{var}}
\newcommand{\var}[1]{\mathtt{#1}}
\newcommand{\tvar}[1]{\mathtt{\bar{#1}}}

\newcommand{\grl}{\psi}  
\newcommand{\gr}[2]{\fname{ground}(#1,#2)} 
\newcommand{\grp}[3]{\fname{ground}_{#1}(#2,#3)}
  
\newcommand{\nw}{\mathcal{N}}
 
\newcommand{\sh}[1]{(#1)}

\newcommand{\toloct}[1]{#1^{\mathrm{LT}}}

\newcommand{\addlt}[3]{#1^{\Uparrow#2,#3}}

\newcommand{\projlt}[3]{#1|^{#2,#3}}

\newcommand{\shprojlt}[3]{\projlt{\sh{#1}}{#2}{#3}}
\newcommand{\shaddlt}[3]{\addlt{\sh{#1}}{#2}{#3}}
\newcommand{\droplt}[1]{#1^{\Downarrow}}
\newcommand{\reltime}{\pred{time}}
\newcommand{\timesucc}{\pred{tsucc}}
\newcommand{\relall}{\pred{all}}
\newcommand{\schtime}{\sch_{\mathrm{time}}}
\newcommand{\Itime}{I_{\reltime}}
\newcommand{\before}{\pred{before}}
\newcommand{\chosen}{\pred{chosen}}
\newcommand{\other}{\pred{other}}
\newcommand{\cand}{\pred{cand}}
\newcommand{\rcvinf}{\pred{rcvInf}}
\newcommand{\issmaller}{\pred{isSmaller}}
\newcommand{\hasmax}{\pred{hasMax}}
\newcommand{\hassender}{\pred{hasSender}}
  
\newcommand{\purech}[1]{\fname{pure}_{\mathrm{ch}}(#1)}

\newcommand{\cronrulecandidate}[1]{\cand_{R}(\var x,\var s,\var y,\var t,\tvar u)\gets\addlt{\simplebody{\tvar u,\tvar v,\var y}}{\var x}{\var s},\,\relall(\var y),\,\reltime(\var t)#1 }

\newcommand{\decl}[1]{\fname{decl}(#1)}

\newcommand{\purecaus}[1]{\fname{pure}_{\mathrm{ca}}(#1)}

\newcommand{\pure}[1]{\fname{pure}(#1)}

\newcommand{\cnf}{\rho}
\newcommand{\cnfstart}[2]{\fname{\fname{start}}(#1,#2)}
\newcommand{\cnfs}{\mathit{st}}
\newcommand{\cnfb}{\mathit{bf}}
\newcommand{\pair}[2]{(#1,#2)}
\newcommand{\run}{\mathcal{R}}
\newcommand{\trace}[1]{\fname{trace}(#1)}
  
\newcommand{\untag}[1]{\fname{untag}(#1)}
\newcommand{\sendto}[2]{\delta^{#1\to#2}}

\newcommand{\deduc}[1]{\fname{deduc}_{#1}}
\newcommand{\induc}[1]{\fname{induc}_{#1}}
\newcommand{\async}[1]{\fname{async}_{#1}}
\newcommand{\mstep}[1]{(#1)}

\newcommand{\locR}[1]{\fname{loc}_{\run}(#1)}
\newcommand{\globR}[1]{\fname{glob}_{\run}(#1)}
\newcommand{\nwnat}{\nw\times\Nat}
  
\newcommand{\arr}{\alpha_{\run}}

\newcommand{\slice}[1]{\mathrm{trans}_{\run}^{[#1]}}
\newcommand{\slicefin}[1]{\text{fin}_{\run}^{[#1]}}
\newcommand{\sliceduc}[1]{\text{duc}_{\run}^{[#1]}}
\newcommand{\slicesnd}[1]{\text{snd}_{\run}^{[#1]}}
\newcommand{\slicecaus}[1]{\text{caus}_{\run}^{[#1]}}
\newcommand{\senders}[1]{\text{senders}_{\run}^{[#1]}}
\newcommand{\mesg}[1]{\mathrm{mesg}_{\run}^{[#1]}}

\newcommand{\caus}{\prec_{\run}}
\newcommand{\mesgev}[1]{\fname{mesg}(#1)}

\newcommand{\grded}{G}
\newcommand{\cauM}{\prec_{M}}

\newcommand{\totM}{<_{M}}
\newcommand{\globM}[1]{\fname{glob}_{M}(#1)}
\newcommand{\locM}[1]{\fname{loc}_{M}(#1)}
\newcommand{\Mind}{M^{\mathrm{ind}}}
\newcommand{\tup}{\tau}

\begin{document}

\title{Putting Logic-Based Distributed Systems on Stable Grounds}

\author[Ameloot et al.]
{Tom~J.~Ameloot\thanks{T.J.~Ameloot is a Postdoctoral Fellow of the Research Foundation -- Flanders (FWO).}, Jan~Van~den~Bussche \\
    Hasselt University \& transnational University of Limburg
\and William~R.~Marczak\\
    University of California, Berkeley
\and Peter~Alvaro \\
    University of California, Santa Cruz
\and Joseph~M.~Hellerstein\\
    University of California, Berkeley}

\submitted{5 September 2012}
\revised{12 November 2013, 24 April 2015}
\accepted{16 July 2015}

\maketitle

\begin{abstract}
In the Declarative Networking paradigm, $\datalog$-like languages
are used to express distributed computations. Whereas recently formal
operational semantics for these languages have been developed, a corresponding
declarative semantics has been lacking so far. The challenge is to
capture precisely the amount of nondeterminism that is inherent to
distributed computations due to concurrency, networking delays, and
asynchronous communication. This paper shows how a declarative, model-based
semantics can be obtained by simply using the well-known stable model
semantics for $\datalog$ with negation. We show that the model-based
semantics matches previously proposed formal operational
semantics.

\emph{To appear in Theory and Practice of Logic Programming (TPLP).}
\end{abstract}

\begin{keywords}
 Dedalus, Datalog, stable model semantics, distributed system, asynchronous communication
\end{keywords}

\section{Introduction}

Cloud environments have emerged as a modern way to store and manipulate
data \cite{cloud,cavage_2013}. 
For our purposes, a cloud is a distributed
system that should produce output as the result of some computation. We use the common term ``node'' as a synonym for an individual computer or server in a network.

In recent years, logic programming has been proposed as an attractive
foundation for distributed and cloud programming, building on work
in declarative networking \cite{decl_netw_cacm}.
The essential idea in declarative networking, is that the programmer uses a high-level declarative language (like $\datalog$) to specify only what has to happen, and not exactly how. 
For example, the programmer could specify only that certain messages are generated in reply to other messages; the exact technical details to send (and possibly resend) messages over transmission protocols are filled in by some runtime engine. This frees the programmer from thinking in low-level terms that distract from the actual meaning of the specific program at hand. 
In particular, complex distributed algorithms and protocols
can be expressed in relatively few lines of code \cite{trevor_SD3,i-do-declare,hellerstein_declimp}.
Besides the interest in declarative networking, we are also seeing a more general resurgence of $\datalog$ (with negation)
\cite{datalog2.0,datalog_sigmod}. 
Moreover, issues related to data-oriented distributed computing are receiving attention at database theory conferences \cite{hellerstein_datalog_talk,rtdn_pods,webdamlog,ameloot_consistency,green_winmove}.

One of the latest languages proposed in declarative networking is
$\dedalus$~\cite{dedalus,dedalus-datalog20,hellerstein_declimp},
a $\datalog$-inspired language that has influenced other recent language
designs for distributed and cloud computing such as $\webdamlog$
\cite{webdamlog} and $\langname{Bloom}$ \cite{hellerstein_calm}.

\paragraph*{Model-based semantics}

In this paper, we describe the meaning of distributed $\datalog$
programs using a model-based semantics. This approach contrasts with
most previous work in declarative networking, where the meaning of
programs was typically described with an operational semantics \cite{deutsch_network,declnetw_opsem,grumbach_netlog,rtdn_pods},
with a few exceptions~\cite{lobo_2012,ma_decldist}. 

There are several important motivations for a model-based semantics
of a distributed program. First, we can better separate the program
structure, i.e., the rules, from the (distributed) implementation
that may change over time. For example, consider rules that generate
messages. These rules \emph{can} be implemented with asynchronous
communication, but how we evaluate them across machines is eventually
just a physical performance decision. Said differently, the point
of message rules is not to model a physical phenomenon, but rather
to admit a wider array of physical implementations than a local evaluation
strategy. Model-based interpretations of a program admit all such
implementations, and can perhaps suggest some new ones. Second, we
can investigate the \emph{need} for time: we can think about when
temporal delay is needed for expressivity, rather than when it is
imposed upon us by some implementation detail like physical separation
of nodes. In this context we mention the CRON conjecture by Hellerstein,
that relates causality on messages to the nature of the computations
in which those messages participate~\cite{hellerstein_declimp,ameloot_noncausality2014}.
We elaborate on causality below.

Concretely, our approach will be to model a distributed program with
$\datalog$ under the stable model semantics~\cite{gelfond_stable}
because this semantics is widely used in logic programming. Following
the language $\dedalus$~\cite{dedalus,dedalus-datalog20,hellerstein_declimp},
we express the functionality of the distributed program with three
kinds of rules: ``deductive rules'' for local computation, ``inductive
rules'' for persisting memory across local computation steps, and,
``asynchronous rules'' for representing message sending. The asynchronous
rules will nondeterministically choose the arrival times of messages~\cite{datalog_choice,sz_nondet}.

However, using only the above rules is not sufficient, as this still
allows stable models that express undesirable computations, where
messages can be sent ``into the past''. Therefore, each program
is augmented with a set of rules that express \emph{causality} on
the messages. Causality stands for the physical constraint that an
effect can only happen after its cause. Applied to message delivery,
this intuitively means that a sent message can only be delivered in
the future, not in the past. The rules for causality reason from the
perspective of the local times of each node, which is a justified
approach since there is no common ``global clock'' in a distributed
environment \cite{attiyawelch_dcbook}. As a second improvement, we
also introduce rules to ensure that only a finite number of messages
arrive at each local step of a node, as occurs in a real distributed
system. Applying the stable model semantics to the augmented $\datalog$
programs constitutes our modeling of a distributed ($\datalog$)
program.

On another note, it is already well-known that for finite input domains,
the combination of $\datalog$ and stable model semantics allows for
expressing all problems in NP~\cite{marek_stable}. However, it is
not yet clear what can be represented when infinite input domains
are considered. From this perspective, our work demonstrates that
the stable model semantics is indeed also suitable for modeling distributed
programs, whose execution is unbounded in time. Here, time would be
provided as an infinite input.

\paragraph{Correctness}

As we have motivated above, our goal is to describe the workings of
a distributed system declaratively, so that new insights can emerge
from this perspective. Hence, it is important to verify that the model-based
semantics really corresponds to the execution of a distributed program.

To this end, we additionally formalize the execution of a distributed
$\datalog$ program by means of an \emph{operational} semantics
\cite{deutsch_network,declnetw_opsem,grumbach_netlog,rtdn_pods}.
This second semantics is defined as a transition system. The transition
system is infinite because nodes run indefinitely and keep sending
messages. In addition, the transition system is highly nondeterministic,
because nodes work concurrently and messages can be delayed.

We establish rigorously a correspondence between the features of the
operational semantics and the features of the proposed model-based
semantics. To formulate our result, we describe each operational execution
by a structure that we call a \emph{trace}, which includes for each node in the network the detailed information about the local
steps it has performed and about the messages it has sent and received.
For our distributed $\datalog$ programs, we show that such operational
traces correspond to the set of stable models.

\paragraph*{Outline}

This paper is organized as follows. First, Section~\ref{sec:related-work}
discusses related work. Section~\ref{sec:preliminaries} gives preliminaries.
Next, Section~\ref{sec:declarative} represents distributed $\datalog$
programs under the model-based semantics; this section is based on $\dedalus$, a $\datalog$-like language.
Section~\ref{sec:operational}
justifies the intuitions of the model-based semantics by establishing
an equivalence with an operational semantics. Section~\ref{sec:discussion}
finishes with the conclusion.

\section{Related Work}

\label{sec:related-work}

The work of \citeNS{lobo_2012} is closely related to our
work. For a $\dedalus$-inspired language, they give a model-theoretic
semantics based on answer set programming, i.e., stable models. To
define this semantics, they syntactically translate the rules of their
language to $\datalog$, where all literals are given an explicit
location and time variable, to represent the data that each node has
during each local time. This translation resembles the model-theoretic
semantics for distributed $\datalog$ programs in this paper. To enforce
natural execution properties in their semantics, like causality, Lobo
et al.\ specify auxiliary rules in the syntactical translation. The
work of \citeNS{lobo_2012} does not yet mention the connection
between the model-theoretic semantics and desired executions of a
distributed system, i.e., an operational semantics.

Extending the work of Lobo et al, the work of \citeNS{ma_decldist}
formalizes a distributed system as a composition of I/O automata~\cite{lynch_book}.
An operational execution of such a system is a sequence of valid transitions,
called a trace. Global properties of the system can be analyzed by
translating it into a logic program, to which an answer set solver
can be applied. Ma et al.\ mention that operational traces of the
system correspond to answer sets of the logic program, and that this
provides a formal foundation for the analysis tools based on answer
set programming. Thus, the work of \citeNS{ma_decldist}
indicates a practical benefit of having a correspondence between a
declarative and operational semantics for languages used in declarative
networking. As mentioned above, we also establish a similar correspondence
in the current paper, for our distributed $\datalog$ programs. We
note, however, a few differences between our work and that of Ma et
al. First, in the work of Ma et al, the message buffer of a node has
a maximum size. In our operational semantics, the buffers are unbounded.
Moreover, Ma et al.\ construct their logic programs for a fixed range
of timestamps. In our declarative, model-based semantics, time is
given as an infinite input to a $\datalog$ program whose rules are
independent of a fixed time range. Lastly, our work devotes much attention
to rigorously showing the correspondence between the declarative and
operational semantics, whereas this is not elaborated in the work
of Ma et al.

Also in the setting of distributed systems, \citeNS{interlandi_2013}
give a $\dedalus$-inspired language for describing \emph{synchronous}
systems. In such systems, the nodes of the network proceed in rounds
and the messages can not be arbitrarily delayed. During each round,
the nodes share the same global clock. Interlandi et al.\ specify
an operational semantics for their language, based on relational transducer
networks \cite{ameloot_jacm_2013}. They also show that this operational
semantics coincides with a model-theoretic semantics of a single holistic
$\datalog$ program.  It should be noted that \citeNS{lobo_2012},
and the current paper, deal with \emph{asynchronous} systems, that
in general pose a bigger challenge for a distributed program to be
correct, i.e., the program should remain unaffected by nondeterministic
effects caused by message delays. 

An area of artificial intelligence that is closely related to declarative
networking is that of programming multi-agent systems in declarative
languages. The knowledge of an agent can be expressed by a logic program,
which also allows for non-monotone reasoning, and agents update their
knowledge by modifying the rules in these logic programs \cite{leite_minerva,nigam_agents,leite_evolp}.
The language $\lups$ \cite{alferes_lups} was designed to specify
such dynamic updates to logic programs, and $\lups$ is also a declarative
language itself. After applying a sequence of updates specified in
$\lups$, the semantics of the resulting logic program can be defined
in an inductive way. But an interesting connection to this current
work, is that the semantics can also be given by first syntactically
translating the original program and its updates into a single normal
logic program, after which the stable model semantics is applied~\cite{alferes_lups}.
It should be noted however that in this second semantics, there is
no modeling of causality or the sending of messages.

Of course, logic programming is not the only means for specifying
a (distributed) system. For example, in the area of formal methods,
logic-based languages like TLA~\cite{lamport_dist_tla}, Z~\cite{using_z},
and Event-B~\cite{event_b} can be used to specify various distributed
algorithms. Specifications written in these languages can also be
automatically checked for correctness.

Although we work within the established setting of declarative networking~\cite{decl_netw_cacm}, the scientific debate on the merits of $\datalog$ versus other formalisms for programming distributed systems remains open. 
It seems desirable to have an analysis of how features of $\datalog$ relate to the features of other languages for formal specification, e.g.~\cite{lamport_dist_tla,using_z,event_b}, both on the syntactical and the semantical level. 
However, a deep understanding of the other languages would be needed. Moreover, one may expect that features of $\datalog$ will in general not map naturally to features of the other languages. Hence, we consider such a comparison to be a separate research project, outside the scope of the current paper.

\section{Preliminaries}

\label{sec:preliminaries}

\subsection{Database Basics}

\label{sub:databases-basics}

A \emph{database schema} $\sch$ is a finite set of pairs $(R,k)$
where $R$ is a \emph{relation name} and $k\in\Nat$ its associated
\emph{arity}. A relation name occurs at most once in a database
schema. We often write $(R,k)$ as $\rar Rk$.

We assume some infinite universe $\univ$ of atomic data values. A
\emph{fact} $\fc$ is a pair $(R,\bar{a})$, often denoted as $R(\bar{a})$,
where $R$ is a relation name and $\bar{a}$ is a tuple of values
over $\univ$. For a fact $R(\bar{a})$, we call $R$ the \emph{predicate}.
We say that a fact $R(a_{1},\ldots,a_{k})$ is \emph{over} database
schema $\sch$ if $\rar Rk\in\sch$. A database \emph{instance}
$I$ over $\sch$ is a set of facts over $\sch$. For a subset $\sch'\subseteq\sch$,
we write $\proj I{\sch'}$ to denote the subset of facts in $I$ whose
predicate is a relation name in $\sch'$. We write $\adom I$ to denote
the set of values occurring in facts of $I$.

\subsection{Datalog with Negation}

\label{sub:datalog-with-negation}

We recall $\datalog$ with negation \cite{ahv_book}, abbreviated
$\datalogneg$. We assume the standard database perspective, where
a $\datalogneg$ program is evaluated over a given set of facts, i.e.,
where these facts are not part of the program itself.

Let $\uvar$ be a universe of \emph{variables}, disjoint from $\univ$.
An \emph{atom} is of the form $R(u_{1},\ldots,u_{k})$ where $R$
is a relation name and $u_{i}\in\uvar\cup\univ$ for each $i=1,\ldots,k$.
We call $R$ the \emph{predicate}. If an atom contains no data values,
we call it \emph{constant-free}. A \emph{literal} is an atom or
an atom with ``$\neg$'' prepended. A literal that is an atom is
called \emph{positive} and otherwise it is called \emph{negative}.

It will be technically convenient to use a slightly unconventional
definition of rules. Formally, a $\datalogneg$ \emph{rule} $\rl$
is a triple
\[
(\head{\rl},\,\bpos{\rl},\,\bneg{\rl})
\]
where $\head{\rl}$ is an atom; $\bpos{\rl}$ and $\bneg{\rl}$ are
sets of atoms; and, the variables in $\rl$ all occur in $\bpos{\rl}$.
This last condition is called \emph{safety}. The components $\head{\rl}$,
$\bpos{\rl}$ and $\bneg{\rl}$ are called respectively the \emph{head},
the \emph{positive body atoms} and the \emph{negative body atoms}.
We refer to $\bpos{\rl}\cup\bneg{\rl}$ as the \emph{body atoms}.
Note, $\bneg{\rl}$ contains just atoms, not negative literals. Every
$\datalogneg$ rule $\rl$ must have a head, whereas $\bpos{\rl}$
and $\bneg{\rl}$ may be empty. If $\bneg{\rl}=\emptyset$ then $\rl$
is called \emph{positive}.

A rule $\rl$ may be written in the conventional syntax. For instance,
if $\head{\rl}=T(\var u,\var v)$, $\bpos{\rl}=\{R(\var u,\var v)\}$
and $\bneg{\rl}=\{S(\var v)\}$, with $\var u,\var v\in\uvar$, then
we can write $\rl$ as
\[
T(\var u,\var v)\gets R(\var u,\var v),\,\neg S(\var v).
\]
The specific ordering of literals to the right of the arrow has no
significance in this paper.

The set of variables of $\rl$ is denoted $\vars{\rl}$.  If $\vars{\rl}=\emptyset$
then $\rl$ is called \emph{ground}, in which case $\{\head{\rl}\}\cup\bpos{\rl}\cup\bneg{\rl}$
is a set of facts.

Let $\sch$ be a database schema. A rule $\rl$ is said to be \emph{over
schema} $\sch$ if for each atom $R(u_{1},\ldots,u_{k})\in\{\head{\rl}\}\cup\bpos{\rl}\cup\bneg{\rl}$
we have $\rar Rk\in\sch$. A $\datalogneg$ \emph{program $P$ over
$\sch$} is a set of (safe) $\datalogneg$ rules over $\sch$. 
We write $\schof P$ to denote the smallest database schema that $P$
is over; note, $\schof P$ is uniquely defined. We define $\idb P\subseteq\schof P$
to be the database schema consisting of all relations in rule-heads
of $P$. We abbreviate $\edb P=\schof P\setminus\idb P$.%
\footnote{The abbreviation ``idb'' stands for ``intensional database schema''
and ``edb'' stands for ``extensional database schema'' \cite{ahv_book}.%
} 

Any database instance $I$ over $\schof P$ can be given as \emph{input}
to $P$. Note, $I$ may already contain facts over $\idb P$.%
\footnote{The need for this will become clear in Section~\ref{sec:operational}.%
} Let $\rl\in P$. A \emph{valuation for} $\rl$ is a total function
$V:\vars{\rl}\to\univ$. The \emph{application} of $V$ to an atom
$R(u_{1},\ldots,u_{k})$ of $\rl$, denoted $V(R(u_{1},\ldots,u_{k}))$,
results in the \emph{fact} $R(a_{1},\ldots,a_{k})$ where for each
$i\in\{1,\ldots,k\}$ we have $a_{i}=V(u_{i})$ if $u_{i}\in\uvar$
and $a_{i}=u_{i}$ otherwise. In words: applying $V$ replaces the
variables by data values and leaves the old data values unchanged.
This is naturally extended to a set of atoms, which results in a set
of facts. Valuation $V$ is said to be \emph{satisfying for $\rl$
on $I$} if $V(\bpos{\rl})\subseteq I$ and $V(\bneg{\rl})\cap I=\emptyset$.
If so, $\rl$ is said to \emph{derive} the fact $V(\head{\rl})$.

\subsubsection{Positive and Semi-positive}

\newcommand{\tp}{T_{P}}

\label{sub:positive-and-semi-positive}

Let $P$ be a $\datalogneg$ program. We say that $P$ is \emph{positive}
if all rules of $P$ are positive. We say that $P$ is \emph{semi-positive}
if for each rule $\rl\in P$, the atoms of $\bneg{\rl}$ are over
$\edb P$. Note, positive programs are semi-positive.

We now give the semantics of a semi-positive $\datalogneg$ program
$P$ \cite{ahv_book}. First, let $\tp$ be the \emph{immediate consequence
operator} that maps each instance $J$ over $\schof P$ to the instance
$J'=J\cup A$ where $A$ is the set of facts derived by all possible
satisfying valuations for the rules of $P$ on $J$. 

Let $I$ be an instance over $\schof P$. Consider the infinite sequence
$I_{0}$, $I_{1}$, $I_{2}$, etc, inductively defined as follows:
$I_{0}=I$ and $I_{i}=\tp(I_{i-1})$ for each $i\geq1$. The \emph{output
of $P$ on input $I$}, denoted $P(I)$, is defined as $\bigcup_{j}I_{j}$;
this is the \emph{minimal fixpoint} of the $T_{P}$ operator. Note,
$I\subseteq P(I)$. When $I$ is finite, the fixpoint is finite and
can be computed in polynomial time according to data complexity~\cite{vardi_comp}.

\subsubsection{Stratified Semantics}

\label{sub:stratified-semantics}

We now recall the stratified semantics for a $\datalogneg$ program
$P$ \cite{ahv_book}. As a slight abuse of notation, here we will
treat $\idb P$ as a set of only relation names (without associated
arities). First, $P$ is called \emph{syntactically stratifiable}
if there is a function $\sigma:\idb P\to\{1,\ldots,|\idb P|\}$ such
that for each rule $\rl\in P$, having some head predicate $T$, the
following conditions are satisfied:
\begin{itemize}
\item $\sigma(R)\leq\sigma(T)$ for each $R(\bar{u})\in\proj{\bpos{\rl}}{\idb P}$;
\item $\sigma(R)<\sigma(T)$ for each $R(\bar{u})\in\proj{\bneg{\rl}}{\idb P}$.
\end{itemize}
For $R\in\idb P$, we call $\sigma(R)$ the \emph{stratum number}
of $R$. For technical convenience, we may assume that if there is
an $R\in\idb P$ with $\sigma(R)>1$ then there is an $S\in\idb P$
with $\sigma(S)=\sigma(R)-1$. Intuitively, function $\sigma$ partitions
$P$ into a sequence of semi-positive $\datalogneg$ programs $P_{1}$,
\ldots, $P_{k}$ with $k\leq|\idb P|$ such that for each $i=1,\ldots,k$,
the program $P_{i}$ contains the rules of $P$ whose head predicate
has stratum number $i$. This sequence is called a \emph{syntactic
stratification} of $P$. We can now apply the \emph{stratified semantics}
to $P$: for an input $I$ over $\schof P$, we first compute the
fixpoint $P_{1}(I)$, then the fixpoint $P_{2}(P_{1}(I))$, etc. The
\emph{output of $P$ on input $I$}, denoted $P(I)$, is defined
as $P_{k}(P_{k-1}(\ldots P_{1}(I)\ldots))$. It is well known that
the output of $P$ does not depend on the chosen syntactic stratification
(if more than one exists). Not all $\datalogneg$ programs are syntactically
stratifiable.

\subsubsection{Stable Model Semantics}

\label{sub:stable-model-semantics}

We now recall the stable model semantics for a $\datalogneg$ program
$P$ \cite{gelfond_stable,sz_nondet}. Let $I$ be an instance over
$\schof P$. Let $\rl\in P$. Let $V$ be a valuation for $\rl$ whose
image is contained in $\adom I\cup C$, where $C$ is the set of all
constants appearing in $P$. Valuation $V$ does not have to be satisfying
for $\rl$ on $I$. Together, $V$ and $\rl$ give rise to a ground
rule $\grl$, obtained from $\rl$ by replacing each $u\in\vars{\rl}$
with $V(u)$. We call $\grl$ a \emph{ground rule of $\rl$ with
respect to $I$}. Let $\gr{\rl}I$ denote the set of all ground rules
of $\rl$ with respect to $I$. The \emph{ground program of $P$
on $I$}, denoted $\gr PI$, is defined as $\bigcup_{\rl\in P}\gr{\rl}I$.
Note, if $I=\emptyset$, the set $\gr PI$ contains only rules whose
ground atoms are made with $C$, or atoms that are nullary.

Let $M$ be another instance over $\schof P$. We write $\grp MPI$
to denote the program\emph{ }obtained from $\gr PI$ as follows:
\begin{enumerate}
\item remove every rule $\grl\in\gr PI$ for which $\bneg{\grl}\cap M\neq\emptyset$;
\item remove the negative (ground) body atoms from all remaining rules.
\end{enumerate}
Note, $\grp MPI$ is a positive program. We say that $M$ is a \emph{stable
model of $P$ on input $I$} if $M$ is the output of $\grp MPI$
on input $I$. If so, the semantics of positive $\datalogneg$ programs
implies $I\subseteq M$. Not all $\datalogneg$ programs have stable
models on every input \cite{gelfond_stable}.

\subsection{Network and Distributed Databases}

\label{sub:network-and-distributed-database}

A \emph{(computer) network} is a nonempty finite set $\nw$ of \emph{nodes},
which are values in $\univ$. Intuitively, $\nw$ represents the identifiers
of compute nodes involved in a distributed system. Communication channels
(edges) are not explicitly represented because we allow a node $x$
to send a message to any node $y$, as long as $x$ knows about $y$
by means of input relations or received messages. For general distributed
or cluster computing, the delivery of messages is handled by the network
layer, which is abstracted away. But ($\datalog$) programs can also
describe the network layer itself \cite{decl_netw_cacm,hellerstein_declimp},
in which case we would restrict attention to programs where nodes
only send messages to nodes to which they are explicitly linked; these
nodes would again be provided as input.

A \emph{distributed database instance} $H$ over a network $\nw$
and a database schema $\sch$ is a function that maps every node of
$\nw$ to an ordinary finite database instance over $\sch$. This
represents how data over the same schema $\sch$ is spread over a
network.

As a small example of a distributed database instance, consider the
following instance $H$ over a network $\nw=\{x,y\}$ and a schema
$\sch=\{\rar R1,\,\rar S1\}$: $H(x)=\{R(a),\, S(b)\}$ and $H(y)=\{R(a),\,S(c)\}$.
In words: we put facts $R(a)$ and $S(b)$ at node $x$, and we put facts $R(a)$ and $S(c)$ at node $y$. 
Note that it is possible that the same fact is given to multiple nodes.

\section{Model-Based Semantics}

\label{sec:declarative}

Here we describe a class of distributed $\datalogneg$ programs that
we give a model-based semantics. First, in Section~\ref{sub:user-language},
we recall the user language $\dedalus$, that is based on $\datalogneg$ with annotations,
in which the programmer can express the functionality of the distributed
program. Next, we discuss how to assign a declarative, model-based
semantics to $\dedalus$ programs. This semantics consists of applying
the stable model semantics to the $\dedalus$ programs after they are transformed
into pure $\datalogneg$ programs, i.e., without annotations. We introduce
some auxiliary notations and symbols in Section~\ref{sub:pure-P-notations-and-relations}.
Next, in Section~\ref{sub:basic-transformation}, we give a basic
transformation of $\dedalus$ programs in order to apply the stable model
semantics. However, this basic transformation has some shortcomings,
that we iteratively correct in Sections~\ref{sub:causal-transformation} and
\ref{sub:causal-finite-transformation}.

\subsection{User Language: $\dedalus$}

\label{sub:user-language}

\newcommand{\simplebody}[1]{\mathbf{B}\{#1\}}

\newcommand{\sugind}{\bullet}

\newcommand{\sugas}[1]{\mid#1}

\newcommand{\relid}{\pred{Id}}

\newcommand{\relnode}{\pred{Node}}

Our user language for distributed $\datalogneg$ programs is $\dedalus$~\cite{dedalus,dedalus-datalog20,hellerstein_declimp},
here presented as $\datalogneg$ with annotations.%
\footnote{These annotations correspond to syntactic sugar in the previous presentations
of $\dedalus$.%
} Essentially, the language represents updatable memory for the nodes
of a network and provides a mechanism for communication between these
nodes.

\subsubsection{Syntax}

Let $\sch$ be a database schema. We write $\simplebody{\tvar v}$,
where $\tvar v$ is a tuple of variables, to denote any sequence $\beta$
of literals over database schema $\sch$, such that the variables
in $\beta$ are precisely those in the tuple $\tvar v$. Let $R(\tvar u)$
denote any atom over $\sch$. There are three types of $\dedalus$
rules over $\sch$:
\begin{itemize}
\item A \emph{deductive} rule is a normal $\datalogneg$ rule over $\sch$.
\item An \emph{inductive} rule is of the form 
\[
R(\tvar u)\sugind\gets\simplebody{\tvar u,\tvar v}.
\]

\item An \emph{asynchronous} rule is of the form 
\[
R(\tvar u)\sugas{\var y}\gets\simplebody{\tvar u,\tvar v,\var y}.
\]

\end{itemize}
For asynchronous rules, the annotation \sq{$\sugas{\var y}$}
with $\var y\in\uvar$ means that the derived head facts are transferred
(``piped'') to the addressee node represented by $\var y$. Deductive,
inductive and asynchronous rules will express respectively local computation,
updatable memory, and message sending. 
As in Section~\ref{sub:datalog-with-negation},
a $\dedalus$ rule is called \emph{safe} if all its variables occur
in at least one positive body atom. 

We already provide some intuition of how asynchronous rules operate.
There are four conceptual time points involved in the execution of
an asynchronous rule: the time when the body is evaluated; the time
when the derived fact is sent to the addressee; the time when the
fact arrives at the addressee; and, the time when the arrived fact
becomes visible at the addressee. In the model-based semantics presented
later, the first two time points coincide and the last two time points
coincide; and, there is no upper bound on the interval between these
two pairs, although it will be finite.

Now consider the following definition:

\begin{definition}A \emph{$\dedalus$ program over a schema $\sch$}
is a set of deductive, inductive and asynchronous $\dedalus$ rules
over $\sch$, such that all rules are safe, and the set of deductive
rules is syntactically stratifiable.\end{definition} 

In the current work, we will additionally assume that $\dedalus$
programs are constant-free, as is common in the theory of database
query languages, and which is not really a limitation, since constants
that are important for the program can always be indicated by unary
relations in the input.

Let $\ded$ be a $\dedalus$ program. The definitions of $\schof{\ded}$,
$\idb{\ded}$, and $\edb{\ded}$ are like for $\datalogneg$ programs.
An \emph{input} for $\ded$ is a \emph{distributed} database instance
over some network $\nw$ and the schema $\edb{\ded}$. 

\subsubsection{Semantics Sketch}

We sketch the main idea behind the semantics of a $\dedalus$ program
$\ded$. We illustrate the semantics in Section~\ref{sub:examples}.

Let $H$ be an input distributed database instance for $\ded$,
over a network $\nw$. 
The idea is that all nodes $x\in\nw$ run the same program $\ded$
and use their local input fragment $H(x)$ to do local computation
and to send messages. Conceptually, each node of $\nw$ should be
thought of as doing local computation steps, indefinitely. During
each step, a node reads the following facts: $\romI$ the local input;
$\romII$ some received message facts, generated by asynchronous rules
on other nodes or the node itself; and, $\romIII$ the facts derived
by inductive rules during the previous step on this same node. Next,
the deductive rules are applied to these available facts, to compute
a fixpoint $D$ under the stratified semantics. 

Subsequently, the asynchronous and inductive rules are fired in parallel on the deductive fixpoint $D$, trying all possible valuations in single-step
derivations (i.e., no fixpoint). The asynchronous rules send messages
to other nodes or to the same node. Messages arrive after an arbitrary
(but finite) delay, where the delay can vary for each message. The
inductive rules store facts in the memory of the local node. The effect
of an inductive derivation is only visible in the very next step;
so, if a fact is to be remembered over multiple steps, it should always
be explicitly rederived by inductive rules.

\subsubsection{Examples}
\label{sub:examples}

We consider several examples to demonstrate the three kinds of $\dedalus$ rules, and how they work together. 
These examples also illustrate the utility of $\dedalus$ when applied to some practical problems.
Here, we follow the principle that the
    output on a node $x$ consists of the facts that are eventually derived during every step of $x$.

    
\begin{example}
    \label{ex:cover}
    In this example we compute reachable vertices on graph data.
    Consider the $\dedalus$ program $\ded$ in Figure~\ref{fig:example-vertex-reach}.
    We assume the \emph{edb} relations $\rar R2$, $\rar{\pred{start}}1$,
    and $\rar{\relnode}1$. For each node, relation $R$ describes a local
    graph, and relation $\pred{start}$ provides certain starting vertices.
    In any input distributed database instance $H$ over a network $\nw$, we assume that for each node, relation $\relnode$ is initialized
    to contain all nodes of $\nw$; intuitively, $\relnode$ can be regarded
    as an address book for $\nw$.

    Now, the idea is that each node of $\nw$ will check whether all of
    its local vertices are reachable from the (distributed) start vertices.
    Communication is needed to share these start vertices, which is accomplished
    by the asynchronous rule. The receipt of a start vertex initializes
    a local relation $\rar{\pred{marked}}1$ at each node; this relation
    contains reachable vertices. The inductive rule says that all reachable
    vertices that we know during the current step, are remembered in the
    next step. This way, the effect of the communication is preserved.
    Moreover, the third rule, which is deductive, collects all local graph
    vertices reachable from the currently known reachable vertices. Note,
    the inductive rule will cause the result of this deductive computation
    to be also remembered in the next step, although this effect is not
    really needed here. The last four rules, which are deductive, check
    that all local vertices are reachable from the start vertices seen
    so far; if so, a local flag $\pred{covered}\ntup$ is derived. 

    In our semantics, we will enforce that all messages eventually arrive.
    In such a semantics, eventually a node will produce $\pred{covered}\ntup$
    during each step iff all its local vertices are reachable from the
    distributed start vertices.
    \qed
\end{example}

\begin{figure}
    \begin{framed}
    \begin{cenchop}
        $\pred{marked}(\var u)\sugas{\var y}\gets\pred{start}(\var u),\,\relnode(\var y).$\qsp
        
        $\pred{marked}(\var u)\sugind\gets\pred{\pred{marked}}(\var u).$\qsp

        $\pred{marked}(\var v)\gets\pred{marked}(\var u),\, R(\var u,\var v).$\\

        $\pred{vert}(\var u)\gets R(\var u,\var v).$\\

        $\pred{vert}(\var u)\gets R(\var v,\var u).$\\

        $\pred{missing}\ntup\gets\pred{vert}(\var u),\,\neg\pred{marked}(\var u).$\\

        $\pred{covered}\ntup\gets\neg\pred{missing}\ntup.$
    \end{cenchop}
    \end{framed}
    \caption{\label{fig:example-vertex-reach}$\dedalus$ program for Example~\ref{ex:cover}.}
\end{figure}


\begin{example}
    \label{ex:set-order}
    In this example we generate a random ordering of a set through asynchronous delivery of messages. 
    Every node generates a random ordering of a local \emph{edb} relation $\rar S1$ that represents an input set. We also assume an \emph{edb} relation $\rar{\relid}1$ that contains on each node the identifier of that node; the relation $\relid$ allows a node to send a message to itself.
    The idea is that a node sends all elements of $S$ to itself as messages, and the arbitrary arrival order is used to generate an ordering of the elements. This ordering depends on the execution, and some executions will not lead to orderings if some elements are always jointly delivered.
    
    The corresponding program is shown in Figure~\ref{fig:set-order}. We use relation $\rar M1$ to send the elements of $S$, as accomplished by the single asynchronous rule.
    The relations $\rar F1$ and $\rar N2$ represent the ordering of $S$ so far, and they are considered as the output of the program; the letters `F' and `N' stand for ``first'' and ``next'' respectively. 
    For example, a possible ordering of the set $\myset{a,b,c,d}$ could be expressed by the following facts: %
        $F(d)$, $N(d,c)$, $N(c,b)$, $N(b,a)$.
    
    Inductive rules are responsible for remembering the iteratively updated versions of $F$ and $N$. 
    The other rules are deductive, and they can conceptually be executed in the order in which they are written.
    The main technical challenge is to only update the ordering when precisely one element of $S$ arrives; otherwise, because we have no choice mechanism, we would accidentally give the same ordinal to two different elements. Checking whether we may update the ordering is accomplished through other auxiliary relations.
    We use a nullary relation $\pred{started}$ as a flag to know whether we still have to initialize relation $F$ or not.

    Note that the program keeps sending all elements of $S$ through the single asynchronous rule. 
    Alternatively, by adapting the program, we could send the elements only once by making sure the asynchronous rule is fired only once (in parallel for all elements of $S$).
    In that case, as soon as two elements are later delivered together, the ordering will not contain all elements.
    \qed
\end{example}

\begin{figure}
    \newcommand{\relused}{\pred{used}}
    \newcommand{\relnew}{\pred{new}}
    \newcommand{\reltwo}{\pred{two}}
    \newcommand{\relkeep}{\pred{keep}}
    \newcommand{\relstarted}{\pred{started}}
    \newcommand{\rellast}{\pred{last}}
    \newcommand{\relnotlast}{\pred{notlast}}
    \begin{framed}
    \begin{cenchop}        
        $M(\var u)\sugas{\var x}\gets S(\var u),\,\relid(\var x).$\qsp
        $\relused(\var u)\gets F(\var u).$\\
        $\relused(\var u)\gets N(\var u,\var v).$\\
        $\relused(\var u)\gets N(\var v,\var u).$\\
        $\relnew(\var u)\gets M(u),\,%
            \neg\relused(\var u).$\\
        $\pred{eq}(\var u, \var u)\gets S(\var u).$\\
        $\reltwo\ntup\gets\relnew(\var u),\, %
            \relnew(\var v),\, \neg\pred{eq}(\var u, \var v).$\\
        $\relkeep(\var u)\gets\relnew(\var u),\,\neg\reltwo\ntup.$\\
        $\relnotlast(\var u)\gets N(\var u,\var v).$\\
        $\rellast(\var u)\gets F(\var u),\,\neg\relnotlast(\var u).$\\
        $\rellast(\var u)\gets N(\var v, \var u),\,\neg\relnotlast(\var u)$.
        \qsp
        $\relstarted\ntup\gets F(\var u).$\\
        $F(\var u)\sugind \gets \neg\relstarted\ntup,\,\relkeep(\var u).$\\
        $N(\var u,\var v)\sugind\gets\relstarted\ntup,\,\rellast(\var u),\,\relkeep(\var v).$\\
        $F(\var u)\sugind\gets F(\var u).$\\
        $N(\var u,\var v)\sugind\gets N(\var u,\var v).$
    \end{cenchop}
    \end{framed}
    \caption{\label{fig:set-order}$\dedalus$ program
    for Example~\ref{ex:set-order}.}
\end{figure}


\begin{example}
    \label{ex:2pc}
    This example is inspired by commit protocols that were expressed in a precursor language of $\dedalus$~\cite{i-do-declare}. 
    In particular, we implement a two-phase commit protocol where agents, represented by nodes, vote either ``yes'' or ``no'' for transaction identifiers. 
    Such a protocol could be part of a bigger system, where transactions are distributed across agents and each agent may only perform the transaction locally if \emph{all} agents want to do this. 
    A single \emph{coordinator} node is responsible for combining the votes for each transaction identifier $t$: the coordinator broadcasts ``yes'' for $t$ if all votes for $t$ are ``yes'', and ``no'' otherwise. 
    Each agent stores the decision of the coordinator.
    
    Because the agents and the coordinator have different roles, we make two separate $\dedalus$ programs.%
        \footnote{In our formal definitions, all nodes execute the same $\dedalus$ program. However, it is easy to simulate two different programs by giving every node the union of both programs, but using a flag to guard the rules of each program. In this example, we can then assume that one node gets a ``coordinator'' flag as input, and the other nodes get an ``agent'' flag as input.}
    First, the agent nodes are assigned the following simple $\dedalus$ program, whose relations are explained below:
    \begin{cenchop}
    $\pred{vote}(\var t,\var x,\var v)\sugas\var y\gets \pred{myVote}(\var t,\var v),\,%
        \relid(\var x),\,\pred{coord}(\var y).$
    \qsp
    $\pred{outcome}(\var t,\var v)\sugind\gets\pred{outcome}(\var t,\var v).$
    \end{cenchop}
    Here, the \emph{edb} relations are: $\rar{\pred{myVote}}2$ that maps each transaction identifier $t$ to a local vote ``yes'' or ``no'', $\rar{\relid}1$ storing the identifier of the agent, and $\rar{\pred{coord}}1$ storing the identifier of the coordinator.
    Also, the relations $\rar{\pred{vote}}3$ and $\rar{\pred{outcome}}2$ represent respectively the outgoing votes and the final decision by the coordinator.
    
    Second, the coordinator node is assigned the $\dedalus$ program shown in Figure~\ref{fig:2pc}. The coordinator has the following \emph{edb} relations: relation $\rar T1$ containing all transaction identifiers, relations $\rar Y1$ and $\rar N1$ containing the constants ``yes'' and ``no'' respectively, and relation $\rar{\pred{agents}}1$ containing all voting agents.
    The coordinator uses an inductive rule to gradually accumulate all votes for each transaction identifier. Votes can have arbitrary delays, but in our model the delays are always finite.
    In each computation step, the deductive rules at the coordinator recompute a relation $\pred{complete}$ that contains the transaction identifiers for which all votes have been received.
    When a transaction identifier $t$ has at least one ``no'' vote, the coordinator decides ``no'' for $t$, and otherwise the coordinator decides ``yes'' for $t$. The final decision is broadcast to all agents. The coordinator adds the transactions with a decision to a log, so the decision will not be broadcast again.
    \qed
\end{example}

\begin{figure}
    \begin{framed}
    \begin{cenchop}
        $\pred{vote}(\var t, \var x, \var v)\sugind\gets\pred{vote}(\var t,\var x,\var v).$\\
        $\pred{known}(\var t,\var x)\gets\pred{vote}(\var t,\var x,\var v).$\\
        $\pred{missing}(\var t)\gets T(\var t),\, \pred{agent}(\var x),\, \neg\pred{known}(\var t,\var x).$\\
        $\pred{complete}(\var t)\gets T(\var t),\, \neg\pred{missing}(\var t).$\\      
        $\pred{decideNo}(\var t)\gets \pred{votes}(\var t,\var x,\var v),\,N(\var v).$\\
        $\pred{decideYes}(\var t)\gets \pred{complete}(\var t),\,\neg\pred{decideNo}(\var t).$
        \qsp
        $\pred{outcome}(\var t,\var v)\sugas\var y \gets \pred{decideNo}(\var t),\, \neg\pred{log}(\var t),\, N(\var v),\, \pred{agent}(\var y).$\\
        $\pred{outcome}(\var t,\var v)\sugas\var y \gets\pred{decideYes}(\var t),\,\neg\pred{log}(\var t),\,Y(\var v),\,\pred{agent}(\var y).$\\
        $\pred{log}(\var t)\sugind\gets\pred{complete}(\var t).$\\
        $\pred{log}(\var t)\sugind\gets\pred{log}(\var t)$.
    \end{cenchop}
    \end{framed}
    \caption{\label{fig:2pc}$\dedalus$ (coordinator) program
    for Example~\ref{ex:2pc}.}
\end{figure}

\subsection{Auxiliary Notations and Relations}

\label{sub:pure-P-notations-and-relations}

Let $\ded$ be a $\dedalus$ program. Let $\rar Rk\in\schof{\ded}$.
We will use facts of the form $R(x,s,a_{1},\ldots,a_{k})$ to express that
fact $R(a_{1},\ldots,a_{k})$ is present at a node $x$ during its
local step $s$, with $s\in\Nat$, after the deductive rules are executed.
We call $x$ the \emph{location specifier} and $s$ the \emph{timestamp}.
In order to represent timestamps, we assume $\Nat\subseteq\univ$. 

We write $\toloct{\schof{\ded}}$ to denote the database schema obtained
from $\schof{\ded}$ by incrementing the arity of every relation by
two. The two extra components will contain the location specifier
and timestamp.%
\footnote{The abbreviation \sq{$\mathrm{LT}$} stands for ``location specifier
and timestamp''. %
} For an instance $I$ over $\schof{\ded}$, $x\in\univ$ and $s\in\Nat$,
we write $\addlt Ixs$ to denote the facts over $\toloct{\schof{\ded}}$
that are obtained by prepending location specifier $x$ and timestamp
$s$ to every fact of $I$. Also, if $L$ is a sequence of literals
over $\schof{\ded}$, and $\var x,\var s\in\uvar$, we write $\addlt L{\var x}{\var s}$
to denote the sequence of literals over $\toloct{\schof{\ded}}$ that
is obtained by adding location specifier $\var x$ and timestamp $\var s$
to the literals in $L$ (negative literals stay negative).

We also need auxiliary relation names, that are assumed not to be
used in $\schof{\ded}$; these are listed in Table~\ref{tab:relation-names}.%
\footnote{In practice, auxiliary relations can be differentiated from those
in $\schof{\ded}$ by a namespace mechanism.%
} The concrete purpose of these relations will become clear in the
following subsections.

\begin{table}

\caption{Relation names not in $\protect\schof{\protect\ded}$.}
\label{tab:relation-names}
\begin{centering}

\begin{tabular}{p{0.4\textwidth}|p{0.4\textwidth}}
\hline 
Relation Names & Meaning\tabularnewline
\hline 
\hline 
$\relall$ & network\tabularnewline
\hline 
$\reltime$, $\timesucc$, $<$, $\neq$ & timestamps\tabularnewline
\hline 
$\before$ & happens-before relation\tabularnewline
\hline 
$\cand_{R}$, $\chosen_{R}$, $\other_{R}$, for each relation name
$R$ in $\idb{\ded}$ & messages\tabularnewline
\hline 
$\hassender$, $\issmaller$, $\hasmax$, $\rcvinf$ & only a finite number of messages arrive at each step of a node\tabularnewline
\hline 
\end{tabular}

\end{centering}
\end{table}

We define the following schema 
\[
\schtime=\{\rar{\reltime}1,\,\rar{\timesucc}2,\,\rar{<\!}2,\,\rar{\neq\!}2\}.
\]
The relations \sq{$<$} and \sq{$\neq$} will be written in infix
notation in rules. We consider only the following instance over $\schtime$:
\begin{eqnarray*}
\Itime & = & \{\reltime(s),\,\timesucc(s,s+1)\mid s\in\Nat\}\\
 &  & {}\cup\{(s<t)\mid s,t\in\Nat:\, s<t\}\\
 &  & {}\cup\{(s\neq t)\mid s,t\in\Nat:\, s\neq t\}.
\end{eqnarray*}
Intuitively, the instance $\Itime$ provides timestamps together with
relations to compare them.

\subsection{Dynamic Choice Transformation}

\label{sub:basic-transformation}

Let $\ded$ be a $\dedalus$ program. We describe the \emph{dynamic
choice transformation} to transform $\ded$ into a pure $\datalogneg$
program $\purech{\ded}$. The most technical part of the transformation
involves the use of dynamic choice to select an arrival timestamp
for each message generated by an asynchronous rule. The actual transformation
is presented first; next we give the semantics; and, lastly, we discuss
how the transformation can be improved.

\subsubsection{Transformation}

We incrementally construct $\purech{\ded}$. In particular, for each
rule in $\ded$, we specify what corresponding rule (or rules) should
be added to $\purech{\ded}$. For technical convenience, we assume
that rules of $\ded$ always contain at least one positive body atom.
This assumption allows us to more elegantly enforce that head variables
in rules of $\purech{\ded}$ also occur in at least one positive body
atom.%
\footnote{This assumption is not really a restriction, since a nullary positive
body atom is already sufficient.%
} Let $\var x,\var s,\var t,\var{t'}\in\uvar$ be distinct variables
not yet occurring in rules of $\ded$.  We write $\simplebody{\tvar v}$,
where $\tvar v$ is a tuple of variables, to denote any sequence $\beta$
of literals over $\schof{\ded}$, such that the variables in $\beta$
are precisely those in $\tvar v$. Also recall the notations and relation
names from Section~\ref{sub:pure-P-notations-and-relations}.

\paragraph*{Deductive rules}

For each deductive rule $R(\tvar u)\gets\simplebody{\tvar u,\tvar v}$
in $\ded$, we add to $\purech{\ded}$ the following rule:
\begin{equation}
R(\var x,\var s,\tvar u)\gets\addlt{\simplebody{\tvar u,\tvar v}}{\var x}{\var s}.\label{eq:pure-duc}
\end{equation}
This rule expresses that deductively derived facts at some node $x$
during step $s$ are (immediately) visible within step $s$ of $x$.
Note, all atoms in this rule are over $\toloct{\schof{\ded}}$.

\paragraph*{Inductive rules}

For each inductive rule $R(\tvar u)\sugind\gets\simplebody{\tvar u,\tvar v}$
in $\ded$, we add to $\purech{\ded}$ the following rule: 
\begin{equation}
R(\var x,\var t,\tvar u)\gets\addlt{\simplebody{\tvar u,\tvar v}}{\var x}{\var s},\,\timesucc(\var s,\var t).\label{eq:pure-ind}
\end{equation}
This rule expresses that inductively derived facts becomes visible
in the \emph{next} step of the \emph{same} node.

\paragraph*{Asynchronous rules}

\newcommand{\rulechosen}[1]{\chosen_{R}(\var x,\var s,\var y,\var t,\tvar w)\gets\cand_{R}(\var x,\var s,\var y,\var t,\tvar w),\,\neg\other_{R}(\var x,\var s,\var y,\var t,\tvar w)#1}

\newcommand{\ruleother}[1]{\other_{R}(\var x,\var s,\var y,\var t,\tvar w)\gets\cand_{R}(\var x,\var s,\var y,\var t,\tvar w),\,\chosen_{R}(\var x,\var s,\var y,\var{t'},\tvar w),\,\var t\neq\var{t'}#1}

\newcommand{\ruledeliv}[1]{R(\var y,\var t,\tvar w)\gets\chosen_{R}(\var x,\var s,\var y,\var t,\tvar w)#1}

We use facts of the form $\relall(x)$ to say that $x$ is a node
of the network at hand. We use facts of the form $\cand_{R}(x,s,y,t,\bar{a})$
to express that node $x$ at its step $s$ sends a message $R(\bar{a})$
to node $y$, and that $t$ could be the arrival timestamp of this
message at $y$.%
\footnote{Here, \sq{$\cand$} abbreviates ``candidate''.%
} Within this context, we use a fact $\chosen_{R}(x,s,y,t,\bar{a})$
to say that $t$ is the \emph{effective} arrival timestamp of this
message at $y$. Lastly, a fact $\other_{R}(x,s,y,t,\bar{a})$ means
that $t$ is \emph{not} the arrival timestamp of the message. Now,
for each asynchronous rule 
\[
R(\tvar u)\sugas{\var y}\gets\simplebody{\tvar u,\tvar v,\var y}
\]
in $\ded$, letting $\tvar w$ be a tuple of new and distinct variables
with $\len{\tvar w}=\len{\tvar u}$, we add to $\purech{\ded}$ the
following rules, for which the intuition is given below:
\begin{equation}
\cronrulecandidate .\label{eq:sz-cand}
\end{equation}
\begin{equation}
\rulechosen .\label{eq:chosen}
\end{equation}
\begin{equation}
\ruleother .\label{eq:other}
\end{equation}
\begin{equation}
\ruledeliv .\label{eq:deliv}
\end{equation}

Rule~(\ref{eq:sz-cand}) represents the messages that are sent.
It evaluates the body of the original asynchronous rule, verifies
that the addressee is within the network by using relation $\relall$,
and it generates all possible candidate arrival timestamps.

Now remains the matter of actually choosing \emph{one} arrival timestamp
amongst all these candidates. Intuitively, rule~(\ref{eq:chosen})
selects an arrival timestamp for a message with the condition that
this timestamp is not yet ignored, as expressed with relation $\other_{R}$.
Also, looking at rule~(\ref{eq:other}), a possible arrival timestamp
$t$ becomes ignored if there is already a chosen arrival timestamp
$t'$ with $t\neq t'$. Together, both rules have the effect that
exactly one arrival timestamp will be chosen under the stable model
semantics. This technical construction is due to \citeNS{sz_nondet},
who show how to express dynamic choice under the stable model semantics.

Rule~(\ref{eq:deliv}) represents the actual arrival of an $R$-message
with the chosen arrival timestamp: the data-tuple in the message becomes
part of the addressee's state for relation $R$. When the addressee
reads relation $R$, it thus transparently reads the arrived $R$-messages.

Note, if multiple asynchronous rules in $\ded$ have the same head
predicate $R$, only new $\cand_{R}$-rules have to be added because
the rules (\ref{eq:chosen})--(\ref{eq:deliv}) are general for all
$R$-messages.

Note that if there are asynchronous rules in $\ded$, program $\purech{\ded}$
is not syntactically stratifiable if a $\cand_{R}$-rule contains
a body atom that (indirectly) negatively depends on $R$.%
\footnote{Indeed, $\cand_{R}$ is used to compute $R$, but $R$ is also used
to compute $\cand_{R}$, giving a cycle through negation.%
} In that case, $\purech{\ded}$ might not even be locally stratifiable~\cite{apt_bol}.

\subsubsection{Semantics}

Now we define the semantics of $\purech{\ded}$. Let $H$ be an input
distributed database instance for $\ded$, over a network $\nw$.
Using the notations from Section~\ref{sub:pure-P-notations-and-relations},
we define $\decl H$ to be the following database instance over the
schema $\toloct{\edb{\ded}}\cup\{\rar{\relall}1\}\cup\schtime$:
\begin{eqnarray*}
\decl H & = & \{R(x,s,\bar{a})\mid x\in\nw,\, s\in\Nat,\, R(\bar{a})\in H(x)\}\\
 &  & {}\cup\{\relall(x)\mid x\in\nw\}\cup\Itime.
\end{eqnarray*}
In words: we make for each node its input facts available at all timestamps;
we provide the set of all nodes; and, $\Itime$ provides the timestamps
with comparison relations.%
\footnote{For simplicity we already include relation $<$ in this definition,
although this relation will only be used later.%
} Note, instance $\decl H$ is infinite because $\Nat$ is infinite.

The stable model semantics for $\datalogneg$ programs is reviewed
in Section~\ref{sub:stable-model-semantics}. Consider now the following
definition:

\begin{definition}For an input distributed database instance $H$
for $\ded$, we call any stable model of $\purech{\ded}$ on input
$\decl H$ a \emph{choice-model} of $\ded$ on input $H$.\end{definition}

\subsubsection{Possible Improvement}

\label{sub:choice-transform-improvement}

\newcommand{\lsnd}{\mathrm{snd}}

\newcommand{\lrcv}{\mathrm{rcv}}

We illustrate a shortcoming of the dynamic choice transformation.
Consider the $\dedalus$ program $\ded$ in Figure~\ref{fig:non-causality}.
We assume that in each input distributed database, the \emph{edb}
relation $\rar{\relid}1$ contains on each node just the identifier
of this node. This way, the node can send messages to itself. Relation
$T$ is the intended output relation of $\ded$. The idea is that
a node sends $A\ntup$ to itself continuously. When $A\ntup$ arrives,
we send $B\ntup$, but we also want to create an output fact $T\ntup$.
We only create $T\ntup$ when $B\ntup$ is absent. When $B\ntup$
is received, it is remembered by inductive rules. Now, we see that
the delivery of at least one $A\ntup$ is necessary to cause a $B\ntup$
to be sent. This creates the expectation that $T\ntup$ is always
created: at least one $A\ntup$ is delivered before any $B\ntup$.
This intuition can be formalized as \emph{causality}~\cite{attiyawelch_dcbook}
(see also Section~\ref{sub:run-happens-before}). 

However, this intuition is violated by some choice-models of $\ded$,
as we demonstrate next. Consider the input distributed database instance
$H$ over a singleton network $\{z\}$ that assigns the fact $\relid(z)$
to $z$. Now, consider the following choice-model $M$ of $\ded$
on $H$:%
\footnote{Using straightforward arguments, it can indeed be shown that $M$
is a stable model of $\purech{\ded}$ on $\decl H$.%
}
\[
M=\decl H\cup M_{A}^{\lsnd}\cup M_{A}^{\lrcv}\cup M_{B}^{\lsnd}\cup M_{B}^{\lrcv},
\]
where 
    \[
    \begin{array}{ll}
    M_A^\lsnd = %
                {} & \{\cand_A(z,s,z,t)\mid s,t\in\Nat\} \\
                & {}\cup \{\chosen_A(z,s,z,s+1)\mid s\in\Nat\} \\
                & {}\cup \{\other_A(z,s,z,t)\mid s,t\in\Nat,\, t\neq s+1\};\\   
    \\
    M_A^\lrcv = %
                {} & \{A(z,s)\mid s\in\Nat,\, s\geq 1\}; \\
    \\
    M_B^\lsnd = %
                {} & \{\cand_B(z,s,z,t)\mid s,t\in\Nat,\, s\geq 1\} \\
                & {}\cup \{\chosen_B(z,1,z,0)\} \\
                & {}\cup \{\chosen_B(z,s,z,s+1)\mid s\in\Nat,\, s\geq 2\} \\
                & {}\cup \{\other_B(z,1,z,t)\mid t\in\Nat,\, t\neq 0\} \\
                & {}\cup \{\other_B(z,s,z,t)\mid s,t\in\Nat,\, s\geq 2,\, t\neq s+1\};\\
    \\
    M_B^\lrcv = %
                {} & \{B(z,s)\mid s\in\Nat\}.
    \end{array}
    \]    
In $M_{B}^{\lsnd}$, note that one
$B$-message is sent at timestamp $1$ of $z$, and arrives at timestamp
$0$ of $z$. We immediately see that this message is peculiar: we
should not be able to send a message to arrive in the past. Because
of the stray message $B\ntup$, the fact $B\ntup$ exists at all timestamps:
it arrives at timestamp $0$ and is henceforth persisted by the inductive
rule for relation $B$; this is modeled by set $M_{B}^{\lrcv}$. Subsequently,
there are no ground rules of the form $T(z,s)\gets A(z,s)$ with $s\in\Nat$
in the ground program $\grp MCI$, where $C=\purech{\ded}$ and $I=\decl H$.

In the next subsection, we exclude such unintuitive stable models
using an extended transformation of $\dedalus$ programs.

\begin{figure}
\begin{framed}
\begin{cenchop}

$A\ntup\sugas{\var x}\gets\relid(\var x).$\\

$B\ntup\sugas{\var x}\gets A\ntup,\,\relid(\var x).$\qsp

$T\ntup\gets A\ntup,\,\neg B\ntup.$\\

$T\ntup\sugind\gets T\ntup.$\\

$B\ntup\sugind\gets B\ntup.$

\end{cenchop}
\end{framed}
\caption{$\dedalus$ program sensitive to non-causality.}

\label{fig:non-causality}

\end{figure}

\subsection{Causality Transformation}

\label{sub:causal-transformation}

Let $\ded$ be a $\dedalus$ program. In this section, we present
the causality transformation $\purecaus{\ded}$ that extends $\purech{\ded}$
to exclude the unintuitive stable models that we have encountered in the
previous subsection. We first present the new transformation, and
then we discuss how the transformation can still be improved.

\subsubsection{Transformation}

\newcommand{\rulebeforestep}[1]{\before(\var x,\var s,\var x,\var t)\gets\relall(\var x),\,\timesucc(\var s,\var t)#1}

\newcommand{\rulebeforetr}[1]{\before(\var x,\var s,\var y,\var t)\gets\before(\var x,\var s,\var z,\var u),\,\before(\var z,\var u,\var y,\var t)#1}

\newcommand{\rulecandidate}[1]{%
    \begin{array}{ll}%
        \cand_{R}(\var x,\var s,\var y,\var t,\tvar u)\gets %
            &  \addlt{\simplebody{\tvar u,\tvar v,\var y}}{\var x}{\var s},\,\relall(\var y),\,\reltime(\var t),\\%
            &  \neg\before(\var y,\var t,\var x,\var s)#1 %
    \end{array}%
 }

\newcommand{\rulebeforesend}[1]{\before(\var x,\var s,\var y,\var t)\gets\chosen_{R}(\var x,\var s,\var y,\var t,\tvar w)#1}

We define $\purecaus{\ded}$ again incrementally. First, we transform
deductive and inductive rules just as in $\purech{\ded}$. 

Next, we use facts of the form $\before(x,s,y,t)$ to express that
local step $s$ of node $x$ happens before local step $t$ of node
$y$. Regardless of $\ded$, we always add the following rules to
$\purecaus{\ded}$:
\begin{equation}
\rulebeforestep .\label{eq:before-step}
\end{equation}
\begin{equation}
\rulebeforetr .\label{eq:before-tr}
\end{equation}
Rule (\ref{eq:before-step}) expresses that on every node, a step
happens before the next step. Rule (\ref{eq:before-tr}) makes relation
$\before$ transitive.

Now, for each asynchronous rule 
\[
R(\tvar u)\sugas{\var y}\gets\simplebody{\tvar u,\tvar v,\var y}
\]
in $\ded$, we add to $\purecaus{\ded}$ the previous transformation
rules (\ref{eq:chosen}), (\ref{eq:other}) and (\ref{eq:deliv})
(omitting the $\cand_{R}$-rule), and we add the following new rules,
where $\tvar w$ is a tuple of new and distinct variables with $\len{\tvar w}=\len{\tvar u}$,
and $\var x$, $\var s$, and $\var t$ are also new variables:
\begin{equation}
\rulecandidate .\label{eq:cand}
\end{equation}
\begin{equation}
\rulebeforesend .\label{eq:before-send}
\end{equation}
Like the old rule~(\ref{eq:sz-cand}), rule (\ref{eq:cand}) represents
the messages that are sent, but now candidate arrival timestamps are
restricted by relation $\before$ to enforce causality. Intuitively,
this restriction prevents cycles from occurring in relation $\before$.
This aligns with the semantics of a real distributed system, where
the happens-before relation is a strict partial order~\cite{attiyawelch_dcbook}
(see also Section~\ref{sub:run-happens-before}).

Rule~(\ref{eq:before-send}) adds the causal restriction that the
local step of the sender happens before the arrival step of the addressee.
Together with the previously introduced rules (\ref{eq:before-step})
and (\ref{eq:before-tr}), this will make sure that when the addressee
later \emph{causally} replies to the sender, the reply --- as generated
by a rule of the form (\ref{eq:cand}) --- will arrive after this
first send-step of the sender.

\begin{remark}The new program $\purecaus{\ded}$ excludes unintuitive
models like the one in Section~\ref{sub:choice-transform-improvement}.
In the context of that particular example, it will be impossible to
exhibit a stable model of $\purecaus{\ded}$ in which $B\ntup$ is
sent to timestamp $0$. Indeed, $B\ntup$ can only be sent starting
from timestamp $1$; timestamp $0$ at $z$ (locally) happens before
timestamp $1$ at $z$; and, the negative $\before$-literal in rule
(\ref{eq:cand}) will prevent sending from timestamp $1$ at $z$
to timestamp $0$ at $z$. Also in scenarios where different nodes
$x$ and $y$ send messages to each other, when node $x$ replies
to a message of node $y$ sent at timestamp $s$ of $y$, node $x$
can not send the reply to a timestamp $t$ of $y$ with $t<s$.\end{remark}

\subsubsection{Semantics}

The semantics of the causality transformation is the same as for the
dynamic choice transformation:

\begin{definition}For an input distributed database instance $H$
for $\ded$, we call any stable model of $\purecaus{\ded}$ on input
$\decl H$ a \emph{causal model} of $\ded$ on input $H$.\end{definition}

\subsubsection{Possible Improvement}

\newcommand{\lbefore}{\mathrm{before}}

We illustrate a shortcoming of the causality transformation. Consider
the $\dedalus$ program $\ded$ in Figure~\ref{fig:message-heaping}.
We assume that in each input distributed database, the \emph{edb}
relation $\rar{\pred{contact}}1$ contains intended recipients of
messages. Relation $T$ serves as the output relation of $\ded$.
The idea is that a node sends $A\ntup$ to its recipients continuously.
When $A\ntup$ arrives, a recipient sets a local flag $\pred{first}\ntup$.
Later, when a second $A\ntup$ arrives, the recipient creates an output
fact $T\ntup$ that we remember by means of inductive rules. Intuitively,
we expect that $T\ntup$ is always created because the fact $A\ntup$
is sent infinitely often to a recipient, making this recipient witness
the arrival of $A\ntup$ at (hopefully) two distinct moments.

However, this intuition is violated by some causal models of $\ded$.
Consider the input distributed database instance $H$ over a network
$\{x,y\}$ that (only) assigns the fact $\pred{contact}(y)$ to $x$.
Now, consider the following causal model $M$ of $\ded$ on $H$:%
\footnote{Using straightforward arguments, it can be shown that $M$ is a stable
model of $\purecaus{\ded}$ on $\decl H$.%
}
\[
M=\decl H\cup M_{A}^{\lsnd}\cup M_{A}^{\lrcv}\cup M^{\lbefore},
\]
where 
    \[
    \begin{array}{ll}
    M_A^\lsnd = {} %
                & \{\cand_A(x,s,y,t)\mid s,t\in\Nat\} \\
                & {}\cup \{\chosen_A(x,s,y,0)\mid s\in\Nat\} \\
                & {}\cup \{\other_A(x,s,y,t)\mid s,t\in\Nat,\, t\neq 0\};\\
    \\
    M_A^\lrcv = {} %
                & \{A(y,0)\} \\
                & {}\cup \{\pred{first}(y,s) \mid s\in\Nat,\, s\geq 1\}; \\
    \\
    M^\lbefore = {} %
                 & \{\before(x,s,x,t) \mid s,t\in\Nat,\, s<t\};\\                 
                 & {}\cup \{\before(y,s,y,t) \mid s,t\in\Nat,\, s<t\};\\
                 & {}\cup \{\before(x,s,y,t) \mid s,t\in\Nat\}
    \end{array}
    \]    
In this causal model, all instances
of message $A\ntup$ that $x$ sends to $y$ arrive at timestamp $0$
of $y$. For this reason, node $y$ can not witness two different
arrivals of message $A\ntup$.  In practice, however, node $y$ can
not receive an \emph{infinite} number of messages during a timestamp,
and the deliveries of the $A\ntup$ messages would be spread out more
evenly in time. So, in the next subsection, we will additionally exclude
such infinite message arrivals, to obtain our final transformation
of $\dedalus$ programs.

\begin{figure}
\begin{framed}
\begin{cenchop}

$A\ntup\sugas{\var y}\gets\pred{contact}(\var y).$\qsp

$\pred{first}\ntup\sugind\gets A\ntup.$\\

$\pred{first}\ntup\sugind\gets\pred{first}\ntup.$\qsp

$T\ntup\gets\pred{first}\ntup,\, A\ntup.$\\

$T\ntup\sugind\gets T\ntup.$

\end{cenchop}
\end{framed}

\caption{$\dedalus$ program sensitive to infinite message grouping.}

\label{fig:message-heaping}

\end{figure}

\subsection{Causality-Finiteness Transformation}

\label{sub:causal-finite-transformation}

Let $\ded$ be a $\dedalus$ program. As seen in the previous subsection,
program $\purecaus{\ded}$ allows an infinite number of messages to
arrive at any step of a node. This does not happen in any real-world
distributed system; indeed, no node has to process an infinite number
of messages at any given moment. We consider this to be an additional
restriction that must be explicitly enforced. To this purpose, we
present in this section the causality-finiteness transformation $\pure{\ded}$
that extends $\purecaus{\ded}$.

We will approach this problem as follows. Suppose there are an infinite
number of messages that arrive at some node $y$ during its step $t$.
Since in a network there are only a finite number of nodes and a node
can only send a finite number of messages during each step (the input
domain is finite), there must be at least one node $x$ that sends
messages to step $t$ of $y$ during an infinite number of steps of
$x$. Hence there is no maximum value amongst the corresponding send-timestamps
of $x$. Thus, in order to prevent the arrival of an infinite number
of messages at step $t$ of $y$, it will be sufficient to demand
that there always \emph{is} such a maximum send-timestamp for every
sender. Below, we will implement this strategy with some concrete
rules in $\pure{\ded}$.

\subsubsection{Transformation}

\newcommand{\rulehassender}[1]{\hassender(\var y,\var t,\var x,\var s)\gets\chosen_{R}(\var x,\var s,\var y,\var t,\tvar w),\,\neg\rcvinf(\var y,\var t)#1}

\newcommand{\ruleissmaller}[1]{\begin{array}{ll}%
    \issmaller(\var y,\var t,\var x,\var s)\gets  %
            &  \hassender(\var y,\var t,\var x,\var s),\,\hassender(\var y,\var t,\var x,\var{s'}),\\%
            &  \var s<\var{s'}#1 %
    \end{array}%
}

\newcommand{\rulehasmax}[1]{\hasmax(\var y,\var t,\var x)\gets\hassender(\var y,\var t,\var x,\var s),\,\neg\issmaller(\var y,\var t,\var x,\var s)#1}

\newcommand{\rulercvinf}[1]{\rcvinf(\var y,\var t)\gets\hassender(\var y,\var t,\var x,\var s),\,\neg\hasmax(\var y,\var t,\var x)#1}

We define $\pure{\ded}$ as $\purecaus{\ded}$ extended as follows.
The additional rules can be thought of as being relative to an addressee
and a step of this addressee, represented by the variables $\var y$
and $\var t$ respectively. 

We use a fact $\rcvinf(y,t)$ to express that node $y$ receives an
infinite number of messages during its step $t$. First, we add the
following rule to $\pure{\ded}$ for \emph{each} relation $\chosen_{R}$
that results from the transformation of asynchronous rules in $\purecaus{\ded}$,
where $\var x$, $\var s$, $\var y$, and $\var t$ are variables
and $\tvar w$ is a tuple of distinct variables disjoint from the
previous ones with $\len{\tvar w}$ the arity of relation $R$ in
$\schof{\ded}$: 
\begin{equation}
\rulehassender .\label{eq:has-sender}
\end{equation}
This rule intuitively means that as long as addressee $\var y$ has
not received an infinite number of messages during its step $\var t$,
we register the senders and their send-timestamps. 

Recall the auxiliary relations defined in Section~\ref{sub:pure-P-notations-and-relations}.
Next, we add to $\pure{\ded}$ the following rules, for which the
intuition is provided below:
\begin{equation}
\ruleissmaller .\label{eq:is-smaller}
\end{equation}
\begin{equation}
\rulehasmax .\label{eq:has-max}
\end{equation}
\begin{equation}
\rulercvinf .\label{eq:rcv-inf}
\end{equation}
Rule~(\ref{eq:is-smaller}) checks for each sender and each of its
send-timestamps whether there is a later send-timestamp of that same
sender. Rule~(\ref{eq:has-max}) tries to find a maximum send-timestamp.
Finally, rule~(\ref{eq:rcv-inf}) derives a $\rcvinf$-fact if no
maximum send-timestamp was found for at least one sender.

We will show in Section~\ref{sub:partial-order-M} that in any stable
model, the above rules make sure that every node receives only a finite
number of messages at every step.

\subsubsection{Semantics}

The semantics of the causality-finiteness transformation is again
the same as for the dynamic choice transformation and the causality
transformation:

\begin{definition}For an input distributed database instance $H$
for $\ded$, we call any stable model of $\pure{\ded}$ on input $\decl H$
a \emph{causal-finite model} of $\ded$ on input $H$.\end{definition}We
will refer to a causal-finite model also simply as \emph{model}.

\section{Correctness}

\label{sec:operational}

In Section~\ref{sec:declarative}, we have described the computation
of a distributed $\datalogneg$ program by means of stable models.
By using suitable rules, we have excluded some unintuitive stable
models. But at this point we are still not sure whether the remaining
stable models really correspond to the execution of a distributed
system. We fill that gap in this section: we show that each remaining
stable model corresponds to an execution of the distributed $\datalogneg$
program under an operational semantics, and vice versa. We call such
an execution a \emph{run}, and we will only be concerned with so-called
\emph{fair} runs, where each node is made active infinitely often
and all sent messages are eventually delivered. 

We extract from each run $\run$ a \emph{trace}, denoted $\trace{\run}$,
which is a set of facts that shows in detail what each node computes
during each step. We will make this concrete in the following subsections.
But we can already state our main result, as follows:

\begin{theorem}\label{theorem:main}Let $\ded$ be a $\dedalus$
program. For each input distributed database instance $H$ for $\ded$,

\begin{enumerate}

\item[\emph{(i)}]for every fair run $\run$ of $\ded$ there is a
model $M$ of $\ded$ such that $\trace{\run}=\proj M{\toloct{\schof{\ded}}}$,
and

\item[\emph{(ii)}]for every model $M$ of $\ded$ there is a fair
run $\run$ of $\ded$ such that $\trace{\run}=\proj M{\toloct{\schof{\ded}}}$.

\end{enumerate}\qed\end{theorem}

First, Section~\ref{sub:operational-semantics} formalizes runs and
traces of runs. The proof of item \emph{(i)} of the theorem is described
in Section \ref{sub:run-to-model}. The proof of item \emph{(ii)},
which is the most difficult, is described in Section \ref{sub:model-to-run}.
We only describe the crucial reasoning steps of the proofs; the intricate
technical details can be found in the online appendix to the paper.

\subsection{Operational Semantics}

\label{sub:operational-semantics}

In this section, we give an operational semantics for $\dedalus$
that is in line with earlier formal work on declarative networking
\cite{deutsch_network,declnetw_opsem,grumbach_netlog,rtdn_pods,webdamlog}. 

Let $\ded$ be a $\dedalus$ program, and let $H$ be an input distributed
database instance for $\ded$, over a network $\nw$. The essence
of the operational semantics is as follows. Every node of $\nw$ runs
program $\ded$, and a node has access only to its own local state
and any received messages. The nodes are made active one by one in
some arbitrary order, and this continues an infinite number of times.
During each active moment of a node $x$, called a \emph{local (computation)
step}, node $x$ receives message facts and applies its deductive,
inductive and asynchronous rules. Concretely, the deductive rules,
forming a stratified $\datalogneg$ subprogram, are applied to the
incoming messages and the previous state of $x$. Next, the inductive
rules are applied to the output of the deductive subprogram, and these
allow $x$ to store facts in its memory: these facts become visible
in the next local step of $x$. Finally, the asynchronous rules are
also applied to the output of the deductive subprogram, and these
allow $x$ to send facts to the other nodes or to itself. These facts
become visible at the addressee after some arbitrary delay, which
represents asynchronous communication, as occurs for instance on the
Internet. We assume that all messages are eventually delivered (and
are thus never lost). We will refer to local steps simply as ``steps''.

We make the above sketch more concrete in the next subsections.

\subsubsection{Configurations}

\label{sub:configurations}

Let $\ded$, $H$, and $\nw$ be as above. A configuration describes
the network at a certain point in its evolution. Formally, a \emph{configuration}
of $\ded$ on $H$ is a pair $\cnf=(\cnfs,\cnfb)$ where 
\begin{itemize}
\item $\cnfs$ is a function mapping each node of $\nw$ to an instance
over $\schof{\ded}$; and,
\item $\cnfb$ is a function mapping each node of $\nw$ to a set of pairs
of the form $\pair i{\fc}$, where $i\in\Nat$ and $\fc$ is a fact
over $\idb{\ded}$.
\end{itemize}
We call $\cnfs$ and $\cnfb$ the \emph{state} and \emph{(message)
buffer} respectively. The state says for each node what facts it
has stored in its memory, and the message buffer $\cnfb$ says for
each node what messages have been sent to it but that are not yet
received. The reason for having numbers $i$, called \emph{send-tags},
attached to facts in the image of $\cnfb$ is merely a technical convenience:
these numbers help separate multiple instances of the same fact when
it is sent at different moments (to the same addressee), and these
send-tags will not be visible to the $\dedalus$ program. For example,
if the buffer of a node $x$ simultaneously contains pairs $\pair 3{\fc}$
and $\pair 7{\fc}$, this means that $\fc$ was sent to $x$ during
the operational network transitions with indices $3$ and $7$, and
that both particular instances of $\fc$ are not yet delivered to
$x$. This will become more concrete in Section~\ref{sub:transitions-and-runs}.

The \emph{start configuration of $\ded$ on input $H$}, denoted
$\cnfstart{\ded}H$, is the configuration $\cnf=(\cnfs,\cnfb)$ defined
by $\cnfs(x)=H(x)$ and $\cnfb(x)=\emptyset$ for each $x\in\nw$.
In words: for every node, the state is initialized with its local
input fragment in $H$, and there are no sent messages.

\subsubsection{Subprograms}

\label{sub:subprograms}

We look at the operations that are executed locally during each step
of a node. We have mentioned that the three types of $\dedalus$ rules
each have their own purpose in the operational semantics. For this
reason, we split the program $\ded$ into three subprograms, that
contain respectively the deductive, inductive and asynchronous rules.
In Section~\ref{sub:transitions-and-runs}, we describe how these
subprograms are used in the operational semantics.
\begin{itemize}
\item First, we define $\deduc{\ded}$ to be the $\datalogneg$ program
consisting of precisely all deductive rules of $\ded$. 
\item Secondly, we define $\induc{\ded}$ to be the $\datalogneg$ program
consisting of all inductive rules of $\ded$ after the annotation
\sq{$\sugind$} in their head is removed. 
\item Thirdly, we define $\async{\ded}$ to be the $\datalogneg$ program
consisting of precisely all rules
\[
T(\var y,\tvar u)\gets\simplebody{\tvar u,\var y}
\]
where 
\[
T(\tvar u)\sugas{\var y}\gets\simplebody{\tvar u,\var y}
\]
is an asynchronous rule of $\ded$. So, we basically put the variable
$\var y$ as the first component in the (extended) head atom. The
intuition for the generated head facts is that the first component
will represent the addressee. 
\end{itemize}
Note that the programs $\deduc{\ded}$, $\induc{\ded}$ and $\async{\ded}$
are just $\datalogneg$ programs over the schema $\schof{\ded}$,
or a subschema thereof. Moreover, $\deduc{\ded}$ is syntactically
stratifiable because the deductive rules in every $\dedalus$ program
must be syntactically stratifiable. It is possible however that $\induc{\ded}$
and $\async{\ded}$ are not syntactically stratifiable. Now we define
the semantics of each of these three subprograms.

Let $I$ be a database instance over $\schof{\ded}$. During each
step of a node, the intuition of the deductive rules is that they
``complete'' the available facts by adding all new facts that can
be logically derived from them. This calls for a fixpoint semantics,
and for this reason, we define the \emph{output of $\deduc{\ded}$
on input $I$}, denoted as $\deduc{\ded}(I)$, to be given by the
stratified semantics. This implies $I\subseteq\deduc{\ded}(I)$. Importantly,
$I$ is allowed to contain facts over $\idb{\ded}$, and the intuition
is that these facts were derived during a previous step (by inductive
rules) or received as messages (as sent by asynchronous rules). This
will become more explicit in Section~\ref{sub:transitions-and-runs}.

During each step of a node, the intuition behind the inductive rules
is that they store facts in the memory of the node, and these stored
facts will become visible during the next step. There is no notion
of a fixpoint here because facts that will become visible in the next
step are not available in the current step to derive more facts. For
this reason, we define the \emph{output of $\induc{\ded}$ on input
$I$} to be the set of facts derived by the rules of $\induc{\ded}$
for all possible satisfying valuations in $I$, in just one derivation
step. This output is denoted as $\induc{\ded}\mstep I$.

During each step of a node, the intuition behind the asynchronous
rules is that they generate message facts that are to be sent around
the network. The \emph{output for $\async{\ded}$ on input $I$}
is defined in the same way as for $\induc{\ded}$, except that we
now use the rules of $\async{\ded}$ instead of $\induc{\ded}$. This
output is denoted as $\async{\ded}\mstep I$. The intuition for not
requiring a fixpoint for $\async{\ded}$ is that a message fact will
arrive at another node, or at a later step of the sender node, and
can therefore not be read during sending.

Regarding data complexity \cite{vardi_comp}, for each subprogram
the output can be computed in $\ptime$ with respect to the size of
its input.

\subsubsection{Transitions and Runs}

\label{sub:transitions-and-runs}

Transitions formalize how to go from one configuration to another.
Here we use the subprograms of $\ded$. Transitions are chained to
form a \emph{run}. Regarding notation, for a set $m$ of pairs of
the form $\pair i{\fc}$, we define $\untag m=\{\fc\mid\exists i\in\Nat:\,\pair i{\fc}\in m\}$.

A \emph{transition with send-tag $i\in\Nat$} is a five-tuple $(\cnf_{a},x,m,i,\cnf_{b})$
such that $\cnf_{a}=(\cnfs_{a},\cnfb_{a})$ and $\cnf_{b}=(\cnfs_{b},\cnfb_{b})$
are configurations of $\ded$ on input $H$, $x\in\nw$, $m\subseteq\cnfb_{a}(x)$,
and, letting
\[
 \begin{array}{ll}
 & I=\cnfs_{a}(x)\cup\untag m,\\
 & D=\deduc{\ded}(I),\\
 & \sendto iy=\{\pair i{\, R(\bar{a})}\mid R(y,\bar{a})\in\async{\ded}(D)\}\text{ for each }y\in\nw,
 \end{array}
\]
for $x$ and each $y\in\nw\setminus\{x\}$ we have 
\[
 \begin{array}{ll}
    & \cnfs_{b}(x)=H(x)\cup\induc{\ded}(D),\\
    & \cnfb_{b}(x)=(\cnfb_{a}(x)\setminus m)\cup\sendto ix,
 \end{array}
 \qquad
 \begin{array}{ll}
    & \cnfs_{b}(y)=\cnfs_{a}(y),\\
    & \cnfb_{b}(y)=\cnfb_{a}(y)\cup\sendto iy.
 \end{array}
\]

We call $\cnf_{a}$ and $\cnf_{b}$ respectively the \emph{source}
and \emph{target} configuration, and say this transition is \emph{of}
the \emph{active} node $x$. Intuitively, the transition expresses
that $x$ reads its old state together with the received facts in
$\untag m$ (thus without the tags), and describes the subsequent
computation: subprogram $\deduc{\ded}$ completes the available information;
the new state of $x$ consists of the input facts of $x$ united with
all facts derived by subprogram $\induc{\ded}$; and, subprogram $\async{\ded}$
generates \emph{messages}, whose first component indicates the addressee.%
\footnote{Note, input facts are preserved by the transition. This aligns with
the design of $\dedalus$, where we do not allow facts to be retracted;
only negation as failure is permitted.%
} Note, $\induc{\ded}$ and $\async{\ded}$ do not influence each other,
and can be thought of as being executed in parallel. Also, for each
$y\in\nw$, the set $\sendto iy$ contains all messages addressed
to $y$, with send-tag $i$ attached. Messages with an addressee outside
the network are ignored. This way of defining local computation closely
corresponds to that of the language $\webdamlog$ \cite{webdamlog}.
If $m=\emptyset$, we call the transition a \emph{heartbeat}.

A \emph{run $\run$ of $\ded$ on input $H$} is an infinite sequence
of transitions, such that \emph{(i)} the source configuration of
the first transition is $\cnfstart{\ded}H$, \emph{(ii)} the target
configuration of each transition is the source configuration of the
next transition, and \emph{(iii)} the transition at ordinal $i$
of the sequence uses send-tag $i$. Ordinals start at $0$ for technical
convenience. The resulting transition system is highly non-deterministic
because in each transition we can choose the active node and also
what messages to deliver; the latter choice is represented by the
set $m$ from above.

\begin{remark}[Parallel transitions]Transitions as defined here can
simulate \emph{parallel} transitions in which multiple nodes are
active at the same time and receive messages from their respective
buffers. Indeed, if we would have multiple nodes active during a parallel
transition, they would receive messages from their buffers in isolation,
and this can be represented by a chain of transitions in which these
nodes receive one after the other precisely the messages that they
received in the parallel transition. For this reason, we limit our
attention to transitions with single active nodes.\end{remark}

\subsubsection{Fairness and Arrival Function}

\label{sub:paper-fairness-and-arrival}

In the literature on process models it is customary to require certain
fairness conditions on the execution of a system, for instance to
exclude some extreme situations that are expected not to happen in
reality~\cite{francez_fairness,apt_fairness,lamport_fairness}. 

Let $\run$ be a run of $\ded$ on $H$. For every transition $i\in\Nat$,
let $\cnf_{i}=(\cnfs_{i},\cnfb_{i})$ denote the source configuration
of transition $i$. Now, $\run$ is called \emph{fair} if:
\begin{itemize}
\item every node is the active node in an infinite number of transitions
of $\run$; and,
\item for every transition $i\in\Nat$, for every $y\in\nw$, for every
pair $\pair j{\fc}\in\cnfb_{i}(y)$, there is a transition $k$ with
$i\leq k$ in which $\pair j{\fc}$ is delivered to $y$.
\end{itemize}
Intuitively, the fairness conditions disallow starvation: every node
does an infinite number of local computation steps and every sent
message is eventually delivered. We consider only fair runs in this
paper. Note, a fair run exists for every input because heartbeats
remain possible even when there are no messages to deliver.

In the second condition about message deliveries, it is possible that
$k=i$, and in that case $\pair j{\fc}$ is delivered in the transition
immediately following configuration $\cnf_{i}$. Because the pair
$\pair j{\fc}$ can be in the message buffer of multiple nodes, this
$k$ is not unique for the pair $\pair j{\fc}$ by itself. But, when
we also consider the addressee $y$, it follows from the operational
semantics that this $k$ is unique for the triple $(j,y,\fc)$. 

This reasoning gives rise to a function $\arr$, called the \emph{arrival
function for $\run$}, that is defined as follows: for every transition
$i$, for every node $y$, for every message $\fc$ sent to addressee
$y$ during $i$, the function $\arr$ maps $(i,y,\fc)$ to the transition
ordinal $k$ in which $\pair i{\fc}$ is delivered to $y$. We always
have $\arr(i,y,\fc)>i$. Indeed, the delivery of a message can only
happen after it was sent. So, when the delivery of one message causes
another to be sent, then the second one is delivered in a later transition.
This is related to the topic of causality that we have introduced
in Section~\ref{sec:declarative}. This topic will also be further
discussed in Sections~\ref{sub:run-to-model} and \ref{sub:model-to-run}.

\subsubsection{Timestamps and Trace}

\label{sub:run-timestamps-and-trace}

For each transition $i$ of a run, we define the \emph{timestamp}
of the active node $x$ during $i$ to be the number of transitions
of $x$ that come strictly before $i$. This can be thought of as
the \emph{local} (zero-based) clock of $x$ during $i$, and is
denoted $\locR i$. For example, suppose we have the following sequence
of active nodes: $x$, $y$, $y$, $x$, $x$, etc. If we would write
the timestamps next to the nodes, we get this sequence: $(x,0)$,
$(y,0)$, $(y,1)$, $(x,1)$, $(x,2)$, etc.

As a counterpart to function $\locR{\cdot}$, for each $(x,s)\in\nwnat$
we define $\globR{x,s}$ to be the transition ordinal $i$ of $\run$
such that $x$ is the active node in transition $i$ and $\locR i=s$.
In words: we find the transition in which node $x$ does its local
computation step with timestamp $s$. It follows from the definition
of $\locR{\cdot}$ that $\globR{x,s}$ is uniquely defined.

Let $\run$ be a run of $\ded$ on input $H$. Recall that $H$ is
over network $\nw$. We now capture the computed data during $\run$
as a set of facts that we call the \emph{trace}. For each transition
$i\in\Nat$, let $x_{i}$ denote the active node, and let $D_{i}$
denote the output of subprogram $\deduc{\ded}$ during $i$. The operational
semantics implies that $D_{i}$ consists of $\romI$ the input $\mathit{edb}$-facts
at $x_{i}$; $\romII$ the inductively derived facts during the previous
step of $x_{i}$ (if $\locR i\geq1$); $\romIII$ the messages delivered
during transition $i$; and, $\romIV$ all facts deductively derived
from the previous ones. So, intuitively, $D_{i}$ contains \emph{all}
local facts over $\schof{\ded}$ that $x_{i}$ has during transition
$i$. 

Recall the notations of Section~\ref{sub:pure-P-notations-and-relations}.
Now, the trace of $\run$ is the following instance over $\toloct{\schof{\ded}}$:
\[
\trace{\run}=\bigcup_{i\in\Nat}\addlt{D_{i}}{x_{i}}{\,\locR i}.
\]
The trace shows in detail what happens in the run, in terms of what
facts are available on the nodes during which of their steps.

\subsection{Run to Model}

\label{sub:run-to-model}

Let $\ded$ be a $\dedalus$ program and let $H$ be an input distributed
database instance for $\ded$, over a network $\nw$. Let $\run$
be a fair run of $\ded$ on input $H$. We show there is a model $M$
of $\ded$ on $H$ such that $\trace{\run}=\proj M{\toloct{\schof{\ded}}}$.
The main idea is that we translate the transitions of $\run$ to facts
over the schema of $\pure{\ded}$. 

First, in Section~\ref{sub:run-happens-before}, we extract the happens-before
relation on nodes and timestamps from $\run$. Next, in Section~\ref{sub:run-to-model-construction},
we define the desired model $M$.

\subsubsection{Happens-before Relation}

\label{sub:run-happens-before}

In the operational semantics, we order the actions of the nodes on
a fine-grained global time axis, by ordering the transitions in the
runs. By contrast, we now define a partial order on $\nwnat$, saying
which steps of nodes must have come before which steps of (other)
nodes, without referring to the global ordering imposed by transitions. 

First, we extract from $\run$ the message sending and receiving events.
Formally, we define $\mesgev{\run}$ to be the set of all tuples $(x,s,y,t,\fc)$,
with $\fc$ a fact, and denoting $i=\globR{x,s}$ and $j=\globR{y,t}$,
such that $\arr(i,y,\fc)=j$, i.e., node $x$ during step $s$ sends
message $\fc$ to $y$ that arrives at the step $t$ of $y$, with
possibly $x=y$. In words: $\mesgev{\run}$ contains the direct relationships
between local steps of nodes that arise through message sending. 

From $\run$ we can now extract the \emph{happens-before} relation
\cite{attiyawelch_dcbook} on the set $\nwnat$, which is defined
as the smallest relation $\caus$ on $\nwnat$ that satisfies the
following three conditions:
\begin{itemize}
\item for each $(x,s)\in\nwnat$, we have $(x,s)\caus(x,s+1)$;
\item $(x,s)\caus(y,t)$ whenever for some fact $\fc$ we have $(x,s,y,t,\fc)\in\mesgev{\run}$;
\item $\caus$ is transitive, i.e., $(x,s)\caus(z,u)\caus(y,t)$ implies
$(x,s)\caus(y,t)$.
\end{itemize}
We call these three cases respectively \emph{local} edges, \emph{message}
edges and \emph{transitive} edges. Naturally, the first two cases
express a direct relationship, whereas the third case is more indirect.

Note, if two runs on the same input have the same happens-before relation,
they do not necessarily have the same trace. This is because relation
$\caus$ does not talk about the specific messages that arrive at
the nodes.

We will now show that $\caus$ is a strict partial order. Consider
first the following property:

\begin{lemma}\label{lem:caus-global-order}For every run $\run$,
for each $(x,s)\in\nwnat$ and $(y,t)\in\nwnat$, if $(x,s)\caus(y,t)$
then $\globR{x,s}<\globR{y,t}$.\end{lemma}

\begin{proof}We can consider a path from $(x,s)$ to $(y,t)$ in
$\caus$. We can substitute each transitive edge in this path with
a subpath of non-transitive edges. This results in a path of only
non-transitive edges: 
\[
(x_{1},s_{1})\caus(x_{2},s_{2})\caus\ldots\caus(x_{n},s_{n}),
\]
where $n\geq2$, $(x_{1},s_{1})=(x,s)$ and $(x_{n},s_{n})=(y,t)$.
Because there are no transitive edges, for each $i\in\{1,\ldots,n-1\}$,
the edge $(x_{i},s_{i})\caus(x_{i+1},s_{i+1})$ falls into one of
the following two cases:
\begin{itemize}
\item $x_{i}=x_{i+1}$ and $s_{i+1}=s_{i}+1$ (local edge);
\item $x_{i}$ during step $s_{i}$ sends a message to $x_{i+1}$ that arrives
in step $s_{i+1}$ of $x_{i+1}$ (message edge).
\end{itemize}
In the first case, it follows from the definition of $\locR{\cdot}$
that 
\[
\globR{x_{i},s_{i}}<\globR{x_{i+1},s_{i+1}}.
\]
For the second case, by our operational semantics, every message is
always delivered in a later transition than the one in which it was
sent. So, again we have 
\[
\globR{x_{i},s_{i}}<\globR{x_{i+1},s_{i+1}}.
\]
Since this property holds for all the above edges, by transitivity
we thus have $\globR{x,s}<\globR{y,t}$, as desired.\end{proof}

\begin{corollary}For every run $\run$, the relation $\caus$ is
a strict partial order on $\nwnat$.\end{corollary}

\begin{proof}From its definition, we immediately have that $\caus$
is transitive. Secondly, irreflexivity for $\caus$ follows from Lemma~\ref{lem:caus-global-order}.\end{proof}

\subsubsection{Definition of $M$}

\label{sub:run-to-model-construction}

Now we define the model $M$: 
\[
M=\decl H\cup\bigcup_{i\in\Nat}\slice i,
\]
where $\slice i$ for each $i\in\Nat$ is an instance over the schema
of $\pure{\ded}$ that describes transition $i$ of $\run$.%
\footnote{Note, $M$ must include the input $\decl H$ by definition of stable
model (see Section~\ref{sub:stable-model-semantics}).%
} Let $i\in\Nat$. We define $\slice i$ as
\[
\slice i=\slicecaus i\cup\slicefin i\cup\sliceduc i\cup\slicesnd i,
\]
where each of these sets focuses on different aspects of transition
$i$, and they are defined next. Regarding notation, let $\caus$
be the happens-before relation as defined in the preceding subsection;
let $\locR{\cdot}$, $\globR{\cdot}$, and $\arr$ be as defined in
Section~\ref{sub:operational-semantics}; let $x_{i}$ denote the
active node of transition $i$; and, let us abbreviate $s_{i}=\locR i$.

\paragraph*{Causality}

We define $\slicecaus i$ to consist of all facts $\before(x,s,x_{i},s_{i})$
for which $(x,s)\in\nwnat$ and $(x,s)\caus(x_{i},s_{i})$. Intuitively,
$\slicecaus i$ represents the joint result of rules (\ref{eq:before-step}),
(\ref{eq:before-tr}), and (\ref{eq:before-send}), corresponding
to respectively the local edges, transitive edges, and message edges
of $\caus$.

\paragraph*{Finite Messages}

We define $\slicefin i$ to represent that only a finite number of
messages are delivered in transition $i$, thus at step $s_{i}$ of
node $x_{i}$. We proceed as follows. First, let $\senders i$ be
the set of all pairs $(x,s)\in\nwnat$ such that, denoting $j=\globR{x,s}$,
for some fact $\fc$ we have $\arr(j,x_{i},\fc)=i$, i.e., the node
$x$ during its step $s$ sends a message to $x_{i}$ with arrival
timestamp $s_{i}$. It follows from the operational semantics that
for each $(x,s)\in\senders i$ we have $\globR{x,s}<i$. Now, we define
$\slicefin i$ to consist of the following facts:
\begin{itemize}
\item the fact $\hassender(x_{i},s_{i},x,s)$ for each $(x,s)\in\senders i$,
representing the result of rule (\ref{eq:has-sender});
\item the fact $\issmaller(x_{i},s_{i},x,s)$ for each $(x,s)\in\senders i$
and $(x,s')\in\senders i$ with $s<s'$, representing the result of
rule (\ref{eq:is-smaller}); and,
\item the fact $\hasmax(x_{i},s_{i},x)$ for each sender-node $x$ mentioned
in $\senders i$, representing the result of rule (\ref{eq:has-max}).
\end{itemize}
We know that in $\run$ only a finite number of messages arrive at
step $s_{i}$ of $x_{i}$. Hence, we add no fact $\rcvinf(x_{i},s_{i})$
to $\slicefin i$.  This also explains why the specification of the
$\hasmax$-facts above is relatively simple: there is always a maximum
send-timestamp for each sender-node.

\paragraph*{Deductive}

Let $D_{i}$ denote the output of subprogram $\deduc{\ded}$ during
transition $i$. We define $\sliceduc i$ to consist of the facts
$\addlt{D_{i}}{x_{i}}{s_{i}}$. Intuitively, $\sliceduc i$ represents
all facts over $\schof{\ded}$ that are available at $x_{i}$ during
step $s_{i}$, i.e., the joint result of rules in $\pure{\ded}$ of
the form (\ref{eq:pure-duc}), (\ref{eq:pure-ind}) and (\ref{eq:deliv}).

\paragraph*{Sending}

We define $\slicesnd i$ to represent the sending of messages during
transition $i$. We proceed as follows. Let $\mesg i$ denote the
output of subprogram $\async{\ded}$ during transition $i$, restricted
to the facts having their addressee-component in the network. Now,
we define $\slicesnd i$ to consist of the following facts:
\begin{itemize}
\item all facts $\cand_{R}(x_{i},s_{i},y,t,\bar{a})$ for which $R(y,\bar{a})\in\mesg i$
and $t\in\Nat$ such that $(y,t)\not\caus(x_{i},s_{i})$, representing
the result of rule~(\ref{eq:cand});
\item all facts $\chosen_{R}(x_{i},s_{i},y,t,\bar{a})$ for which $R(y,\bar{a})\in\mesg i$
and $t=\locR j$ with $j=\arr(i,\, y,\, R(\bar{a}))$, representing
the result of rule~(\ref{eq:chosen}); and,
\item all facts $\other_{R}(x_{i},s_{i},y,u,\bar{a})$ for which $R(y,\bar{a})\in\mesg i$,
$u\in\Nat$, $(y,u)\not\caus(x_{i},s_{i})$ and $u\neq\locR j$ with
$j=\arr(i,\, y,\, R(\bar{a}))$, representing the result of rule~(\ref{eq:other}).
\end{itemize}

\paragraph*{Conclusion}

We can show that $M$ is indeed a model
of $\ded$ on input $H$; this proof can be found in \ref{app:proof-dir-1} of the online appendix to the paper.
By construction of $M$, we have, as desired:
\[
\proj M{\toloct{\schof{\ded}}}=\bigcup_{i\in\Nat}\sliceduc i=\bigcup_{i\in\Nat}\addlt{D_{i}}{x_{i}}{s_{i}}=\trace{\run}.
\]

\subsection{Model to Run}

\label{sub:model-to-run}

Let $\ded$ be a $\dedalus$ program and let $H$ be an input distributed
database instance for $\ded$, over some network $\nw$. Let $M$
be a model of $\ded$ on input $H$. We show there is a fair run $\run$
of $\ded$ on input $H$ such that $\trace{\run}=\proj M{\toloct{\schof{\ded}}}$. 

The direction shown in Section~\ref{sub:run-to-model} is perhaps
the most intuitive direction because we only have to show that a concrete
set of facts is actually a stable model. In this section we do not
yet understand what $M$ can contain. So, a first important step is
to show that $M$ has some desirable properties which allow us to
construct a run from it.

Using the notation from Section~\ref{sub:stable-model-semantics},
let $\grded$ abbreviate the ground program $\grp MCI$ where $C=\pure{\ded}$
and $I=\decl H$. By definition of $M$ as a stable model, we have
$M=\grded(I)$. 

First, it is important to know that in $M$ we find location specifiers
where we expect location specifiers and we find timestamps where we
expect timestamps. Formally, we call $M$ \emph{well-formed} if:
\begin{itemize}
\item for each $R(x,s,\bar{a})\in\proj M{\toloct{\schof{\ded}}}$ we have
$x\in\nw$ and $s\in\Nat$;
\item for each $\before(x,s,y,t)\in M$, we have $x,y\in\nw$ and $s,t\in\Nat$;
\item for each fact $\cand_{R}(x,s,y,t,\bar{a})$, $\chosen_{R}(x,s,y,t,\bar{a})$
and $\other_{R}(x,s,y,t,\bar{a})$ in $M$, we have $x,y\in\nw$ and
$s,t\in\Nat$;
\item for each fact $\hassender(x,s,y,t)$, $\issmaller(x,s,y,t)$, $\hasmax(x,s,y)$
and $\rcvinf(x,s)$ in $M$, we have $x,y\in\nw$ and $s,t\in\Nat$.
\end{itemize}
It can be shown by induction on the fixpoint computation of $\grded$
that $M$ is always well-formed. We omit the details.

The rest of this subsection is organized as follows. In Section~\ref{sub:partial-order-M},
we extract a happens-before relation $\cauM$ from $M$. Next, in
Section~\ref{sub:construction-of-run}, we construct a run $\run$:
we use $\cauM$ to establish a total order on $\nwnat$ that tells
us which are the active nodes in the transitions of $\run$. Finally,
we show in Section~\ref{sub:fair-run} that $\run$ is fair.

\subsubsection{Partial Order}

\label{sub:partial-order-M}

\newcommand{\cauMx}[1]{\prec_{M}^{#1}}

We define the following relation $\cauM$ on $\nwnat$: for each $(x,s)\in\nwnat$
and $(y,t)\in\nwnat$, we write $(x,s)\cauM(y,t)$ if and only if
$\before(x,s,y,t)\in M$. The rest of this section is dedicated to
showing that $\cauM$ is a well-founded strict partial order on $\nwnat$.

Let $\grded$ abbreviate the ground program $\grp MCI$ where $C=\pure{\ded}$
and $I=\decl H$. Regarding terminology, an edge $(x,s)\cauM(y,t)$
is called a \emph{local edge}, a \emph{message edge} or a \emph{transitive
edge} if the fact $\before(x,s,y,t)\in M$ can be derived by a ground
rule in $\grded$ of respectively the form (\ref{eq:before-step}),
the form (\ref{eq:before-send}), or the form (\ref{eq:before-tr}).%
\footnote{The body of such a ground rule has to be in $M$.%
}  It is possible that an edge is of two or even three types at the
same time.

Consider the following claim:

\begin{claim}\label{claim:M-partial-order}Relation $\cauM$ is a
strict partial order on $\nwnat$.\end{claim}

\begin{proof}
We show that $\cauM$ is transitive and irreflexive.

\paragraph*{Transitive}

First, we show that $\cauM$ is transitive. Suppose we have $(x,s)\cauM(z,u)$
and $(z,u)\cauM(y,t)$. We have to show that $(x,s)\cauM(y,t)$. By
definition of $\cauM$, we have $\before(x,s,z,u)\in M$ and $\before(z,u,y,t)\in M$.
Because rule~(\ref{eq:before-tr}) is positive, we have the following
ground rule in $\grded$:
\[
\before(x,s,y,t)\gets\before(x,s,z,u),\,\before(z,u,y,t).
\]
Because $M$ is a stable model and the body of the previous ground
rule is in $M$, we obtain $\before(x,s,y,t)\in M$. Hence, $(x,s)\cauM(y,t)$,
as desired.

\paragraph*{Irreflexive}

Because an edge $(x,s)\cauM(x,s)$ for any $(x,s)\in\nwnat$ would
form a cycle of length one, it is sufficient to show that there are
no cycles in $\cauM$ at all. This gives us irreflexivity, as desired.

First, let $\cauM'$ denote the restriction of $\cauM$ to the edges
that are local or message edges. Note that this definition allows
some edges in $\cauM'$ to also be transitive. The edges that are
missing from $\cauM'$ with respect to $\cauM$ are only derivable
by ground rules of the form (\ref{eq:before-tr}); we call these the
\emph{pure} transitive edges. We start by showing that $\cauM'$
contains no cycles. We show this with a proof by contradiction. So,
suppose that there is a cycle in $\nwnat$ through the edges of $\cauM'$:
\[
(x_{1},s_{1})\cauM(x_{2},s_{2})\cauM\ldots\cauM(x_{n},s_{n})
\]
with $n\geq2$ and $(x_{1},s_{1})=(x_{n},s_{n})$. We have $\before(x_{i},s_{i},x_{i+1},s_{i+1})\in M$
for each $i\in\{1,\ldots,n-1\}$. Based on these $\before$-facts,
ground rules in $\grded$ of the form (\ref{eq:before-tr}) will have
derived $\before(x_{i},s_{i},x_{j},s_{j})\in M$ for each $i,j\in\{1,\ldots,n\}$. 

If each edge on the above cycle would be only local,  then for each
$i,j\in\{1,\ldots,n\}$ with $i<j$ we have $x_{i}=x_{j}$ and $s_{i}<s_{j}$,
and hence $s_{1}\neq s_{n}$, which is false. So, there has to be
some $k\in\{1,\ldots,n-1\}$ such that $(x_{k},s_{k})\cauM(x_{k+1},s_{k+1})$
is a message edge, derived by a ground rule of the form (\ref{eq:before-send}):
\[
\before(x_{k},s_{k},x_{k+1},s_{k+1})\gets\chosen_{R}(x_{k},s_{k},x_{k+1},s_{k+1},\bar{a}).
\]
Therefore $\chosen_{R}(x_{k},s_{k},x_{k+1},s_{k+1},\bar{a})\in M$.
This $\chosen_{R}$-fact must be derived by a ground rule of the form
(\ref{eq:chosen}) in $\grded$, which implies that 
\[
\cand_{R}(x_{k},s_{k},x_{k+1},s_{k+1},\bar{a})\in M.
\]
This $\cand_{R}$-fact must in turn be derived by a ground rule $\grl$
of the form (\ref{eq:cand}). Because rules of the form (\ref{eq:cand})
in $\pure{\ded}$ contain a negative $\before$-atom in their body,
the presence of $\grl$ in $\grded$ requires that $\before(x_{k+1},s_{k+1},x_{k},s_{k})\notin M$.
But that is a contradiction, because $\before(x_{i},s_{i},x_{j},s_{j})\in M$
for each $i,j\in\{1,\ldots,n\}$ (see above).

Now we show there are no cycles in the entire relation $\cauM$. Since
$M=\grded(\decl H)$, we have $M=\bigcup_{i\in\Nat}M_{i}$ where $M_{0}=\decl H$
and $M_{i}=T(M_{i-1})$ for each $i\geq1$ where $T$ is the immediate
consequence operator of $\grded$. By induction on $i$, we show that
an edge $\before(x,s,y,t)\in M_{i}$ either is a local or message
edge, or it can be replaced by a path of local or message edges in
$M_{i}$. Then any cycle in $\cauM$ would imply there is a cycle
in $\cauM'$, which is impossible. So, $\cauM$ can not contain cycles.
Now, this induction property is satisfied for the base case because
$M_{0}$ does not contain $\before$-facts. For the induction hypothesis,
assume the property holds for $M_{i-1}$, where $i\geq1$. For the
inductive step, let $\before(x,s,y,t)\in M_{i}\setminus M_{i-1}$.
If this fact is derived by a ground rule of the form (\ref{eq:before-step})
or (\ref{eq:before-send}) then the property is satisfied. Now suppose
the fact is derived by a ground rule of the form (\ref{eq:before-tr}):
\[
\before(x,s,y,t)\gets\before(x,s,z,u),\,\before(z,u,y,t).
\]
Both body facts are in $M_{i-1}$, implying $M_{i-1}$ contains a
path of local or message edges from $(x,s)$ to $(z,u)$ and from
$(z,u)$ to $(y,t)$. Hence, using $M_{i-1}\subseteq M_{i}$, the
edge $\before(x,s,y,t)\in M_{i}$ can be replaced by a path of local
or message edges in $M_{i}$. \end{proof}

In Section~\ref{sub:causal-finite-transformation} we have added
extra rules to $\pure{\ded}$ to enforce that every node only receives
a finite number of messages during each step. We now verify that this
works correctly:

\begin{claim}\label{claim:finite-message-edges}For each $(y,t)\in\nwnat$
there are only a finite number of pairs $(x,s)\in\nwnat$ such that
$(x,s)\cauM(y,t)$ is a message edge.\end{claim}

\begin{proof}

We start by noting that $M$ does not contain the fact $\rcvinf(y,t)$.
Indeed, in order to derive this fact, we need a ground rule in $\grded$
of the form (\ref{eq:rcv-inf}), which has a body fact of the form
$\hassender(y,t,x,s)$. Such $\hassender$-facts must be generated
by ground rules in $\grded$ of the form (\ref{eq:has-sender}). The
rule (\ref{eq:has-sender}) negatively depends on relation $\rcvinf$.
Thus, specifically, if we want a ground rule in $\grded$ that can
derive $\hassender(y,t,x,s)$, we should require the absence of $\rcvinf(y,t)$
from $M$. So $\rcvinf(y,t)\in M$ requires $\rcvinf(y,t)\notin M$,
which is impossible.

The rest of the proof works towards a contradiction. So, suppose that
$(y,t)$ has an infinite number of incoming message edges. Because
there are only a finite number of nodes in $\nw$, there has to be
a node $x$ that has an infinite number of timestamps $s$ such that
$\before(x,s,y,t)\in M$ is a message edge. Since it is a message
edge, such a fact $\before(x,s,y,t)$ can be generated by a ground
rule in $\grded$ of the form (\ref{eq:before-send}), which implies
that there is a relation $R$ in $\idb{\ded}$ and a tuple $\bar{a}$
such that $\chosen_{R}(x,s,y,t,\bar{a})\in M$. Because $\rcvinf(y,t)\notin M$
(see above), for \emph{each} of these $\chosen_{R}$-facts, there
is a ground rule of the form (\ref{eq:has-sender}) in $M$ that derives
$\hassender(y,t,x,s)\in M$. 

Rule (\ref{eq:rcv-inf}) has a negative $\hasmax$-atom in its body.
If we can show that $\hasmax(y,t,x)\notin M$, then there will be
a ground rule in $\grded$ of the form (\ref{eq:rcv-inf}), where
$\hassender(y,t,x,s)\in M$:
\[
\rcvinf(y,t)\gets\hassender(y,t,x,s).
\]
This then causes $\rcvinf(y,t)\in M$, giving the desired contradiction.

Also towards a proof by contradiction, suppose that $\hasmax(y,t,x)\in M$.
This means that there is a ground rule $\grl$ in $\grded$ of the
form (\ref{eq:has-max}):
\[
\hasmax(y,t,x)\gets\hassender(y,t,x,s).
\]
Because the rule (\ref{eq:has-max}) contains a negative $\issmaller$-atom
in the body, and because $\grl\in\grded$, we know that $\issmaller(y,t,x,s)\notin M$.
But because there are infinitely many facts of the form $\hassender(y,t,x,s')\in M$,
there is at least one fact $\hassender(y,t,x,s')\in M$ with $s<s'$.
Moreover, the rule (\ref{eq:is-smaller}) is positive, and therefore
the following ground rule is always in $\grded$:
\[
\issmaller(y,t,x,s)\gets\hassender(y,t,x,s),\,\hassender(y,t,x,s'),\, s<s'.
\]
Since the body of this ground rule is in $M$, the rule derives $\issmaller(y,t,x,s)\in M$,
which gives the desired contradiction. \end{proof}

An ordering $\prec$ on a set $A$ is called \emph{well-founded}
if for each $a\in A$, there are only a finite number of elements
$b\in A$ such that $b\prec a$. We now use Claim~\ref{claim:finite-message-edges}
to show:

\begin{claim}Relation $\cauM$ on $\nwnat$ is well-founded.\end{claim}

\begin{proof}

Let $(x,s)\in\nwnat$. We have to show that there are only a finite
number of pairs $(y,t)\in\nwnat$ such that $(y,t)\cauM(x,s)$. Technically,
we can limit our attention to paths in $\cauM$ consisting of local
edges and message edges, because if we can show that there are only
a finite number of predecessors of $(x,s)$ on such paths, then there
are only a finite number of predecessors when we include the transitive
edges as well. First we show that every pair $(y,t)\in\nwnat$ has
only a finite number of incoming local and message edges. If $t>0$,
we can immediately see that $(y,t)$ has precisely one incoming local
edge, as created by a ground rule of the form (\ref{eq:before-step}),
and if $t=0$ then $(y,t)$ has no incoming local edge. Also, Claim~\ref{claim:finite-message-edges}
tells us that $(y,t)$ has only a finite number of incoming message
edges. So, the number of incoming local and message edges in $(y,t)$
is finite.

Let $(y,t)\in\nwnat$ be a pair such that $(y,t)\cauM(x,s)$ is a
local edge or a message edge. Starting in $(x,s)$, we can follow
this edge backwards so that we reach $(y,t)$. If $(y,t)$ itself
has incoming local or message edges, from $(y,t)$ we can again follow
an edge backwards. This way we can incrementally construct backward
paths starting from $(x,s)$. Because at each pair of $\nwnat$ there
are only a finite number of incoming local or message edges (shown
above), if $(x,s)$ would have an infinite number of predecessors,
we must be able to construct a backward path of infinite length. 
We now show that the existence of such an infinite path leads to a
contradiction.  So, suppose that there is a backward path of infinite
length. Because there are only a finite number of nodes in the network
$\nw$, there must be a node $y$ that occurs infinitely often on
this path. We will now show that, as we progress further along the
backward path, we must see the local timestamps of $y$ strictly decrease.
Hence, we must eventually reach timestamp $0$ of $y$, after which
we cannot decrement the timestamps of $y$ anymore, and thus it is
impossible that $y$ occurs infinitely often along the path. Suppose
that the timestamps of $y$ do not strictly decrease. There are two
cases. First, if the same pair $(y,t)$ would occur twice on the path,
we would have a cycle in $\cauM$, which is not possible by Claim~\ref{claim:M-partial-order}.
Secondly, suppose that there are two timestamps $t$ and $t'$ of
$y$ such that $t<t'$ and $(y,t)$ occurs before $(y,t')$ on the
backward path, meaning that $(y,t)$ lies closer to $(x,s)$. Because
the edges were followed in reverse, we have 
\[
(y,t')\cauM\ldots\cauM(y,t).
\]
But since $t<t'$, by means of local edges, we always have 
\[
(y,t)\cauM(y,t+1)\cauM\ldots\cauM(y,t').
\]
So, there would be a cycle between $(y,t')$ and $(y,t)$. But that is again impossible
by Claim~\ref{claim:M-partial-order}. \end{proof}

\subsubsection{Construction of Run}

\label{sub:construction-of-run}

Let $\cauM$ be the well-founded strict partial order on $\nwnat$
as defined in the preceding subsection. The relation $\cauM$ has
the intuition of a happens-before relation of a run (Section~\ref{sub:run-happens-before}),
but the novelty is that it comes from a purely declarative model $M$.
We will now use $\cauM$ to construct a run $\run$ such that $\trace{\run}=\proj M{\toloct{\schof{\ded}}}$.

\paragraph*{Total order}

It is well-known that a well-founded strict partial order can be extended
to a well-founded strict total order. So, let $\totM$ be a well-founded
strict total order on $\nwnat$ that extends $\cauM$, i.e., for each
$(x,s)\in\nwnat$ and $(y,t)\in\nwnat$, if $(x,s)\cauM(y,t)$ then
$(x,s)\totM(y,t)$, but the reverse does not have to hold. 

Ordering the set $\nwnat$ according to $\totM$ gives us a sequence
of pairs that will form the transitions in the constructed run $\run$.
Concretely, we obtain a sequence of nodes by taking the node-component
from each pair. This will form our sequence of active nodes. Similarly,
by taking the timestamp-component from each pair of $\nwnat$, we
obtain a sequence of timestamps. These are the local clocks of the
active nodes during their transitions. 

We introduce some extra notations to help us reason about the ordering
of time that is implied by $\totM$. For each $(x,s)\in\nwnat$, let
$\globM{x,s}\in\Nat$ denote the ordinal of $(x,s)$ as implied by
$\totM$, which is well-defined because $\totM$ is well-founded.
For technical convenience, we let ordinals start at $0$. Note, $\globM{\cdot}$
is an injective function. For any $i\in\Nat$, we define $(x_{i},s_{i})$
to be the unique pair in $\nwnat$ such that $\globM{x_{i},s_{i}}=i$.

As a counterpart to function $\globM{\cdot}$, for each $i\in\Nat$
and each $x\in\nw$, let $\locM{i,x}$ denote the \emph{size} of
the set 
\[
\{s\in\Nat\mid\globM{x,s}<i\}.
\]
Intuitively, if $i$ is regarded to be the ordinal of a transition
in a run, $\locM{i,x}$ is the number of local steps of $x$ that
came before transition $i$, i.e., the number of transitions before
$i$ in which $x$ was the active node. If $x=x_{i}$ (the active
node) then $\locM{i,x}$ is effectively the timestamp of $x$ \emph{during}
transition $i$, and if $x\neq x_{i}$ then $\locM{i,x}$ is the next
timestamp of $x$ that still has to come \emph{after} transition
$i$. Note, the functions $\globM{\cdot}$ and $\locM{\cdot}$ closely
resemble the functions $\globR{\cdot}$ and $\locR{\cdot}$ of Section~\ref{sub:run-timestamps-and-trace}.

\paragraph*{Configurations}

We will now define the desired run $\run$ of $\ded$ on $H$. First
we define an infinite sequence of configurations $\cnf_{0}$, $\cnf_{1}$,
$\cnf_{2}$, etc. In a second step we will connect each pair of subsequent
configurations by a transition. Recall from Section~\ref{sub:configurations}
that a configuration describes for each node what facts it has stored
locally (state), and also what messages have been sent to this node
but that are not yet received (message buffer). The facts that are
stored on a node are either input $\mathit{edb}$-facts, or facts
derived by inductive rules in a previous step of the node. The first
kind of facts can be easily obtained from $M$ by keeping only the
facts over schema $\toloct{\edb{\ded}}$, which gives a subset of
$\decl H$. 

For the second kind of state facts, we look at the inductively derived
facts in $M$. Rules in $\pure{\ded}$ that represent inductive rules
of $\ded$ are recognizable as rules of the form (\ref{eq:pure-ind}):
they have a head atom over $\toloct{\schof{\ded}}$ and they have
a (positive) $\timesucc$-atom in their body. No other kind of rule
in $\pure{\ded}$ has this form. Hence, the ground rules in $\grded$
that are based on rules of the form (\ref{eq:pure-ind}) are also
easily recognizable, and we will call these \emph{inductive ground
rules}. A ground rule $\grl\in\grded$ is called \emph{active}
on $M$ if $\bpos{\grl}\subseteq M$, which implies $\head{\grl}\in M$
because $M$ is stable. Let $\Mind$ denote all head atoms of inductive
ground rules in $\grded$ that are active on $M$. Note that $\Mind\subseteq M$.
Regarding notation, for an instance $I$ over $\toloct{\schof{\ded}}$,
we write $\droplt I$ to denote the set $\{R(\bar{a})\mid\exists x,s:\, R(x,s,\bar{a})\in I\}$,
and we write $\projlt Ixs$ to denote the set $\{R(y,t,\bar{a})\in I\mid y=x,\, t=s\}$. 

Now, for each $i\in\Nat$, for each node $x\in\nw$, denoting $s=\locM{i,x}$,
in configuration $\cnf_{i}=(\cnfs_{i},\cnfb_{i})$, the state $\cnfs_{i}(x)$
is defined as
\[
\droplt{\left(\shprojlt{\proj M{\toloct{\edb{\ded}}}}xs\cup\projlt{\Mind}xs\right)}.
\]
We remove the location specifier and timestamp because we have to
obtain facts over the schema of $\ded$, not over the schema of $\pure{\ded}$.

Now we define the message buffers in the configurations. Recall that
the message buffer of a node always contains pairs of the form $\pair j{\fc}$,
where $j\in\Nat$ is the transition in which fact $\fc$ was sent.
For each $i\in\Nat$, for each node $x\in\nw$, in configuration $\cnf_{i}=(\cnfs_{i},\cnfb_{i})$,
the message buffer $\cnfb_{i}(x)$ is defined as
\[
\begin{array}{ll}
\{\pair{\globM{y,t}}{\, R(\bar{a})}\mid
    & \exists u:\,\chosen_{R}(y,t,x,u,\bar{a})\in M,\\&\globM{y,t}<i\leq\globM{x,u}\}.
\end{array}
\]
Note the use of addressee $x$ in this definition. The definition
of $\cnfb_{i}(x)$ reflects the operational semantics, in that the
messages in the buffer of node $x$ must be sent in a previous transition,
as expressed by the constraint $\globM{y,t}<i$. Moreover, the constraint
$i\leq\globM{x,u}$ says that $\cnfb_{i}(x)$ contains only messages
that will be delivered in transitions of $x$ that come after configuration
$\cnf_{i}$. Possibly $i=\globM{x,u}$, and in that case the message
will be delivered in the transition immediately after configuration
$\cnf_{i}$, which is transition $i$ (see also below).

\paragraph*{Transitions}

So far we have obtained a sequence of configurations $\cnf_{0}$,
$\cnf_{1}$, $\cnf_{2}$, etc. Now we define a sequence of tuples,
one tuple per ordinal $i\in\Nat$, that represents the transition
$i$. Let $i\in\Nat$. Recall from above that $(x_{i},s_{i})$ is
the unique pair in $\nwnat$ such that $\globM{x_{i},s_{i}}=i$. The
tuple $\tup_{i}$ is defined as $(\cnf_{i},x_{i},m_{i},i,\cnf_{i+1})$,
where 
\[
m_{i}=\{\pair{\globM{y,t}}{\, R(\bar{a})}\mid\chosen_{R}(y,t,z,u,\bar{a})\in M,\,\globM{z,u}=i\}.
\]
Intuitively, $m_{i}$ selects all messages that arrive in transition
$i$. And since $\globM{z,u}=i$ implies $z=x_{i}$ and $u=s_{i}$,
we thus select all messages destined for step $s_{i}$ of node $x_{i}$.

\paragraph*{Trace}

We can show that sequence $\run$ is indeed
a legal run of $\ded$ on input $H$ such that $\trace{\run}=\proj M{\toloct{\schof{\ded}}}$; this proof can be found in \ref{app:proof-dir-2} of the online appendix to the paper.
In the following subsection we show that $\run$ is also fair.

\subsubsection{Fair Run}

\label{sub:fair-run}

Let $\run$ be the run as constructed in the previous subsection.
We now show that $\run$ is fair. For each transition index $i\in\Nat$,
let $\cnf_{i}=(\cnfs_{i},\cnfb_{i})$ denote the source configuration
of transition $i$. Recall from Section~\ref{sub:paper-fairness-and-arrival}
that we have to check two fairness conditions:
\begin{enumerate}
\item every node is the active node in an infinite number of transitions;
and,
\item for every transition $i\in\Nat$, for every $y\in\nw$, for every
pair $\pair j{\fc}\in\cnfb_{i}(y)$, there is a transition $k$ with
$i\leq k$ in which $\pair j{\fc}$ is delivered to $y$.
\end{enumerate}
We show that $\run$ satisfies the first fairness condition. Let $x\in\nw$
be a node, and let $s\in\Nat$ be a timestamp of $x$. Consider transition
$i=\globM{x,s}$. This transition has active node $x_{i}=x$. We can
find such a transition with active node $x$ for every timestamp $s\in\Nat$
of $x$, and these transitions are all unique because function $\globM{\cdot}$
is injective. So, there are an infinite number of transitions in $\run$
with active node $x$.

We show that $\run$ satisfies the second fairness condition. Let
$i\in\Nat$, $y\in\nw$, and $\pair j{\fc}\in\cnfb_{i}(y)$. Denote
$\fc=R(\bar{a})$. From its construction, the pair $\pair j{\fc}\in\cnfb_{i}(y)$
implies there are values $x\in\nw$, $s\in\Nat$ and $t\in\Nat$ such
that $\chosen_{R}(x,s,y,t,\bar{a})\in M$ and $j=\globM{x,s}<i\leq\globM{y,t}$.
Denote $k=\globM{y,t}$.  Hence, $i\leq k$ and $\pair j{\fc}\in m_{k}$
by definition of $m_{k}$. Thus $\pair j{\fc}$ is delivered to $x_{k}=y$
in transition $k$.

\section{Discussion}

\label{sec:discussion}

We have represented distributed programs in $\datalog$ under the
stable model semantics. Moreover, we have shown that the stable models
represent the desired behavior of the distributed program, as found
in a realistic operational semantics. We now discuss some points for
future work. 

As mentioned, many $\datalog$-inspired languages have been proposed
to implement distributed applications \cite{decl_netw_cacm,declnetw_opsem,grumbach_netlog,webdamlog},
and they contain several powerful features such as aggregation and
non-determinism (choice). Our current framework already represents
the essential features that all these languages possess: reasoning
about distributed state and representing message sending. Nonetheless,
we have probably not yet explored the full power of stable models.
We therefore expect that this work can be extended to languages that
incorporate more powerful language constructs such as the ones mentioned
above. It might also be possible to remove the syntactic stratification
condition that we have used for the deductive rules.

More related to multi-agent systems \cite{leite_minerva,nigam_agents,leite_evolp},
it might be interesting to allow logic programs used in declarative
networking to dynamically modify their rules. The question would be
how (and if) this can be represented in our model-based semantics.

The effect of variants of the model-based semantics can studied. For
example, messages can be sent into the past when the causality rules
are removed. Then, one might ask which (classes of) programs
still work ``correctly'' under such a non-causal semantics; some
preliminary results are in~\cite{ameloot_noncausality2014}.

Lastly, we can think about the output of distributed $\datalog$ programs.
\citeNS{dedalus_consistency} define the output with
\emph{ultimate} facts, which are facts that will eventually always
be present on the network. This way, the output of a run (or equivalently
stable model) can be defined. Then, a \emph{consistent} program
is required to produce the same output in every run. For consistent
programs, the output on an input distributed database instance can
thus be defined as the output of any run. We can now consider the
following decision problem: for a consistent program, an input distributed
database instance for that program, and a fact, decide if this fact
is output by the program on that input. We think that decidability
depends on the semantics of the message buffers. In this paper, we
have represented per addressee duplicate messages in its message buffer.
This is a realistic representation, since in a real network, the same
message can be sent multiple times, and hence, multiple instances
of the same message can be in transmission simultaneously. If we would
forbid duplicate messages in the buffers, then the decision problem
becomes decidable because only a finite number of configurations would
be possible by finiteness of the input domain. But when duplicates
are preserved, the number of configurations is not limited, and we
expect that the problem will be undecidable in general. However, we
might want to investigate whether decidability can be obtained in
particular (syntactically defined) cases. If so, it might be interesting
for those cases to find finite representations of the stable models.
This could serve as a more intuitive programmer abstraction, or it
could perhaps be used to more efficiently simulate the behavior of
the network for testing purposes.

\section*{Acknowledgment}

The second author thanks Serge~Abiteboul for a number of interesting
discussions.

\bibliographystyle{acmtrans}

\section*{Appendix}

\begin{appendix}

\newcommand{\strat}[2]{#1^{\to#2}}

\section*{General Remarks}

Let $\ded$ be a $\dedalus$ program. Recall from Section~\ref{sub:subprograms}
that $\deduc{\ded}\subseteq\ded$ is the subset of all (unmodified)
deductive rules. The semantics of $\deduc{\ded}$ is given by the
stratified semantics. Although the semantics of $\deduc{\ded}$ does
not depend on the chosen syntactic stratification, for technical convenience
in the proofs, we will fix an arbitrary syntactic stratification for
$\deduc{\ded}$. Whenever we refer to the stratum number of an \emph{idb}
relation, we implicitly use this fixed syntactic stratification. Stratum
numbers start at $1$.

\section{Run to Model: Proof Details}

\label{app:proof-dir-1}

In the context of Section~\ref{sub:run-to-model-construction}, we
show that $M$ is a model of $\ded$ on input $H$. Let $\grded$
abbreviate the ground program $\grp MCI$, where $C=\pure{\ded}$
and $I=\decl H$. To show that $M$ is a stable model, we have to
show $M=N$ where $N=\grded(\decl H)$. The inclusions $M\subseteq N$
and $N\subseteq M$ are shown respectively in Sections~\ref{sub:M-in-N}
and \ref{sub:N-in-M}. We use the notations of Section~\ref{sub:run-to-model-construction}.

\subsection{Inclusion $M\subseteq N$}

\label{sub:M-in-N}

By definition,
\[
M=\decl H\cup\bigcup_{i\in\Nat}\slice i.
\]
We immediately have $\decl H\subseteq N$ by the semantics of $\grded$.
Next, we define for uniformity the set $\slice{-1}=\emptyset$. We
will show by induction on $i=-1,$ $0$, $1$, $\ldots$, that $\slice i\subseteq N$.
The base case ($i=-1$) is clear. For the induction hypothesis, let
$i\geq0$, and assume for all $j\in\{-1,0,\ldots,i-1\}$ that $\slice j\subseteq N$.
We show that $\slice i\subseteq N$. By definition, 
\[
\slice i=\slicecaus i\cup\slicefin i\cup\sliceduc i\cup\slicesnd i.
\]
We show inclusion of these four sets in $N$ below. Auxiliary claims
can be found in Section~\ref{sub:run-to-model-first-dir--claims}.

\subsubsection{Causality}

We show that $\slicecaus i\subseteq N$. Concretely, let $(x,s)\in\nwnat$
such that $(x,s)\caus(x_{i},s_{i})$. We show $\before(x,s,x_{i},s_{i})\in N$.
We distinguish between the following cases.

\paragraph*{Local edge}

Suppose $(x,s)\caus(x_{i},s_{i})$ is a local edge, i.e., $x=x_{i}$
and $s_{i}=s+1$. Because rule~(\ref{eq:before-step}) is positive,
the following ground rule is always in $\grded$:
\[
\before(x,s,x,s+1)\gets\relall(x),\,\timesucc(s,s+1).
\]
The body facts of this ground rule are in $\decl H\subseteq N$; hence,
the rule derives $\before(x,s,x,s+1)=\before(x,s,x_{i},s_{i})\in N$.

\paragraph*{Message edge}

Suppose $(x,s)\caus(x_{i},s_{i})$ is a message edge, i.e., there
is an earlier transition $j<i$ with $j=\globR{x,s}$, in which $x$
sends a message $\fc$ to $x_{i}$ such that $\arr(j,x_{i},\fc)=i$.
Denote $\fc=R(\bar{a})$. Because rules of the form (\ref{eq:before-send})
in $\pure{\ded}$ are positive, the following ground rule is always
in $\grded$:
\[
\before(x,s,x_{i},s_{i})\gets\chosen_{R}(x,s,x_{i},s_{i},\bar{a}).
\]
We show $\chosen_{R}(x,s,x_{i},s_{i},\bar{a})\in N$, so that $\before(x,s,x_{i},s_{i})\in N$,
as desired. Since $j=\globR{x,s}$, we have $x_{j}=x$ and $s_{j}=s$.
Also using $s_{i}=\locR i$, we have
\[
\chosen_{R}(x,s,x_{i},s_{i},\bar{a})\in\slicesnd j\subseteq\slice j.
\]
Lastly, we have $\slice j\subseteq N$ by applying the induction hypothesis.

\paragraph*{Transitive edge}

Suppose $(x,s)\caus(x_{i},s_{i})$ is not a local edge nor a message
edge. Then we can choose a pair $(z,u)\in\nwnat$ such that $(x,s)\caus(z,u)$
and $(z,u)\caus(x_{i},s_{i})$, but also such that $(z,u)\caus(x_{i},s_{i})$
is a local edge or a message edge. Because rule (\ref{eq:before-tr})
is positive, the following ground rule is always in $\grded$:
\[
\before(x,s,x_{i},s_{i})\gets\before(x,s,z,u),\,\before(z,u,x_{i},s_{i}).
\]
We now show that the body of this rule is in $N$, so that $\before(x,s,x_{i},s_{i})\in N$,
as desired. Denote $j=\globR{z,u}$. First, because $(x,s)\caus(z,u)$,
we have $\before(x,s,z,u)\in\slicecaus j$. Next, because $(z,u)\caus(x_{i},s_{i})$,
we have $j<i$ by Lemma~\ref{lem:caus-global-order}. So, by applying
the induction hypothesis to $j$, we have $\before(x,s,z,u)\in N$.
Secondly, because $(z,u)\caus(x_{i},s_{i})$ is a local edge or a
message edge, we have $\before(z,u,x_{i},s_{i})\in N$ as shown in
the preceding two cases.

\subsubsection{Finite Messages}

We show that $\slicefin i\subseteq N$. Let $\senders i$ be as defined
in Section~\ref{sub:run-to-model-construction}. For each of the
different kinds of facts in $\slicefin i$, we show inclusion in $N$.

\paragraph*{Senders}

Let $\hassender(x_{i},s_{i},x,s)\in\slicefin i$. We have $(x,s)\in\senders i$,
which means that $x$ during step $s$ sends some message fact $R(\bar{a})$
that arrives in step $s_{i}$ of $x_{i}$. Rules in $\pure{\ded}$
of the form (\ref{eq:has-sender}) have a negative $\rcvinf$-atom
in their body. But since we have not added any $\rcvinf$-facts to
$M$, including $\rcvinf(x_{i},s_{i})$, the following rule is in
$\grded$:
\[
\hassender(x_{i},s_{i},x,s)\gets\chosen_{R}(x,s,x_{i},s_{i},\bar{a}).
\]
We are left to show that $\chosen_{R}(x,s,x_{i},s_{i},\bar{a})\in N$.
Denote $j=\globR{x,s}$. Using that $x=x_{j}$ and $s=s_{j}$, we
have $\chosen_{R}(x,s,x_{i},s_{i},\bar{a})\in\slicesnd j$. Because
$j<i$ by the operational semantics, we can apply the induction hypothesis
to $j$ to know $\slicesnd j\subseteq N$.

\paragraph*{Comparison of timestamps}

Let $\issmaller(x_{i},s_{i},x,s)\in\slicefin i$. We have $(x,s)\in\senders i$
and there is a timestamp $s'\in\Nat$ so that $(x,s')\in\senders i$
and $s<s'$. Rule (\ref{eq:is-smaller}) is positive and therefore
the following ground rule is always in $\grded$:
\begin{eqnarray*}
\issmaller(x_{i},s_{i},x,s) & \gets & \hassender(x_{i},s_{i},x,s),\,\hassender(x_{i},s_{i},x,s'),\\
 &  & s<s'.
\end{eqnarray*}
We immediately have $(s<s')\in\decl H\subseteq N$. By construction
of $\slicefin i$, we also have $\hassender(x_{i},s_{i},x,s)\in\slicefin i$
and $\hassender(x_{i},s_{i},x,s')\in\slicefin i$, and thus both facts
are also in $N$ as shown above. Hence the previous ground rule derives
$\issmaller(x_{i},s_{i},x,s)\in N$.

\paragraph*{Maximum timestamp}

Let $\hasmax(x_{i},s_{i},x)\in\slicefin i$. Thus $x$ is a sender-node
mentioned in $\senders i$. Let $s$ be the maximum send-timestamp
of $x$ in $\senders i$, which surely exists because $\senders i$
is finite. We have not added $\issmaller(x_{i},s_{i},x,s)$ to $\slicefin i$,
and thus also not to $M$. Although rule (\ref{eq:has-max}) contains
a negated $\issmaller$-atom, $\issmaller(x_{i},s_{i},x,s)\notin M$
implies that the following ground rule is in $\grded$:
\[
\hasmax(x_{i},s_{i},x)\gets\hassender(x_{i},s_{i},x,s).
\]
Moreover, $(x,s)\in\senders i$ implies $\hassender(x_{i},s_{i},x,s)\in N$,
and thus the previous ground rule derives $\hasmax(x_{i},s_{i},x)\in N$,
as desired.

\subsubsection{Deductive}

We show that $\sliceduc i\subseteq N$. By definition, $\sliceduc i=\addlt{D_{i}}{x_{i}}{s_{i}}$,
where $D_{i}$ is the output of subprogram $\deduc{\ded}$ during
transition $i$. Recall from Section~\ref{sub:transitions-and-runs}
that $\deduc{\ded}$ is given the following input during transition
$i$:
\[
\cnfs_{i}(x_{i})\cup\untag{m_{i}},
\]
where $\cnfs_{i}$ denotes the state at the beginning of transition
$i$, and $m_{i}$ is the set of (tagged) messages delivered during
transition $i$. If we can show that $\addlt{(\cnfs_{i}(x_{i})\cup\untag{m_{i}})}{x_{i}}{s_{i}}\subseteq N$,
then we can apply Claim~\ref{claim:duc-in-stable} to know that $\addlt{D_{i}}{x_{i}}{s_{i}}\subseteq N$,
as desired.

\paragraph*{State}

We first show $\addlt{\cnfs_{i}(x_{i})}{x_{i}}{s_{i}}\subseteq N$.
There are two cases:
\begin{itemize}
\item Suppose $s_{i}=0$, i.e., $i$ is the first transition of $\run$
with active node $x_{i}$. Then $\cnfs_{i}(x_{i})=H(x_{i})$ by the
operational semantics, which gives $\addlt{\cnfs_{i}(x_{i})}{x_{i}}{s_{i}}\subseteq\decl H\subseteq N$
by definition of $\decl H$.
\item Suppose $s_{i}>0$. Then we can consider the last transition $j$
of $x_{i}$ that came before $i$. By the operational semantics, we
have $\cnfs_{i}(x_{i})=\cnfs_{j+1}(x_{i})$, where $\cnfs_{j+1}$
is the state resulting from transition $j$. More concretely, $\cnfs_{i}(x_{i})=H(x_{i})\cup\induc{\ded}(D_{j})$,
with $D_{j}$ the output of $\deduc{\ded}$ during transition $j$.
As in the previous case, we already know $\addlt{H(x_{i})}{x_{i}}{s_{i}}\subseteq\decl H$.
Now, by applying the induction hypothesis to $j$, we have $\sliceduc j\subseteq\slice j\subseteq N$.
Next, by applying Claim~\ref{claim:ind-in-stable}, and by using
$s_{i}=s_{j}+1$, we obtain
\begin{eqnarray*}
\addlt{\cnfs_{i}(x_{i})}{x_{i}}{s_{i}} & = & \addlt{H(x_{i})}{x_{i}}{s_{i}}\cup\addlt{\induc{\ded}(D_{j})}{x_{i}}{s_{j}+1}\\
 & \subseteq & N.
\end{eqnarray*}

\end{itemize}

\paragraph*{Messages}

Now we show $\addlt{\untag{m_{i}}}{x_{i}}{s_{i}}\subseteq N$. Let
$\fc\in\untag{m_{i}}$. We have to show that $\addlt{\fc}{x_{i}}{s_{i}}\in N$.
First, because $\fc\in\untag{m_{i}}$, there is a transition $k$
with $k<i$ such that $\pair k{\fc}\in m_{i}$, i.e., the fact $\fc$
was sent to $x_{i}$ during transition $k$ (by node $x_{k}$). Denote
$\fc=R(\bar{a})$. So, there must be an asynchronous rule with head-predicate
$R$ in $\ded$, which has a corresponding rule in $\pure{\ded}$
of the form (\ref{eq:deliv}). Rules of the form (\ref{eq:deliv})
are positive and thus the following ground rule is always in $\grded$:
\[
R(x_{i},s_{i},\bar{a})\gets\chosen_{R}(x_{k},s_{k},x_{i},s_{i},\bar{a}).
\]
We show $\chosen_{R}(x_{k},s_{k},x_{i},s_{i},\bar{a})\in N$, so that
the rule derives $\addlt{\fc}{x_{i}}{s_{i}}\in N$, as desired. Because
$x_{k}$ sends $\fc$ to $x_{i}$ during transition $k$, and $i$
is the transition in which this message is delivered to $x_{i}$,
we have $\chosen_{R}(x_{k},s_{k},x_{i},s_{i},\bar{a})\in\slicesnd k\subseteq\slice k$.
By applying the induction hypothesis to $k$, we have $\slicesnd k\subseteq N$.

\subsubsection{Sending}

We show that $\slicesnd i\subseteq N$. For each kind of fact in $\slicesnd i$
we show inclusion in $N$.

\paragraph*{Candidates}

Let $\cand_{R}(x_{i},s_{i},y,t,\bar{a})\in\slicesnd i$. We have $R(y,\bar{a})\in\mesg i$,
$t\in\Nat$ and $(y,t)\not\caus(x_{i},s_{i})$. Since $\addlt{D_{i}}{x_{i}}{s_{i}}\subseteq N$
(see above), we can use Claim~\ref{claim:cand-in-stable} to obtain
$\cand_{R}(x_{i},s_{i},y,t,\bar{a})\in N$, as desired.

\paragraph*{Chosen}

Let $\chosen_{R}(x_{i},s_{i},y,t,\bar{a})\in\slicesnd i$. We have
$R(y,\bar{a})\in\mesg i$ and $t=\locR j$ with $j=\arr(i,y,R(\bar{a}))$.
Because $R(y,\bar{a})\in\mesg i$, this fact was produced by $\async{\ded}$,
and thus there is an asynchronous rule in $\ded$ with head-predicate
$R$. This asynchronous rule has a corresponding rule in $\pure{\ded}$
of the form (\ref{eq:chosen}), that contains a negated $\other_{R}$-atom
in the body. But by construction of $\slicesnd i$, we have not added
$\other_{R}(x_{i},s_{i},y,t,\bar{a})$ to $\slicesnd i$, and thus
also not to $M$. Therefore the following ground rule of the form
(\ref{eq:chosen}) is in $\grded$:
\[
\chosen_{R}(x_{i},s_{i},y,t,\bar{a})\gets\cand_{R}(x_{i},s_{i},y,t,\bar{a}).
\]
Because $j>i$ by the operational semantics, we have $(y,t)\not\caus(x_{i},s_{i})$
by Lemma~\ref{lem:caus-global-order}. Thus, by construction of $\slicesnd i$,
we have $\cand_{R}(x_{i},s_{i},y,t,\bar{a})\in\slicesnd i$, in which
case $\cand_{R}(x_{i},s_{i},y,t,\bar{a})\in N$ (shown above). Hence,
the previous ground rule derives $\chosen_{R}(x_{i},s_{i},y,t,\bar{a})\in N$,
as desired.

\paragraph*{Other}

Let $R(y,\bar{a})$ and $t$ be from above. Let $\other_{R}(x_{i},s_{i},y,u,\bar{a})\in\slicesnd i$.
We have $u\in\Nat$, $(y,u)\not\caus(x_{i},s_{i})$ and $u\neq t$.
Because rule (\ref{eq:other}) is positive, the following ground rule
is in $\grded$:
\begin{eqnarray*}
\other_{R}(x_{i},s_{i},y,u,\bar{a}) & \gets & \cand_{R}(x_{i},s_{i},y,u,\bar{a}),\,\chosen_{R}(x_{i},s_{i},y,t,\bar{a}),\\
 &  & u\neq t.
\end{eqnarray*}
We immediately have $(u\neq t)\in\decl H\subseteq N$. Now we show
that the other body facts are in $N$, so the rule derives $\other_{R}(x_{i},s_{i},y,u,\bar{a})\in N$,
as desired. Because $(y,u)\not\caus(x_{i},s_{i})$, by construction
of $\slicesnd i$, we have $\cand_{R}(x_{i},s_{i},y,u,\bar{a})\in\slicesnd i$
and thus $\cand_{R}(x_{i},s_{i},y,u,\bar{a})\in N$ (shown above).
Moreover, it was shown above that $\chosen_{R}(x_{i},s_{i},y,t,\bar{a})\in N$.

\subsubsection{Subclaims}

\label{sub:run-to-model-first-dir--claims}

\begin{claim}\label{claim:duc-in-stable}Let $i$ be a transition
of $\run$. If $\addlt{(\cnfs_{i}(x_{i})\cup\untag{m_{i}})}{x_{i}}{s_{i}}\subseteq N$,
then $\addlt{D_{i}}{x_{i}}{s_{i}}\subseteq N$.\end{claim}

\begin{proof}

Abbreviate $I_{i}=\cnfs_{i}(x_{i})\cup\untag{m_{i}}$. Recall that
$D_{i}=\deduc{\ded}(I_{i})$, which is computed with the stratified
semantics.

For $k\in\Nat$, we write $\strat{D_{i}}k$ to denote the set obtained
by adding to $I_{i}$ all facts derived in stratum $1$ up to stratum
$k$ during the computation of $D_{i}$. For the largest stratum number
$n$ of $\deduc{\ded}$, we have $\strat{D_{i}}n=D_{i}$. Also, because
stratum numbers start at $1$, we have $\strat{D_{i}}0=I_{i}$. We
show by induction on $k=0$, $1$, $2$, $\ldots$, $n$, that $\shaddlt{\strat{D_{i}}k}{x_{i}}{s_{i}}\subseteq N$.

\paragraph*{Base case}

For the base case, $k=0$, the property holds by the given assumption
$\addlt{I_{i}}{x_{i}}{s_{i}}\subseteq N$.

\paragraph*{Induction hypothesis}

For the induction hypothesis, assume for some stratum number $k$
with $k\geq1$ that $\shaddlt{\strat{D_{i}}{k-1}}{x_{i}}{s_{i}}\subseteq N$.

\paragraph*{Inductive step}

For the inductive step, we show that $\shaddlt{\strat{D_{i}}k}{x_{i}}{s_{i}}\subseteq N$.
Recall that the input of stratum $k$ in $\deduc{\ded}$ is the set
$\strat{D_{i}}{k-1}$, and the semantics is given by the fixpoint
semantics of semi-positive $\datalogneg$ (see Section \ref{sub:stratified-semantics}).
So, we can consider $\strat{D_{i}}k$ to be a fixpoint, i.e., as the
set $\bigcup_{l\in\Nat}A_{l}$ with $A_{0}=\strat{D_{i}}{k-1}$ and
$A_{l}=T(A_{l-1})$ for each $l\geq1$, where $T$ is the immediate
consequence operator of stratum $k$. We show by inner induction on
$l=0$, $1$, etc, that 
\[
\shaddlt{A_{l}}{x_{i}}{s_{i}}\subseteq N.
\]
For the base case ($l=0$), we have $A_{0}=\strat{D_{i}}{k-1}$, for
which we can apply the outer induction hypothesis to know that $\shaddlt{\strat{D_{i}}{k-1}}{x_{i}}{s_{i}}=\shaddlt{A_{0}}{x_{i}}{s_{i}}\subseteq N$,
as desired. For the inner induction hypothesis, we assume for some
$l\geq1$ that $\shaddlt{A_{l-1}}{x_{i}}{s_{i}}\subseteq N$. For
the inner inductive step, we show that $\shaddlt{A_{l}}{x_{i}}{s_{i}}\subseteq N$.
Let $\fc\in A_{l}\setminus A_{l-1}$. Let $\rl\in\deduc{\ded}$ and
$V$ be a rule from stratum $k$ and valuation respectively that have
derived $\fc$. Let $\rl'$ be the rule in $\pure{\ded}$ obtained
by applying the transformation (\ref{eq:pure-duc}) to $\rl$. Let
$V'$ be $V$ extended to assign $x_{i}$ and $s_{i}$ to the new
variables in $\rl'$ that represent the location and timestamp respectively.
Note in particular that $V'(\bpos{\rl'})=\addlt{V(\bpos{\rl})}{x_{i}}{s_{i}}$
and $V'(\bneg{\rl'})=\addlt{V(\bneg{\rl})}{x_{i}}{s_{i}}$. Let $\grl$
be the \emph{positive} ground rule obtained by applying $V'$ to
$\rl'$ and by subsequently removing all negative (ground) body atoms.
We show that $\grl\in\grded$ and that its body is in $N$, so that
$\grl$ derives $\head{\grl}=\addlt{\fc}{x_{i}}{s_{i}}\in N$, as
desired.
\begin{itemize}
\item In order for $\grl$ to be in $\grded$, it is required that $V'(\bneg{\rl'})\cap M=\emptyset$.
Because $V$ is satisfying for $\rl$, and negation in $\rl$ is only
applied to lower strata, we have $V(\bneg{\rl})\cap\strat{D_{i}}{k-1}=\emptyset$.
Moreover, since a relation is computed in only one stratum of $\deduc{\ded}$,
we overall have $V(\bneg{\rl})\cap D_{i}=\emptyset$. Then by Claim~\ref{claim:not-in-duc-not-in-M}
we have $\addlt{V(\bneg{\rl})}{x_{i}}{s_{i}}\cap M=\emptyset$. Hence,
\[
V'(\bneg{\rl'})\cap M=\emptyset.
\]

\item Now we show that $\bpos{\grl}\subseteq N$. Because $V$ is satisfying
for $\rl$, we have $V(\bpos{\rl})\subseteq A_{l-1}$, and by applying
the inner induction hypothesis we have $\addlt{V(\bpos{\rl})}{x_{i}}{s_{i}}\subseteq N$.
Therefore, $\bpos{\grl}=V'(\bpos{\rl'})\subseteq N$.
\end{itemize}
\end{proof}

\tline

\begin{claim}\label{claim:not-in-duc-not-in-M}Let $i$ be a transition
of $\run$. Let $I$ be a set of facts over $\schof{\ded}$. If $I\cap D_{i}=\emptyset$
then $\addlt I{x_{i}}{s_{i}}\cap M=\emptyset$.\end{claim}

\begin{proof} If a fact $\fc\in M$ is over schema $\toloct{\schof{\ded}}$
and has location specifier $x_{i}$ and timestamp $s_{i}$ then $\fc\in\sliceduc i$
because \emph{(i)} for any transition $j$ there are no facts over
$\toloct{\schof{\ded}}$ in $\slicecaus j$, $\slicefin j$ or $\slicesnd j$;
\emph{(ii)} we only add facts with location specifier $x_{i}$ to
$\sliceduc j$ if $j$ is a transition of node $x_{i}$; and, \emph{(iii)}
for every transition $j$ of node $x_{i}$, if $i\neq j$ then $\locR j\neq s_{i}$.

Hence, it suffices to show $\addlt I{x_{i}}{s_{i}}\cap\sliceduc i=\emptyset$.
But this is immediate from $I\cap D_{i}=\emptyset$ because $\sliceduc i$
equals $\addlt{D_{i}}{x_{i}}{s_{i}}$ by definition.\end{proof}

\tline

\begin{claim}\label{claim:ind-in-stable}Let $j$ be a transition
of $\run$. Let $D_{j}$ be the output of $\deduc{\ded}$ during transition
$j$. Suppose $\sliceduc j\subseteq N$. We have $\addlt{\induc{\ded}\mstep{D_{j}}}{x_{j}}{s_{j}+1}\subseteq N$.\end{claim}

\begin{proof}

Let $\fc\in\induc{\ded}\mstep{D_{j}}$. Let $\rl\in\induc{\ded}$
and $V$ respectively be a rule and valuation that have derived $\fc$.
Let $\rl'$ be the rule in $\pure{\ded}$ that is obtained after applying
transformation (\ref{eq:pure-ind}) to $\rl$. Thus, besides the additional
location variable, the rule $\rl'$ has two timestamp variables, one
in the body and one in the head. Moreover, the body contains an additional
positive $\timesucc$-atom. Let $V'$ be $V$ extended to assign $x_{j}$
to the location variable, and to assign timestamps $s_{j}$ and $s_{j}+1$
to the body and head timestamp variables respectively. Let $\grl$
be the \emph{positive} ground rule obtained from $\rl'$ by applying
valuation $V'$ and by subsequently removing all negative (ground)
body atoms. We show that $\grl\in\grded$ and that its body is in
$N$, so that $\grl$ derives $\head{\grl}=\addlt{\fc}{x_{j}}{s_{j}+1}\in N$,
as desired.
\begin{itemize}
\item For $\grl$ to be in $\grded$, we require $V'(\bneg{\rl'})\cap M=\emptyset$.
Since $V'(\bneg{\rl'})=\addlt{V(\bneg{\rl})}{x_{j}}{s_{j}}$, it suffices
to show $\addlt{V(\bneg{\rl})}{x_{j}}{s_{j}}\cap M=\emptyset$. Because
$V$ is satisfying for $\rl$, we have $V(\bneg{\rl})\cap D_{j}=\emptyset$.
Then, by Claim~\ref{claim:not-in-duc-not-in-M} we have $\addlt{V(\bneg{\rl})}{x_{j}}{s_{j}}\cap M=\emptyset$.
\item Now we show $V'(\bpos{\rl'})\subseteq N$. The set $V'(\bpos{\rl'})$
consists of the facts $\addlt{V(\bpos{\rl})}{x_{j}}{s_{j}}$ and the
fact $\timesucc(s_{j},s_{j}+1)$. The latter fact is in $\decl H$
and thus in $N$. For the other facts, because $V$ is satisfying
for $\rl$, we have $V(\bpos{\rl})\subseteq D_{j}$ and thus $\addlt{V(\bpos{\rl})}{x_{j}}{s_{j}}\subseteq\addlt{D_{j}}{x_{j}}{s_{j}}=\sliceduc j$.
And by using the given assumption $\sliceduc j\subseteq N$, we obtain
the inclusion in $N$.
\end{itemize}
\end{proof}

\tline

\begin{claim}\label{claim:cand-in-stable}Let $i$ be a transition
of $\run$. Suppose $\addlt{D_{i}}{x_{i}}{s_{i}}\subseteq N$. For
each $R(y,\bar{a})\in\mesg i$ and timestamp $t\in\Nat$ with $(y,t)\not\caus(x_{i},s_{i})$
we have 
\[
\cand_{R}(x_{i},s_{i},y,t,\bar{a})\in N.
\]
\end{claim}

\begin{proof}

By definition of $\mesg i$, we have $R(y,\bar{a})\in\async{\ded}\mstep{D_{i}}$.
Let $\rl\in\async{\ded}$ and $V$ be a rule and valuation that have
produced $R(y,\bar{a})$. Let $\rl'\in\ded$ be the original asynchronous
rule on which $\rl$ is based. Let $\rl''\in\pure{\ded}$ be the rule
obtained from $\rl'$ by applying transformation (\ref{eq:cand}).
Let $V''$ be valuation $V$ extended to assign $x_{i}$ and $s_{i}$
to respectively the sender location and sender timestamp of $\rl''$,
and to assign $y$ and $t$ respectively to the addressee location
and addressee arrival timestamp. Let $\grl$ denote the \emph{positive}
ground rule that is obtained from $\rl''$ by applying valuation $V''$
and by subsequently removing all negative (ground) body atoms. We
show that $\grl\in\grded$ and that its body is in $N$, so that $\grl$
derives $\head{\grl}=\cand_{R}(x_{i},s_{i},y,t,\bar{a})\in N$, as
desired.
\begin{itemize}
\item For $\grl$ to be in $\grded$, we require $V''(\bneg{\rl''})\cap M=\emptyset$.
By construction of $\rl''$, the set $V''(\bneg{\rl''})$ consists
of the facts $\addlt{V(\bneg{\rl})}{x_{i}}{s_{i}}$ and the fact $\before(y,t,x_{i},s_{i})$.
First, because $V$ is satisfying for $\rl$, we have $V(\bneg{\rl})\cap D_{i}=\emptyset$,
and thus $\addlt{V(\bneg{\rl})}{x_{i}}{s_{i}}\cap M=\emptyset$ by
Claim~\ref{claim:not-in-duc-not-in-M}. Moreover, we are given that
$(y,t)\not\caus(x_{i},s_{i})$, and thus we have not added $\before(y,t,x_{i},s_{i})$
to $\slicecaus i$, and by extension also not to $M$ (since $\slicecaus i$
is the only part of $M$ where we add $\before$-facts with last two
components $x_{i}$ and $s_{i}$). Thus overall $V''(\bneg{\rl''})\cap M=\emptyset$,
as desired.
\item Now we show $V''(\bpos{\rl''})\subseteq N$. By construction of $\rl''$,
the set $V''(\bpos{\rl''})$ consists of the facts $\addlt{V(\bpos{\rl})}{x_{i}}{s_{i}}$,
$\relall(y)$ and $\reltime(t)$. First, we immediately have $\reltime(t)\in\decl H\subseteq N$.
Also, by definition of $\mesg i$, $y$ is a valid addressee and thus
$\relall(y)\in\decl H\subseteq N$. Finally, because $V$ is satisfying
for $\rl$, we have $V(\bpos{\rl})\subseteq D_{i}$. Thus $\addlt{V(\bpos{\rl})}{x_{i}}{s_{i}}\subseteq\addlt{D_{i}}{x_{i}}{s_{i}}$,
and we are given that $\addlt{D_{i}}{x_{i}}{s_{i}}\subseteq N$. Thus
overall $V''(\bpos{\rl''})\subseteq N$.
\end{itemize}
\end{proof}

\subsection{Inclusion $N\subseteq M$}

\label{sub:N-in-M}

\newcommand{\prevN}{N_{l-1}}

\newcommand{\diffN}{N_{l}\setminus N_{l-1}}

In this section we show that $N\subseteq M$. By definition, $N=\grded(\decl H)$.
 Following the semantics of positive $\datalogneg$ programs in Section
\ref{sub:positive-and-semi-positive}, we can view $N$ as a fixpoint,
i.e., $N=\bigcup_{l\in\Nat}N_{l}$, where $N_{0}=\decl H$, and for
each $l\geq1$ the set $N_{l}$ is obtained by applying the immediate
consequence operator of $\grded$ to $N_{l-1}$. This implies $N_{l-1}\subseteq N_{l}$
for each $l\geq1$. We show by induction on $l=0$, $1$, $\ldots$,
that $N_{l}\subseteq M$. For the base case ($l=0$), we immediately
have $N_{0}=\decl H\subseteq M$. For the induction hypothesis, we
assume for some $l\geq1$ that $N_{l-1}\subseteq M$. For the inductive
step, we show that $N_{l}\subseteq N$. Specifically, we divide the
facts of $N_{l}\setminus N_{l-1}$ into groups based on their predicate,
and for each group we show inclusion in $M$. As for terminology,
we call a ground rule $\grl\in G$ \emph{active} on $\prevN$ if
$\bpos{\grl}\subseteq\prevN$. The numbered claims we will refer to
can be found in Section~\ref{sub:run-to-model-second-dir--claims}.

\subsubsection{Causality}

Let $\before(x,s,y,t)\in\diffN$. It is sufficient to show that $(x,s)\caus(y,t)$
because then $\before(x,s,y,t)\in\slicecaus i\subseteq M$ where $i=\globR{y,t}$.
We have the following cases:

\paragraph*{Local edge}

The $\before$-fact was derived by a ground rule in $\grded$ of the
form (\ref{eq:before-step}) (local edge). This implies $x=y$ and
$t=s+1$. Then $(x,s)\caus(y,t)$ by definition of $\caus$.

\paragraph*{Message edge}

The $\before$-fact was derived by a ground rule in $\grded$ of the
form (\ref{eq:before-send}) (message edge):
\[
\before(x,s,y,t)\gets\chosen_{R}(x,s,y,t,\bar{a}).
\]
Since this rule is active on $\prevN$, we have $\chosen_{R}(x,s,y,t,\bar{a})\in\prevN$.
By applying the induction hypothesis, we have $\chosen_{R}(x,s,y,t,\bar{a})\in M$.
Denoting $j=\globR{x,s}$, the set $\slicesnd j$ is the only part
of $M$ where we could have added this fact. This implies that $x$
during its step $s$ sends a message to $y$, and this message arrives
at local step $t$ of $y$. Then $(x,s)\caus(y,t)$ by definition
of $\caus$.

\paragraph*{Transitive edge}

The $\before$-fact was derived by a ground rule in $\grded$ of the
form (\ref{eq:before-tr}) (transitive edge):
\[
\before(x,s,y,t)\gets\before(x,s,z,u),\,\before(z,u,y,t).
\]
Since this rule is active on $\prevN$, its body facts are in $\prevN$.
By applying the induction hypothesis, we have $\before(x,s,z,u)\in M$
and $\before(z,u,y,t)\in M$. The only places we could have added
these facts to $M$ are in the sets $\slicecaus j$ and $\slicecaus k$
respectively, where $j=\globR{z,u}$ and $k=\globR{y,t}$. By construction
of the sets $\slicecaus j$ and $\slicecaus k$ we respectively have
that $(x,s)\caus(z,u)$ and $(z,u)\caus(y,t)$, and thus by transitivity
$(x,s)\caus(y,t)$, as desired.

\subsubsection{Finite Messages}

\paragraph*{Senders}

Let $\hassender(x,s,y,t)\in\diffN$. This fact can only have been
derived by a ground rule in $\grded$ of the form (\ref{eq:has-sender}):
\[
\hassender(x,s,y,t)\gets\chosen_{R}(y,t,x,s,\bar{a}).
\]
Since this rule is active on $\prevN$, we have $\chosen_{R}(y,t,x,s,\bar{a})\in\prevN$.
By applying the induction hypothesis, we have $\chosen_{R}(y,t,x,s,\bar{a})\in M$.
We can only have added this fact in the set $\slicesnd i$ with $i=\globR{y,t}$.
This means that $y$ during its step $t$ sends a message $R(\bar{a})$
to $x$, and this message arrives during step $s$ of $x$. Hence,
denoting $j=\globR{x,s}$, we have $(y,t)\in\senders j$ (with $\senders j$
as defined in Section~\ref{sub:run-to-model-construction}). Thus
we have added the fact $\hassender(x,s,y,t)\in\slicefin j\subseteq M$,
as desired.

\paragraph*{Comparison of timestamps}

Let $\issmaller(x,s,y,t)\in\diffN$. This fact can only have been
derived by a ground rule in $\grded$ of the form (\ref{eq:is-smaller}):
\begin{eqnarray*}
\issmaller(x,s,y,t) & \gets & \hassender(x,s,y,t),\,\hassender(x,s,y,t'),\\
 &  & t<t'.
\end{eqnarray*}
Since this rule is active on $\prevN$, its body facts are in $\prevN$.
By applying the induction hypothesis, we have $\hassender(x,s,y,t)\in M$
and $\hassender(x,s,y,t')\in M$. The only part of $M$ where we could
have added these facts is the set $\slicefin i$ with $i=\globR{x,s}$.
By construction of the set $\slicefin i$, this implies that $(y,t)\in\senders i$
and $(y,t')\in\senders i$. Because $(t<t')\in\prevN$, we more specifically
know that $(t<t')\in\decl H$, which implies $t<t'$. Thus we have
added $\issmaller(x,s,y,t)\in\slicefin i$, as desired.

\paragraph*{Maximum timestamp}

Let $\hasmax(x,s,y)\in\diffN$. This fact can only have been derived
by a ground rule in $\grded$ of the form (\ref{eq:has-max}):
\[
\hasmax(x,s,y)\gets\hassender(x,s,y,t).
\]
Since this rule is active on $\prevN$, we have $\hassender(x,s,y,t)\in\prevN$.
By applying the induction hypothesis, we have $\hassender(x,s,y,t)\in M$.
The only part of $M$ where we could have added this fact, is the
set $\slicefin i$ with $i=\globR{x,s}$. Thus $(y,t)\in\senders i$,
and $y$ is a sender-node mentioned in $\senders i$. Hence, we have
added $\hasmax(x,s,y)\in\slicefin i\subseteq M$, as desired.

\paragraph*{Receive infinite}

Let $\rcvinf(x,s)\in\diffN$. This fact can only have been derived
by a ground rule in $\grded$ of the form (\ref{eq:rcv-inf}):
\[
\rcvinf(x,s)\gets\hassender(x,s,y,t).
\]
Since this rule is active on $\prevN$, we have $\hassender(x,s,y,t)\in\prevN$.
By applying the induction hypothesis, we have $\hassender(x,s,y,t)\in M$.
The only part of $M$ where we could have added this fact, is the
set $\slicefin i$ with $i=\globR{x,s}$. Thus $(y,t)\in\senders i$.
Moreover, because the rule (\ref{eq:rcv-inf}) contains a negative
$\hasmax$-atom in the body, and the above ground rule is in $\grded$,
it must be that $\hasmax(x,s,y)\notin M$, and thus $\hasmax(x,s,y)\notin\slicefin i$.
 But since $y$ is a sender-node mentioned in $\senders i$, the
absence of $\hasmax(x,s,y)$ from $\slicefin i$ is impossible. Therefore
this case can not occur.

\subsubsection{Regular Facts}

Let $R(x,s,\bar{a})\in\proj{(\diffN)}{\toloct{\schof{\ded}}}$. The
fact $R(x,s,\bar{a})$ has been derived by a ground rule $\grl\in\grded$
that is active on $\prevN$. Because $\grl\in\grded$, there is a
rule $\rl\in\pure{\ded}$ and valuation $V$ such that $\grl$ is
obtained from $\rl$ by applying $V$ and by subsequently removing
the negative (ground) body atoms, and such that $V(\bneg{\rl})\cap M=\emptyset$.
We have the following cases:

\paragraph*{Deductive}

Rule $\rl$ is of the form (\ref{eq:pure-duc}). Let $\rl'\in\deduc{\ded}$
be the original deductive rule corresponding to $\rl$. By construction
of $\rl$ out of $\rl'$, we can apply valuation $V$ to $\rl'$ as
well. Denote $i=\globR{x,s}$. We will show now that $V$ is satisfying
for $\rl'$ during transition $i$, which causes $V(\head{\rl'})=R(\bar{a})\in D_{i}$
to be derived, and we obtain as desired: 
\[
R(x,s,\bar{a})\in\addlt{D_{i}}xs=\addlt{D_{i}}{x_{i}}{s_{i}}=\sliceduc i\subseteq M.
\]
By definition of syntactic stratification, relations mentioned in
$\bpos{\rl'}$ are never computed in a stratum higher than $R$, and
relations mentioned in $\bneg{\rl'}$ are computed in a strictly lower
stratum than $R$. Thus, it is sufficient to show that $V(\bpos{\rl'})\subseteq D_{i}$
and $V(\bneg{\rl'})\cap D_{i}=\emptyset$. 

First we show $V(\bpos{\rl'})\subseteq D_{i}$. Because $\rl$ is
of the form (\ref{eq:pure-duc}), all facts in $V(\bpos{\rl})$ are
over $\toloct{\schof{\ded}}$ and have location specifier $x$ and
timestamp $s$. Moreover, since $\grl$ is active on $\prevN$, we
have $\bpos{\grl}=V(\bpos{\rl})\subseteq\prevN$. By applying the
induction hypothesis, we have $V(\bpos{\rl})\subseteq M$, and thus
$\droplt{V(\bpos{\rl})}\subseteq D_{i}$ by Claim~\ref{claim:in-M-in-duc}.
We thus obtain $V(\bpos{\rl'})\subseteq D_{i}$ since $\droplt{V(\bpos{\rl})}=V(\bpos{\rl'})$.

Next we show $V(\bneg{\rl'})\cap D_{i}=\emptyset$. Because $\rl$
is of the form (\ref{eq:pure-duc}), all facts in $V(\bneg{\rl})$
are over $\toloct{\schof{\ded}}$ and have location specifier $x$
and timestamp $s$. Moreover, by choice of $\rl$ and $V$, we have
$V(\bneg{\rl})\cap M=\emptyset$, and thus $\droplt{V(\bneg{\rl})}\cap D_{i}=\emptyset$
by Claim~\ref{claim:not-in-M-not-in-duc}. We thus obtain $V(\bneg{\rl'})\cap D_{i}=\emptyset$
since $\droplt{V(\bneg{\rl})}=V(\bneg{\rl'})$.

\paragraph*{Inductive}

Rule $\rl$ is of the form (\ref{eq:pure-ind}). Let $\rl'\in\induc{\ded}$
be the rule corresponding to $\rl$. First, $\grl$ contains in its
body a fact of the form $\timesucc(r,s)$. Since $\grl$ is active
on $\prevN$, we have $\timesucc(r,s)\in\prevN$ and more specifically,
$\timesucc(r,s)\in\decl H$. This implies that $s=r+1$. Denote $i=\globR{x,r}$
and $j=\globR{x,s}$. Since $s=r+1$, there are no transitions of
node $x$ between $i$ and $j$. By the relationship between $\rl$
and $\rl'$, we can apply $V$ to $\rl'$, and we will now show that
$V$ is satisfying for $\rl'$ during transition $i$. This results
in $V(\head{\rl'})=R(\bar{a})\in\induc{\ded}\mstep{D_{i}}\subseteq\cnfs_{i+1}(x)$,
and since $\cnfs_{i+1}(x)=\cnfs_{j}(x)\subseteq D_{j}$, we obtain
$R(x,s,\bar{a})\in\addlt{D_{j}}xs=\sliceduc j\subseteq M$, as desired.

First we show $V(\bpos{\rl'})\subseteq D_{i}$. Denote $I=\proj{V(\bpos{\rl})}{\toloct{\schof{\ded}}}$,
which allows us to exclude the extra $\timesucc$-fact in the body.
All facts in $I$ have location specifier $x$ and timestamp $r$.
Because $\grl$ is active on $\prevN$, we have $I\subseteq\bpos{\grl}\subseteq\prevN$,
and by applying the induction hypothesis, we have $I\subseteq M$.
Thus $\droplt I\subseteq D_{i}$ by Claim~\ref{claim:in-M-in-duc}.
Hence, $V(\bpos{\rl'})=\droplt I\subseteq D_{i}$. 

Secondly, showing that $V(\bneg{\rl'})\cap D_{i}=\emptyset$ is like
in the previous case, where $\rl$ is deductive.

\paragraph*{Delivery}

Rule $\rl$ is of the form (\ref{eq:deliv}). Then $\grl$ concretely
looks as follows, where $(y,t)\in\nwnat$:
\[
R(x,s,\bar{a})\gets\chosen_{R}(y,t,x,s,\bar{a}).
\]
Since $\grl$ is active on $\prevN$, we have $\chosen_{R}(y,t,x,s,\bar{a})\in\prevN$,
and by applying the induction hypothesis, we have $\chosen_{R}(y,t,x,s,\bar{a})\in M$.
The only part of $M$ where we could have added this fact, is $\slicesnd i$
with $i=\globR{y,t}$. This implies that $x$ will receive $R(\bar{a})$
during its local step $s$, thus during transition $j=\globR{x,s}$.
Then, by the operational semantics, we have $R(\bar{a})\in\untag{m_{j}}\subseteq D_{j}$.
Hence, $R(x,s,\bar{a})\in\addlt{D_{j}}xs=\sliceduc j\subseteq M$.

\subsubsection{Sending}

For a transition $i$ of $\run$, let $D_{i}$ denote the output of
subprogram $\deduc{\ded}$ during transition $i$.

\paragraph*{Candidates}

Let $\cand_{R}(x,s,y,t,\bar{a})\in\diffN$. The fact $\cand_{R}(x,s,y,t,\bar{a})$
is derived by a ground rule $\grl\in\grded$ of the form (\ref{eq:cand})
that is active on $\prevN$. Because $\grl\in\grded$, there is a
rule $\rl\in\pure{\ded}$ and a valuation $V$ such that $\grl$ is
obtained from $\rl$ by applying valuation $V$ and by subsequently
removing the negative (ground) body atoms, and so that $V(\bneg{\rl})\cap M=\emptyset$.
Denote $i=\globR{x,s}$. It is sufficient to show that $R(y,\bar{a})\in\mesg i$
and $(y,t)\not\caus(x,s)$, because then $\cand_{R}(x,s,y,t,\bar{a})\in\slicesnd i\subseteq M$,
as desired. 

First, we show $(y,t)\not\caus(x,s)$. Because there is a negative
$\before$-atom in $\rl$, the existence of $\grl$ in $\grded$ implies
that $\before(y,t,x,s)\notin M$. Hence, $\before(y,t,x,s)\notin\slicecaus i$.
Then by construction of $\slicecaus i$ we obtain $(y,t)\not\caus(x,s)$.

Secondly, we show $R(y,\bar{a})\in\mesg i$. Let $\rl'\in\ded$ be
the original asynchronous rule on which $\rl$ is based. Let $\rl''\in\async{\ded}$
be the rule corresponding to $\rl'$. It follows from the constructions
of $\rl$ out of $\rl'$ and $\rl''$ out of $\rl'$ that valuation
$V$ can be applied to $\rl''$. Note, $V(\head{\rl''})=R(y,\bar{a})$.
We show that $V$ is satisfying for $\rl''$ during transition $i$
on $D_{i}$, which gives $R(y,\bar{a})\in\async{\ded}\mstep{D_{i}}$.
Moreover, the body of $\grl$ contains the fact $\relall(y)\in\decl H$,
and thus $y\in\nw$, making $y$ a valid addressee. Hence, $R(y,\bar{a})\in\mesg i$,
as desired.

We have to show $V(\bpos{\rl''})\subseteq D_{i}$ and $V(\bneg{\rl''})\cap D_{i}=\emptyset$.
Abbreviate $I_{1}=\proj{V(\bpos{\rl})}{\toloct{\schof{\ded}}}$ and
$I_{2}=\proj{V(\bneg{\rl})}{\toloct{\schof{\ded}}}$. Note, $\droplt{I_{1}}=V(\bpos{\rl''})$
and $\droplt{I_{2}}=V(\bneg{\rl''})$. All facts in $I_{1}\cup I_{2}$
have location specifier $x$ and timestamp $s$.
\begin{itemize}
\item Because $\grl$ is active on $\prevN$, we have $I_{1}\subseteq\bpos{\grl}\subseteq\prevN$,
and thus $I_{1}\subseteq M$ by the induction hypothesis. Then $V(\bpos{\rl''})=\droplt{I_{1}}\subseteq D_{i}$
by Claim~\ref{claim:in-M-in-duc}.
\item By choice of $\rl$ and $V$, we have $I_{2}\cap M=\emptyset$. Then
$\droplt{I_{2}}\cap D_{i}=\emptyset$ by Claim~\ref{claim:not-in-M-not-in-duc},
giving $V(\bneg{\rl''})\cap D_{i}=\emptyset$.
\end{itemize}

\paragraph*{Chosen}

Let $\chosen_{R}(x,s,y,t,\bar{a})\in\diffN$. This fact is derived
by a ground rule $\grl$ in $\grded$ of the form (\ref{eq:chosen}):
\[
\chosen_{R}(x,s,y,t,\bar{a})\gets\cand_{R}(x,s,y,t,\bar{a}).
\]
Denote $i=\globR{x,s}$. We show that $R(y,\bar{a})\in\mesg i$ and
that $t$ is the actual arrival timestamp of this message at $y$.
Then $\chosen_{R}(x,s,y,t,\bar{a})\in\slicesnd i\subseteq M$, as
desired.

First, since $\grl$ is active on $\prevN$, we have $\cand_{R}(x,s,y,t,\bar{a})\in\prevN$,
and thus $\cand_{R}(x,s,y,t,\bar{a})\in M$ by the induction hypothesis.
The set $\slicesnd i$ is the only part of $M$ where we could have
added this fact, which implies $R(y,\bar{a})\in\mesg i$ and $(y,t)\not\caus(x,s)$.

We are left to show that $t$ is the actual arrival timestamp of the
message. Because $\grl\in\grded$, there is a rule $\rl\in\pure{\ded}$
and valuation $V$ such that $\grl$ is obtained from $\rl$ by applying
$V$ and by subsequently removing the negative (ground) body atoms,
and so that $V(\bneg{\rl})\cap M=\emptyset$. Now, because rule $\rl$
contains a negative $\other_{R}$-atom in its body, we have $\other_{R}(x,s,y,t,\bar{a})\notin M$
and thus $\other_{R}(x,s,y,t,\bar{a})\notin\slicesnd i$. Since $R(y,\bar{a})\in\mesg i$
and $(y,t)\not\caus(x,s)$ (see above), the absence of this $\other_{R}$-fact
from $\slicesnd i$ can only be explained by the following: $t=\locR j$
with $j=\arr(i,y,R(\bar{a}))$, as desired.

\paragraph*{Other}

Let $\other_{R}(x,s,y,t,\bar{a})\in\diffN$. This fact is derived
by a ground rule $\grl$ of the form (\ref{eq:other}):
\begin{eqnarray*}
\other_{R}(x,s,y,t,\bar{a}) & \gets & \cand_{R}(x,s,y,t,\bar{a}),\,\chosen_{R}(x,s,y,t',\bar{a}),\\
 &  & t\neq t'.
\end{eqnarray*}
We have $\cand_{R}(x,s,y,t,\bar{a})\in\prevN$ and $\chosen_{R}(x,s,y,t',\bar{a})\in\prevN$
since $\grl$ is active on $\prevN$, and these facts are thus also
in $M$ by the induction hypothesis. Denote $i=\globR{x,s}$. The
only part of $M$ where we could have added these $\cand_{R}$- and
$\chosen_{R}$-facts to $M$, is the set $\slicesnd i$. First, $\cand_{R}(x,s,y,t,\bar{a})\in\slicesnd i$
implies that $R(y,\bar{a})\in\mesg i$ and $(y,t)\not\caus(x,s)$.
Second, $\chosen_{R}(x,s,y,t',\bar{a})\in\slicesnd i$ implies that
$t'$ is the real arrival timestamp of the message $R(\bar{a})$ at
$y$. Finally, since $\grl$ is active, we have $(t\neq t')\in\decl H$,
and thus $t\neq t'$. Therefore we have added $\other_{R}(x,s,y,t,\bar{a})$
to $\slicesnd i\subseteq M$, as desired.

\subsubsection{Subclaims}

\label{sub:run-to-model-second-dir--claims}

\begin{claim}\label{claim:in-M-in-duc}Let $I$ be a set of facts
over $\toloct{\schof{\ded}}$, all having the same location specifier
$x\in\nw$ and timestamp $s\in\Nat$. Denote $i=\globR{x,s}$. If
$I\subseteq M$ then $\droplt I\subseteq D_{i}$, where $D_{i}$ denotes
the output of subprogram $\deduc{\ded}$ during transition $i$ of
$\run$.\end{claim}

\begin{proof}The only part of $M$ where we add facts over $\toloct{\schof{\ded}}$
with location specifier $x$ and timestamp $s$ is $\sliceduc i$.
Hence $I\subseteq\sliceduc i=\addlt{D_{i}}xs$ and thus $\droplt I\subseteq D_{i}$.\end{proof}

\tline

\begin{claim}\label{claim:not-in-M-not-in-duc}Let $I$ be a set
of facts over $\toloct{\schof{\ded}}$, all having the same location
specifier $x\in\nw$ and timestamp $s\in\Nat$. Denote $i=\globR{x,s}$.
If $I\cap M=\emptyset$ then $\droplt I\cap D_{i}=\emptyset$, where
$D_{i}$ denotes the output of subprogram $\deduc{\ded}$ during transition
$i$ of $\run$.\end{claim}

\begin{proof}First, $I\cap M=\emptyset$ implies $I\cap\sliceduc i=\emptyset$
because $\sliceduc i\subseteq M$. And since $\sliceduc i=\addlt{D_{i}}xs$,
we have $I\cap\addlt{D_{i}}xs=\emptyset$. Finally, since the facts
in $I\cup\addlt{D_{i}}xs$ all have the same location specifier $x$
and timestamp $s$, we obtain $\droplt I\cap D_{i}=\emptyset$.\end{proof}

\section{Model to Run: Proof Details}

\label{app:proof-dir-2}

Consider the definitions and notations from Section~\ref{sub:model-to-run}.
In this section we show that $\run$ is a run of $\ded$ on input
$H$, and that $\trace{\run}=\proj M{\toloct{\schof{\ded}}}$. We
do this in several parts, where each part is placed in its own subsection: 
\begin{itemize}
\item in Section~\ref{sub:model-to-run--valid-start} we show $\cnf_{0}=\cnfstart{\ded}H$;
\item in Section~\ref{sub:model-to-run--valid-transition} we show that
every transition of $\run$ is valid; and,
\item in Section~\ref{sub:model-to-run--trace} we show $\trace{\run}=\proj M{\toloct{\schof{\ded}}}$.
\end{itemize}
Before we start, the next subsection gives definitions and notations.
The numbered claims we will refer to can be found in Section~\ref{sub:model-to-run--claims}.

\subsection{Definitions and Notations}

\label{sub:model-to-run-defs-notations}

\newcommand{\up}[1]{#1^{\blacktriangle}}

\newcommand{\duc}[2]{#1^{\mathrm{duc},#2}}

\newcommand{\ind}[1]{#1^{\mathrm{ind}}}

\newcommand{\deliv}[1]{#1^{\mathrm{deliv}}}

Using notations of Section~\ref{sub:stable-model-semantics},
let $\grded$ be the ground program $\grp MCI$ where $C=\pure{\ded}$
and $I=\decl H$. By definition of $M$ as a stable model, we have
$M=\grded(I)$.

Let $\rl\in\pure{\ded}$ be a rule having its head atom over $\toloct{\schof{\ded}}$.
From the construction of $\pure{\ded}$, we know that $\rl$ belongs
to exactly one of the following three cases:
\begin{itemize}
\item $\rl$ is of the form (\ref{eq:pure-duc}), i.e., \emph{deductive},
recognizable as a rule in which only atoms over $\toloct{\schof{\ded}}$
are used, and in which the location and timestamp variable in the
head are the same as in the body;
\item $\rl$ is of the form (\ref{eq:pure-ind}), i.e., \emph{inductive},
recognizable as a rule with a head atom over $\toloct{\schof{\ded}}$
and a $\timesucc$-atom in the body;
\item $\rl$ is of the form (\ref{eq:deliv}), i.e., a \emph{delivery},
recognizable as a rule with a head atom over $\toloct{\schof{\ded}}$
and a $\chosen_{R}$-fact in the body (with $R$ the head-predicate).
\end{itemize}
The same classification of deductive, inductive and delivery rules
can also be applied to the (positive) ground rules in $\grded$ that
have a ground head atom over $\toloct{\schof{\ded}}$. 

Recall from the general remarks at the beginning of the appendix that
we are working with a fixed (but arbitrary) syntactic stratification
for the deductive rules. Stratum numbers start at $1$. If $\rl\in\pure{\ded}$
is deductive, we can uniquely identify its stratum number as the stratum
number of the original deductive rule in $\ded$ on which $\rl$ is
based. Similarly, for deductive ground rules, we can also uniquely
identify the stratum number as the stratum number of a corresponding
non-ground rule in $\pure{\ded}$.%
\footnote{We say \emph{a} rather than \emph{the} corresponding rule because
there could be more than one. Indeed, multiple original deductive
rules in $\pure{\ded}$ could be mapped to the same positive ground
rule after applying a valuation and removing their negative ground
body atoms. But in any case, these non-ground rules will have the
same head predicate. Hence, they have the same stratum.%
}

We call a ground rule $\grl\in\grded$ \emph{active} if $\bpos{\grl}\subseteq M$,
which implies that $\head{\grl}\in M$ because $M$ is stable. Now
we define the following subsets of $M$:
\begin{itemize}
\item $\duc Mk$: the head facts of all active deductive rules in $\grded$
with stratum number less than or equal to $k$;
\item $\ind M$: the head facts of all active inductive rules in $\grded$;
\item $\deliv M$: the head facts of all active delivery rules in $\grded$.
\end{itemize}
This allows us to classify the facts in $\proj M{\toloct{\schof{\ded}}}$
as being derived in a deductive manner, an inductive manner or being
message deliveries. We also define:
\[
\up M=\proj M{\toloct{\edb{\ded}}}\cup\ind M\cup\deliv M.
\]
For $(x,s)\in\nwnat$, we write $\projlt Ixs$ to abbreviate $\projlt{(\proj I{\toloct{\schof{\ded}}})}xs$.
So intuitively, when we select the facts with location specifier $x$
and timestamp $s$, we are only interested in facts that provide these
two components, which are the facts over $\toloct{\schof{\ded}}$.

Intuitively, for $i\in\Nat$, the set $\projlt{(\up M)}{x_{i}}{s_{i}}$
is the input for the deductive rules during local step $s_{i}$ of
node $x_{i}$, consisting of $\romI$ the \emph{edb}-facts; $\romII$
the facts derived by inductive rules during a previous step (if any)
of $x_{i}$; and, $\romIII$ the delivered messages. The deductive
rules then complete this information by deriving some new facts, that
are visible within step $s_{i}$ of $x_{i}$.

For a transition number $i$ of $\run$, $\romI$ we denote the source-configuration
of transition $i$ as $\cnf_{i}=(\cnfs_{i},\cnfb_{i})$; $\romII$
we denote the set of (tagged) messages delivered in transition $i$
as $m_{i}$; and, $\romIII$ we denote $D_{i}=\deduc{\ded}(\cnfs_{i}(x_{i})\cup\untag{m_{i}})$.
For a number $k\in\Nat$, we write $\strat{D_{i}}k$ to denote the
set of facts obtained by adding to $\cnfs_{i}(x_{i})\cup\untag{m_{i}}$
all facts derived in stratum $1$ up to stratum $k$ during the computation
of $D_{i}$. To mirror this notation, we write $\strat Mk$ to denote
the set $\up M\cup\duc Mk$. For uniformity in the proofs, we will
consider the case $k=0$, which is an invalid stratum number, and
this gives $\strat{D_{i}}0=\cnfs_{i}(x_{i})\cup\untag{m_{i}}$ and
$\strat M0=\up M$.

\subsection{Valid Start}

\label{sub:model-to-run--valid-start}

We show that $\cnf_{0}=\cnfstart{\ded}H$. Denote $\cnf_{0}=(\cnfs_{0},\cnfb_{0})$.
Let $x\in\nw$. First we show $\cnfs_{0}(x)=H(x)$. By definition,
\[
\cnfs_{0}(x)=\droplt{\left(\shprojlt{\proj M{\toloct{\edb{\ded}}}}xs\cup\projlt{\Mind}xs\right)}
\]
with $s=\locM{0,x}$. Note, $s=0$ because no elements of $\nwnat$
with first component $x$ have an ordinal strictly less than $0$
in the total order $\totM$. Now, there can be no ground inductive
rules in $\grded$ that derive facts with head timestamp $0$ because
it follows from the construction of $\decl H$ that the second component
of a $\timesucc$-fact is always strictly larger than $0$. Therefore
$\projlt{\Mind}xs=\emptyset$, and thus $\cnfs_{0}(x)=\droplt{\left(\shprojlt{\proj M{\toloct{\edb{\ded}}}}xs\right)}$.
Then by Claim~\ref{claim:edb-in-M} we have $\cnfs_{0}(x)=\droplt{(\addlt{H(x)}xs)}=H(x)$,
as desired.

Now we show $\cnfb_{0}(x)=\emptyset$. By definition, $\cnfb_{0}(x)$
is
\[
\begin{array}{ll}
\{\pair{\globM{y,t}}{\, R(\bar{a})}\mid
    & \exists u:\,\chosen_{R}(y,t,x,u,\bar{a})\in M,\\&\globM{y,t}<0\leq\globM{x,u}\}.
\end{array}
\]
By definition of function $\globM{\cdot}$, all facts of the form
$\chosen_{R}(y,t,x,u,\bar{a})\in M$ satisfy $\globM{y,t}\geq0$.
Hence, $\cnfb_{0}(x)=\emptyset$.

We conclude that $\cnf_{0}=\cnfstart{\ded}H$.

\subsection{Valid Transition}

\label{sub:model-to-run--valid-transition}

Let $i\in\Nat$. We show that $(\cnf_{i},x_{i},m_{i},i,\cnf_{i+1})$
is a valid transition. Denote $\cnf_{i}=(\cnfs_{i},\cnfb_{i})$ and
$\cnf_{i+1}=(\cnfs_{i+1},\cnfb_{i+1})$.

We start by showing $m_{i}\subseteq\cnfb_{i}(x_{i})$. Let $\pair j{\fc}\in m_{i}$.
By definition of $m_{i}$, there is a fact of the form $\chosen_{R}(y,t,z,u,\bar{a})\in M$
with $\globM{z,u}=i$ such that $j=\globM{y,t}$ and $\fc=R(\bar{a})$.
Note, $\globM{z,u}=i$ implies $z=x_{i}$ and $u=s_{i}$. Now, because
rules in $\pure{\ded}$ of the form (\ref{eq:before-send}) are always
positive, the following ground rule is in $\grded$, which is of the
form (\ref{eq:before-send}): 
\[
\before(y,t,x_{i},s_{i})\gets\chosen_{R}(y,t,x_{i},s_{i},\bar{a}).
\]
Since its body is in $M$, this rule derives $\before(y,t,x_{i},s_{i})\in M$.
Hence $(y,t)\cauM(x_{i},s_{i})$ by definition of $\cauM$. Moreover,
$\totM$ respects $\cauM$, and thus $(y,t)\totM(x_{i},s_{i})$, which
implies $\globM{y,t}<\globM{x_{i},s_{i}}$. And since $\globM{x_{i},s_{i}}=i$,
we overall have 
\[
\globM{y,t}<i\leq\globM{x_{i},s_{i}}.
\]
Therefore $\pair j{\fc}\in\cnfb_{i}(x_{i})$. 

Now, because $m_{i}\subseteq\cnfb_{i}(x_{i})$, and because transitions
are deterministic once the active node and delivered messages are
fixed, we can consider the unique result configuration $\cnf=(\cnfs,\cnfb)$
such that $(\cnf_{i},x_{i},m_{i},i,\cnf)$ is a valid transition.
We are left to show $\cnf_{i+1}=\cnf$. We divide the work in two
parts: for each $x\in\nw$, we show that $\romI$ $\cnfs_{i+1}(x)=\cnfs(x)$,
and $\romII$ $\cnfb_{i+1}(x)=\cnfb(x)$.

\subsubsection{State}

Let $x\in\nw$. We show $\cnfs_{i+1}(x)=\cnfs(x)$. Denote $s=\locM{i+1,x}$.
By definition, 
\[
\cnfs_{i+1}(x)=\droplt{\left(\shprojlt{\proj M{\toloct{\edb{\ded}}}}xs\cup\projlt{\Mind}xs\right)}.
\]

\paragraph{Case $x\protect\neq x_{i}$. }

By definition, $\cnfs(x)=\cnfs_{i}(x)$. Hence, it suffices to show
$\cnfs_{i+1}(x)=\cnfs_{i}(x)$. Since $x\neq x_{i}$, the number of
pairs from $\nwnat$ containing node $x$ that come strictly before
ordinal $i+1$ is the same as the number of pairs containing node
$x$ that come strictly before ordinal $i$. Formally: $s=\locM{i+1,x}=\locM{i,x}$.
Thus the right-hand side in the previous equation equals $\cnfs_{i}(x)$,
and the result is obtained.

\paragraph{Case $x=x_{i}$.}

By definition, $\cnfs(x)=H(x)\cup\induc{\ded}\mstep{D_{i}}$. Referring
to the definition of $\cnfs_{i+1}(x)$ from above, by Claim~\ref{claim:edb-in-M}
we have 
\[
\shprojlt{\proj M{\toloct{\edb{\ded}}}}xs=\addlt{H(x)}xs.
\]
If we can also show $\projlt{\Mind}xs=\addlt{\induc{\ded}\mstep{D_{i}}}xs$,
then we overall have, as desired: 
\begin{eqnarray*}
\cnfs_{i+1}(x) & = & \droplt{\left(\shprojlt{\proj M{\toloct{\edb{\ded}}}}xs\cup\projlt{\Mind}xs\right)}\\
 & = & H(x)\cup\induc{\ded}\mstep{D_{i}}\\
 & = & \cnfs(x).
\end{eqnarray*}
Since $x=x_{i}$, we have $s=\locM{i+1,x_{i}}=\locM{i,x_{i}}+1$,
and using that $\locM{i,x_{i}}=s_{i}$ (Claim~\ref{claim:localM}),
we have $s=s_{i}+1$. Now, Claim~\ref{claim:ind-M-in-induc} and
Claim~\ref{claim:ind-induc-in-M} together show $\projlt{\Mind}{x_{i}}{s_{i}+1}=\addlt{\induc{\ded}\mstep{D_{i}}}{x_{i}}{s_{i}+1}$.

\subsubsection{Buffer}

Let $x\in\nw$. We show $\cnfb_{i+1}(x)=\cnfb(x)$. Denote
\[
\sendto ix=\{\pair i{R(\bar{a})}\mid R(x,\bar{a})\in\async{\ded}\mstep{D_{i}}\}.
\]
Like in the operational semantics, $\sendto ix$ denotes the (tagged)
messages that are sent to $x$ during transition $i$.

\paragraph*{Case $x\protect\neq x_{i}$. }

By definition, $\cnfb(x)=\cnfb_{i}(x)\cup\sendto ix$. We start by
showing $\cnfb(x)\subseteq\cnfb_{i+1}(x)$. Let $\pair j{\fc}\in\cnfb(x)$.
Denote $\fc=R(\bar{a})$. 
\begin{itemize}
\item Suppose $\pair j{\fc}\in\cnfb_{i}(x)$. By definition of $\cnfb_{i}(x)$,
there are values $y\in\nw$, $t\in\Nat$ and $u\in\Nat$ such that
$\chosen_{R}(y,t,x,u,\bar{a})\in M$ and $j=\globM{y,t}<i\leq\globM{x,u}$.
Now, since $x\neq x_{i}$, we more specifically have $i<\globM{x,u}$
and thus $i+1\leq\globM{x,u}$. Therefore $\pair j{\fc}\in\cnfb_{i+1}(x)$,
as desired.
\item Suppose $\pair j{\fc}\in\sendto ix$. By definition of $\sendto ix$,
this implies $j=i$ and $R(x,\bar{a})\in\async{\ded}\mstep{D_{i}}$.
Then $\pair j{\fc}=\pair i{R(\bar{a})}\in\cnfb_{i+1}(x)$ by Claim~\ref{claim:async-in-bufM},
as desired.
\end{itemize}
Secondly, we show $\cnfb_{i+1}(x)\subseteq\cnfb(x)$. Let $\pair j{\fc}\in\cnfb_{i+1}(x)$.
Denote $\fc=R(\bar{a})$. By definition of $\cnfb_{i+1}(x)$, there
are values $y\in\nw$, $t\in\nw$ and $u\in\nw$ such that $\chosen_{R}(y,t,x,u,\bar{a})\in M$
and $j=\globM{y,t}<i+1\leq\globM{x,u}$. So $j\leq i$. We have the
following cases:
\begin{itemize}
\item Suppose $j<i$. Thus $\globM{y,t}<i$. This immediately gives $\pair j{\fc}\in\cnfb_{i}(x)\subseteq\cnfb(x)$,
as desired.
\item Suppose $j=i$. Then $R(x,\bar{a})\in\async{\ded}\mstep{D_{i}}$ by
Claim~\ref{claim:bufM-in-async}. This implies that $\pair j{\fc}=\pair i{R(\bar{a})}\in\sendto ix\subseteq\cnfb(x)$,
as desired.
\end{itemize}

\paragraph*{Case $x=x_{i}$.}

By definition, $\cnfb(x)=(\cnfb_{i}(x)\setminus m_{i})\cup\sendto ix$.
Some parts of the reasoning are similar to the case $x\neq x_{i}$.
We refer to shared subclaims where possible. 

We start by showing $\cnfb(x)\subseteq\cnfb_{i+1}(x)$. Let $\pair j{\fc}\in\cnfb(x)$.
Denote $\fc=R(\bar{a})$. We have the following cases:
\begin{itemize}
\item Suppose $\pair j{\fc}\in\cnfb_{i}(x)\setminus m_{i}$. Thus $\pair j{\fc}\in\cnfb_{i}(x)$
and $\pair j{\fc}\notin m_{i}$. Here, $\pair j{\fc}\in\cnfb_{i}(x)$
implies there are values $y\in\nw$, $t\in\Nat$ and $u\in\Nat$ such
that $\chosen_{R}(y,t,x,u,\bar{a})\in M$ and $j=\globM{y,t}<i\leq\globM{x,u}$.
Also, $\pair j{\fc}\notin m_{i}$ implies $\globM{x,u}\neq i$. Hence,
$i+1\leq\globM{x,u}$ and we obtain $\pair j{\fc}\in\cnfb_{i+1}(x)$,
as desired.
\item Suppose $\pair j{\fc}\in\sendto ix$. By definition of $\sendto ix$,
we have $j=i$ and $R(x,\bar{a})\in\async{\ded}\mstep{D_{i}}$. By
Claim~\ref{claim:async-in-bufM} we then have $\pair i{R(\bar{a})}\in\cnfb_{i+1}(x)$,
as desired.
\end{itemize}
Secondly, we show $\cnfb_{i+1}(x)\subseteq\cnfb(x)$. Let $\pair j{\fc}\in\cnfb_{i+1}(x)$.
Denote $\fc=R(\bar{a})$. By definition of $\cnfb_{i+1}(x)$, there
are values $y\in\nw$, $t\in\Nat$ and $u\in\Nat$ such that $\chosen_{R}(y,t,x,u,\bar{a})\in M$
and $j=\globM{y,t}<i+1\leq\globM{x,u}$. Now we look at the cases
for $j$:
\begin{itemize}
\item Suppose $j<i$. This gives us $\globM{y,t}<i\leq\globM{x,u}$, which
implies $\pair j{\fc}\in\cnfb_{i}(x)$. Moreover, $i+1\leq\globM{x,u}$
gives $\globM{x,u}\neq i$. Hence, $\pair j{\fc}\notin m_{i}$. Taken
together, we now have $\pair j{\fc}\in\cnfb_{i}(x)\setminus m_{i}\subseteq\cnfb(x)$.
\item Suppose $j=i$. Then $\pair i{R(\bar{a})}\in\cnfb_{i+1}(x)$, and
by Claim~\ref{claim:bufM-in-async} we obtain that $R(x,\bar{a})\in\async{\ded}\mstep{D_{i}}$.
Therefore $\pair j{\fc}=\pair i{R(\bar{a})}\in\sendto ix\subseteq\cnfb(x)$,
as desired.
\end{itemize}

\subsection{Trace}

\label{sub:model-to-run--trace}

In this section we show $\trace{\run}=\proj M{\toloct{\schof{\ded}}}$.
Recall from Section~\ref{sub:run-timestamps-and-trace} that
\[
\trace{\run}=\bigcup_{i\in\Nat}(\addlt{D_{i})}{x_{i}}{\,\locR i}.
\]
For each $i\in\Nat$, $\locR i$ is the number of transitions in $\run$
before $i$ in which $x_{i}$ is also the active node. From the construction
of $\run$ we know $\locR i=\locM{i,x_{i}}$; indeed, $\locM{i,x_{i}}$
counts the number of pairs in $\nwnat$ with node $x_{i}$ that have
an ordinal strictly smaller than $i$, which is precisely the number
of transitions in $\run$ with active node $x_{i}$ that come before
$i$. Moreover, by Claim~\ref{claim:localM} we have $\locM{i,x_{i}}=s_{i}$.
Hence, 
\[
\trace{\run}=\bigcup_{i\in\Nat}\addlt{(D_{i})}{x_{i}}{s_{i}}.
\]
Thus, by Claim~\ref{claim:M-and-duc}:
\[
\trace{\run}=\bigcup_{i\in\Nat}\projlt M{x_{i}}{s_{i}}.
\]
For the next step, let us denote $A=\{(x_{i},s_{i})\mid i\in\Nat\}$.
We show $A=\nwnat$. First, we have $A\subseteq\nwnat$ because $x_{i}\in\nw$
and $s_{i}\in\Nat$ for each $i\in\Nat$. Now, let $(x,s)\in\nwnat$.
Denote $i=\globM{x,s}$. By definition, $x_{i}=x$ and $s_{i}=s$.
Hence $(x,s)=(x_{i},s_{i})\in A$. Now we may write:
\begin{eqnarray*}
\trace{\run} & = & \bigcup_{(x,s)\in A}\projlt Mxs\\
 & = & \bigcup_{(x,s)\in\nwnat}\projlt Mxs.
\end{eqnarray*}
Finally, because $M$ is well-formed (see Section~\ref{sub:model-to-run}),
for each $R(v,w,\bar{a})\in\proj M{\toloct{\schof{\ded}}}$ we have
$v\in\nw$ and $w\in\Nat$. We obtain, as desired:
\[
\trace{\run}=\proj M{\toloct{\schof{\ded}}}.
\]

\subsection{Subclaims}

\label{sub:model-to-run--claims}

\begin{claim}\label{claim:edb-in-M}Let $x\in\nw$ and $s\in\Nat$.
We have $\shprojlt{\proj M{\toloct{\edb{\ded}}}}xs=\addlt{H(x)}xs$.\end{claim}

\begin{proof}First, by construction of $\decl H$ we have $\shprojlt{\proj{\decl H}{\toloct{\edb{\ded}}}}xs=\addlt{H(x)}xs$.
Because $\decl H\subseteq M$, and because facts over $\toloct{\edb{\ded}}$
can not be derived by rules in $\pure{\ded}$, we have $\proj M{\toloct{\edb{\ded}}}=\proj{\decl H}{\toloct{\edb{\ded}}}$.
Hence, 
\[
\shprojlt{\proj M{\toloct{\edb{\ded}}}}xs=\shprojlt{\proj{\decl H}{\toloct{\edb{\ded}}}}xs=\addlt{H(x)}xs.
\]
\end{proof}

\tline

\begin{claim}\label{claim:localM}Let $i\in\Nat$. We have $s_{i}=\locM{i,x_{i}}$.\end{claim}

\begin{proof}

Recall that $(x_{i},s_{i})\in\nwnat$ is the unique pair at ordinal
$i$ in $\totM$, i.e., $\globM{x_{i},s_{i}}=i$. Suppose we would
know for all $s\in\Nat$ and $t\in\Nat$ that $s<t$ implies $\globM{x_{i},s}<\globM{x_{i},t}$.
Then $\locM{i,x_{i}}$, which is 
\[
|\{s\in\Nat\mid\globM{x,s}<i\}|,
\]
is precisely 
\[
|\{s\in\Nat\mid s<s_{i}\}|.
\]
The latter is just $s_{i}$.

We are left to show for any $s\in\Nat$ and $t\in\Nat$ that $s<t$
implies $\globM{x_{i},s}<\globM{x_{i},t}$. It is actually sufficient
to show for any $s\in\Nat$ that $(x_{i},s)\cauM(x_{i},s+1)$. Indeed,
this would imply for any $t\in\Nat$ with $s<t$ that 
\[
(x_{i},s)\cauM(x_{i},s+1)\cauM(x_{i},s+2)\cauM\ldots\cauM(x_{i},t).
\]
And since $\cauM$ is a partial order, it is transitive, and thus
$(x_{i},s)\cauM(x_{i},t)$. Next, since $\totM$ respects $\cauM$,
we obtain $(x_{i},s)\totM(x_{i},t)$ and thus $\globM{x_{i},s}<\globM{x_{i},t}$,
as desired. To show $(x_{i},s)\cauM(x_{i},s+1)$, we observe that
the rule (\ref{eq:before-step}) in $\pure{\ded}$ is positive. Hence,
for any $s\in\Nat$, the following ground rule is always in $\grded$,
and it derives $\before(x_{i},s,x_{i},s+1)\in M$ because $\relall(x_{i})\in\decl H$
and $\timesucc(s,s+1)\in\decl H$:
\[
\before(x_{i},s,x_{i},s+1)\gets\relall(x_{i}),\,\timesucc(s,s+1).
\]
Thus $(x_{i},s)\cauM(x_{i},s+1)$ by definition of $\cauM$.\end{proof}

\tline

\begin{claim}\label{claim:ind-M-in-induc}Let $i\in\Nat$. We have
$\projlt{\Mind}{x_{i}}{s_{i}+1}\subseteq\addlt{\induc{\ded}\mstep{D_{i}}}{x_{i}}{s_{i}+1}$.\end{claim}

\begin{proof}

Let $\fc\in\projlt{\Mind}{x_{i}}{s_{i}+1}$. We show $\fc\in\addlt{\induc{\ded}\mstep{D_{i}}}{x_{i}}{s_{i}+1}$. 

By definition of $\Mind$, there is an active \emph{inductive} ground
rule $\grl\in\grded$ with $\head{\grl}=\fc$. Because $\grl\in\grded$,
there is a rule $\rl\in\pure{\ded}$ and a valuation $V$ so that
$\grl$ can be obtained from $\rl$ by applying $V$ and by subsequently
removing all negative (ground) body literals, and so that $V(\bneg{\rl})\cap M=\emptyset$.
The rule $\rl$ must be of the form (\ref{eq:pure-ind}), which implies
that $V$ must assign $x_{i}$ and $s_{i}$ to the body location and
timestamp variable respectively, and that it must assign $x_{i}$
and $s_{i}+1$ to the head location and timestamp variable respectively.

Let $\rl'\in\ded$ be the original inductive rule on which $\rl$
is based. Let $\rl''\in\induc{\ded}$ be the rule corresponding to
$\rl'$. It follows from the construction of $\rl$ out of $\rl'$
and $\rl''$ out of $\rl'$ that valuation $V$ can also be applied
to rule $\rl''$. Indeed, rule $\rl$ just has more variables for
the location and timestamps. We show that $V$ is satisfying for $\rl''$
with respect to $D_{i}$, so that $\rl''$ and $V$ together derive
$V(\head{\rl''})=\droplt{\fc}\in\induc{\ded}\mstep{D_{i}}$, which
gives $\fc\in\addlt{\induc{\ded}\mstep{D_{i}}}{x_{i}}{s_{i}+1}$, as
desired.

We must concretely show $V(\bpos{\rl''})\subseteq D_{i}$ and $V(\bneg{\rl''})\cap D_{i}=\emptyset$.
We start by showing $V(\bpos{\rl''})\subseteq D_{i}$. From the relationship
between $\grl$, $\rl$ and $\rl''$, we know that 
\[
\proj{\bpos{\grl}}{\toloct{\schof{\ded}}}=\proj{V(\bpos{\rl})}{\toloct{\schof{\ded}}}=\addlt{V(\bpos{\rl''})}{x_{i}}{s_{i}}.
\]
Since $\grl$ is active with respect to $M$, we have $\bpos{\grl}\subseteq M$,
and thus $\addlt{V(\bpos{\rl''})}{x_{i}}{s_{i}}\subseteq M$. Then
by Claim~\ref{claim:in-M-in-duc-ver2} we have $V(\bpos{\rl''})\subseteq D_{i}$,
as desired.

Now we show that $V(\bneg{\rl''})\cap D_{i}=\emptyset$. By the relationship
of $\rl$ and $\rl''$, we have $\addlt{V(\bneg{\rl''})}{x_{i}}{s_{i}}=V(\bneg{\rl})$.
By choice of $\rl$ and $V$, we have $V(\bneg{\rl})\cap M=\emptyset$.
Hence, $\addlt{V(\bneg{\rl''})}{x_{i}}{s_{i}}\cap M=\emptyset$. Finally,
by Claim~\ref{claim:not-in-M-not-in-duc-ver2}, we have $V(\bneg{\rl''})\cap D_{i}=\emptyset$,
as desired.\end{proof}

\tline

\begin{claim}\label{claim:in-M-in-duc-ver2}Let $i\in\Nat$. Let
$I$ be a set of facts over $\toloct{\schof{\ded}}$ that all have
location specifier $x_{i}$ and timestamp $s_{i}$. If $I\subseteq M$
then $\droplt I\subseteq D_{i}$, with $D_{i}$ as defined in Section~\ref{sub:model-to-run-defs-notations}.\end{claim}

\begin{proof}

We are given $I\subseteq M$. By the assumptions on $I$, we more
specifically have $I\subseteq\projlt M{x_{i}}{s_{i}}$. Then by Claim~\ref{claim:M-and-duc}
we have $I\subseteq\addlt{(D_{i})}{x_{i}}{s_{i}}$. Hence $\droplt I\subseteq D_{i}$,
as desired.\end{proof}

\tline

\begin{claim}\label{claim:not-in-M-not-in-duc-ver2}Let $i\in\Nat$.
Let $I$ be a set of facts over $\toloct{\schof{\ded}}$ that all
have location specifier $x_{i}$ and timestamp $s_{i}$. If $I\cap M=\emptyset$
then $\droplt I\cap D_{i}=\emptyset$, with $D_{i}$ as defined in
Section~\ref{sub:model-to-run-defs-notations}.\end{claim}

\begin{proof}

We are given that $I\cap M=\emptyset$. This implies $I\cap\projlt M{x_{i}}{s_{i}}=\emptyset$.
By Claim~\ref{claim:M-and-duc} we have $I\cap\addlt{(D_{i})}{x_{i}}{s_{i}}=\emptyset$.
Hence, by the assumptions on $I$, we have $\droplt I\cap D_{i}=\emptyset$,
as desired.\end{proof}

\tline

\begin{claim}\label{claim:ind-induc-in-M}Let $i\in\Nat$. We have
$\addlt{\induc{\ded}\mstep{D_{i}}}{x_{i}}{s_{i}+1}\subseteq\projlt{\Mind}{x_{i}}{s_{i}+1}$.\end{claim}

\begin{proof}Let $\fc\in\induc{\ded}\mstep{D_{i}}$. We show that
$\addlt{\fc}{x_{i}}{s_{i}+1}\in\projlt{\Mind}{x_{i}}{s_{i}+1}$. 

Recall the semantics for $\induc{\ded}$ from Section~\ref{sub:subprograms}.
Let $\rl\in\induc{\ded}$ and $V$ be the rule and valuation that
together derived $\fc\in\induc{\ded}\mstep{D_{i}}$. Let $\rl'\in\ded$
be the original inductive rule on which $\rl$ is based. Let $\rl''\in\pure{\ded}$
be the inductive rule that in turn is based on $\rl'$, which is of
the form (\ref{eq:pure-ind}). Let $V''$ be the valuation for $\rl''$
that is obtained by extending $V$ to assign $x_{i}$ and $s_{i}$
to respectively the location and timestamp variables in the body,
and to assign $s_{i}+1$ to the head timestamp variable. Let $\grl$
be the positive ground rule obtained from $\rl''$ by applying the
valuation $V''$, and by subsequently removing the negative (ground)
body literals. Note that $\head{\grl}=\addlt{V(\head{\rl})}{x_{i}}{s_{i}+1}=\addlt{\fc}{x_{i}}{s_{i}+1}$.
We will show that $\grl\in\grded$ and that $\bpos{\grl}\subseteq M$,
so that this ground rule derives $\addlt{\fc}{x_{i}}{s_{i}+1}\in M$.
And since $\grl$ is inductive, we more specifically have $\addlt{\fc}{x_{i}}{s_{i}+1}\in\projlt{\Mind}{x_{i}}{s_{i}+1}$,
as desired.
\begin{itemize}
\item For $\grl\in\grded$, we require $V''(\bneg{\rl''})\cap M=\emptyset$.
From the construction of rule $\rl''$, we have $V''(\bneg{\rl''})=\addlt{V(\bneg{\rl})}{x_{i}}{s_{i}}$.
We show $\addlt{V(\bneg{\rl})}{x_{i}}{s_{i}}\cap M=\emptyset$.

Because $V$ is satisfying for $\rl$ with respect to $D_{i}$, we
have $V(\bneg{\rl})\cap D_{i}=\emptyset$. This gives $\addlt{V(\bneg{\rl})}{x_{i}}{s_{i}}\cap\addlt{(D_{i})}{x_{i}}{s_{i}}=\emptyset$.
Then $\addlt{V(\bneg{\rl})}{x_{i}}{s_{i}}\cap\projlt M{x_{i}}{s_{i}}=\emptyset$
by Claim~\ref{claim:M-and-duc}. Next, we obtain $\addlt{V(\bneg{\rl})}{x_{i}}{s_{i}}\cap M=\emptyset$
since $\addlt{V(\bneg{\rl})}{x_{i}}{s_{i}}$ contains only facts over
$\toloct{\schof{\ded}}$ with location specifier $x_{i}$ and timestamp
$s_{i}$.

\item Now we show $\bpos{\grl}\subseteq M$. From the construction of rule
$\rl''$, we have 
\[
\bpos{\grl}=V''(\bpos{\rl''})=\addlt{V(\bpos{\rl})}{x_{i}}{s_{i}}\cup\{\timesucc(s_{i},s_{i}+1)\}.
\]
We immediately have $\timesucc(s_{i},s_{i}+1)\in\decl H\subseteq M$.
Moreover, since $V$ is satisfying for $\rl$ with respect to $D_{i}$,
we have $V(\bpos{\rl})\subseteq D_{i}$. Hence $\addlt{V(\bpos{\rl})}{x_{i}}{s_{i}}\subseteq\addlt{(D_{i})}{x_{i}}{s_{i}}$.
By Claim~\ref{claim:M-and-duc} we then have $\addlt{V(\bpos{\rl})}{x_{i}}{s_{i}}\subseteq\projlt M{x_{i}}{s_{i}}\subseteq M$,
as desired.
\end{itemize}
\end{proof}

\tline

\begin{claim}\label{claim:async-in-bufM}Let $i\in\Nat$. Let $x\in\nw$.
For each $R(x,\bar{a})\in\async{\ded}\mstep{D_{i}}$, we have $\pair i{R(\bar{a})}\in\cnfb_{i+1}(x)$.
\end{claim}

\begin{proof}

The main approach of this proof is as follows. We will show there
is a timestamp $u\in\Nat$ such that $\chosen_{R}(x_{i},s_{i},x,u,\bar{a})\in M$.
Next, because rules of the form (\ref{eq:before-send}) are positive,
in $\grded$ there is always the following ground rule:
\[
\before(x_{i},s_{i},x,u)\gets\chosen_{R}(x_{i},s_{i},x,u,\bar{a}).
\]
Thus if $\chosen_{R}(x_{i},s_{i},x,u,\bar{a})\in M$ then $\before(x_{i},s_{i},x,u)\in M$,
which implies $(x_{i},s_{i})\cauM(x,u)$ by definition of $\cauM$.
Since $\totM$ respects $\cauM$, we obtain $(x_{i},s_{i})\totM(x,u)$
and thus $\globM{x_{i},s_{i}}<\globM{x,u}$. Also, since $\globM{x_{i},s_{i}}=i$,
we overall get 
\[
\globM{x_{i},s_{i}}<i+1\leq\globM{x,u},
\]
which together with $\chosen_{R}(x_{i},s_{i},x,u,\bar{a})\in M$ gives
$\pair{\globM{x_{i},s_{i}}}{\, R(\bar{a})}=\pair i{R(\bar{a})}\in\cnfb_{i+1}(x)$,
as desired.

Now we are left to show that such a timestamp $u$ exists. Recall
the semantics for $\async{\ded}$ from Section~\ref{sub:subprograms}.
Let $\rl\in\async{\ded}$ and $V$ be a rule and valuation that together
have derived $R(x,\bar{a})\in\async{\ded}\mstep{D_{i}}$. Let $\rl'\in\ded$
be the original asynchronous rule on which $\rl$ is based. Let $\rl''\in\pure{\ded}$
be the rule obtained by applying transformation (\ref{eq:cand}) to
$\rl'$. To continue, because $\cauM$ is well-founded, there are
only a finite number of timestamps $v\in\Nat$ of node $x$ such that
$(x,v)\cauM(x_{i},s_{i})$. So, there exists a timestamp $u\in\Nat$
such that $(x,u)\not\cauM(x_{i},s_{i})$. Now, let $V''$ be the
valuation for $\rl''$ that is the extension of valuation $V$ to
assign $x_{i}$ and $s_{i}$ to the body location variable and timestamp
variable respectively (both belonging to the sender), and to assign
$u$ to the addressee arrival timestamp. Note that from the construction
of $\rl''$ we also know that $V$ (and thus $V''$) assigns the value
$x$ to the addressee location variable and the tuple $\bar{a}$ to
the message contents. Let $\grl$ denote the ground rule obtained
by applying $V''$ to $\rl''$, and by subsequently removing the negative
(ground) body literals. We will first show that $\grl\in\grded$,
and then we show that $\bpos{\grl}\subseteq M$, meaning that $\grl$
derives $\head{\grl}=\cand_{R}(x_{i},s_{i},x,u,\bar{a})\in M$. Then
Claim~\ref{claim:cand-implies-chosen} can be applied to know that
there is a timestamp $u'$, with possibly $u'=u$, such that $\chosen_{R}(x_{i},s_{i},x,u',\bar{a})\in M$,
as desired.

In order for $\grl$ to be in $\grded$, we require $V''(\bneg{\rl''})\cap M=\emptyset$.
It follows from the construction of $\rl''$ out of $\rl'$ and $\rl$
out of $\rl'$ that 
\[
V''(\bneg{\rl''})=\addlt{V(\bneg{\rl})}{x_{i}}{s_{i}}\cup\{\before(x,u,x_{i},s_{i})\}.
\]
We have $\before(x,u,x_{i},s_{i})\notin M$ because $(x,u)\not\cauM(x_{i},s_{i})$
by choice of $u$.  Next, we show that $\addlt{V(\bneg{\rl})}{x_{i}}{s_{i}}\cap M=\emptyset$.
Because $V$ is satisfying for $\rl$ with respect to $D_{i}$, we
have $V(\bneg{\rl})\cap D_{i}=\emptyset$, and thus
\[
\addlt{V(\bneg{\rl})}{x_{i}}{s_{i}}\cap\addlt{(D_{i})}{x_{i}}{s_{i}}=\emptyset.
\]
Then, by Claim~\ref{claim:M-and-duc}, 
\[
\addlt{V(\bneg{\rl})}{x_{i}}{s_{i}}\cap\projlt M{x_{i}}{s_{i}}=\emptyset.
\]
Since $\addlt{V(\bneg{\rl})}{x_{i}}{s_{i}}$ contains only facts over
$\toloct{\schof{\ded}}$ with location specifier $x_{i}$ and timestamp
$s_{i}$, we have 
\[
\addlt{V(\bneg{\rl})}{x_{i}}{s_{i}}\cap M=\emptyset.
\]

We now show $\bpos{\grl}\subseteq M$. Note, $\bpos{\grl}=V''(\bpos{\rl''})$.
From the construction of $\rl''$ we have 
\[
V''(\bpos{\rl''})=\addlt{V(\bpos{\rl})}{x_{i}}{s_{i}}\cup\{\relall(x),\,\reltime(u)\}.
\]
Because $x\in\nw$ and $u\in\Nat$, we immediately have $\{\relall(x),\,\reltime(u)\}\subseteq\decl H\subseteq M$.
We are left to show $\addlt{V(\bpos{\rl})}{x_{i}}{s_{i}}\subseteq M$.
Because $V$ is satisfying for $\rl$ with respect to $D_{i}$, we
have $V(\bpos{\rl})\subseteq D_{i}$. Hence $\addlt{V(\bpos{\rl})}{x_{i}}{s_{i}}\subseteq\addlt{(D_{i})}{x_{i}}{s_{i}}$.
By again using Claim~\ref{claim:M-and-duc} we then obtain $\addlt{V(\bpos{\rl})}{x_{i}}{s_{i}}\subseteq\projlt M{x_{i}}{s_{i}}\subseteq M$,
as desired.\end{proof}

\tline

\begin{claim}\label{claim:bufM-in-async}Let $i\in\Nat$ and $x\in\nw$.
For each $\pair i{R(\bar{a})}\in\cnfb_{i+1}(x)$, we have $R(x,\bar{a})\in\async{\ded}\mstep{D_{i}}$.\end{claim}

\begin{proof}

By definition of $\cnfb_{i+1}(x)$, the pair $\pair i{R(\bar{a})}\in\cnfb_{i+1}(x)$
implies that there are values $y\in\nw$, $t\in\Nat$ and $u\in\Nat$
such that $\chosen_{R}(y,t,x,u,\bar{a})\in M$, $\globM{y,t}=i$ and
$\globM{y,t}<i+1\leq\globM{x,u}$. And $\globM{y,t}=i$ gives us that
$y=x_{i}$ and $t=s_{i}$. Thus $\chosen_{R}(x_{i},s_{i},x,u,\bar{a})\in M$.

All ground rules in $\grded$ that can derive $\chosen_{R}(x_{i},s_{i},x,u,\bar{a})\in M$
are of the form (\ref{eq:chosen}), and hence $\cand_{R}(x_{i},s_{i},x,u,\bar{a})\in M$.
Let $\grl\in\grded$ be an active ground rule with head $\cand_{R}(x_{i},s_{i},x,u,\bar{a})$.
Because $\grl\in\grded$, there is a rule $\rl\in\pure{\ded}$ and
a valuation $V$ so that $\grl$ is obtained from $\rl$ by applying
$V$ and by subsequently removing all negative (ground) body literals,
and so that $V(\bneg{\rl})\cap M=\emptyset$. The rule $\rl$ is of
the form (\ref{eq:cand}), which implies that $V$ must assign $x_{i}$
and $s_{i}$ respectively to the body location and timestamp variable
that correspond to the sender, and that it must assign $x$ and $u$
respectively to the location and timestamp variable that correspond
to the addressee. Let $\rl'\in\ded$ be the original asynchronous
rule on which $\rl$ is based. Let $\rl''$ be the corresponding rule
in $\async{\ded}$. From the construction of $\rl$ out of $\rl'$
and $\rl''$ out of $\rl'$, it follows that $V$ can also be applied
to $\rl''$.  Note, $V(\head{\rl''})=R(x,\bar{a})$. We now show
that $V$ is satisfying for $\rl''$ with respect to $D_{i}$, which
causes $R(x,\bar{a})\in\async{\ded}\mstep{D_{i}}$, as desired. Specifically,
we have to show $V(\bpos{\rl''})\subseteq D_{i}$ and $V(\bneg{\rl''})\cap D_{i}=\emptyset$.

First we show $V(\bpos{\rl''})\subseteq D_{i}$. By construction of
$\rl$ and $\rl''$, we have 
\[
\proj{\bpos{\grl}}{\toloct{\schof{\ded}}}=\proj{V(\bpos{\rl})}{\toloct{\schof{\ded}}}=\addlt{V(\bpos{\rl''})}{x_{i}}{s_{i}}.
\]
Since $\grl$ is active, we have $\proj{\bpos{\grl}}{\toloct{\schof{\ded}}}\subseteq M$,
and therefore $\addlt{V(\bpos{\rl''})}{x_{i}}{s_{i}}\subseteq M$.
 Then, because the facts in $\addlt{V(\bpos{\rl''})}{x_{i}}{s_{i}}$
are over $\toloct{\schof{\ded}}$ and have location specifier $x_{i}$
and timestamp $s_{i}$, we can apply Claim \ref{claim:in-M-in-duc-ver2}
to know that $V(\bpos{\rl''})\subseteq D_{i}$, as desired.

Now we show $V(\bneg{\rl''})\cap D_{i}=\emptyset$. By construction
of $\rl$ and $\rl''$, we have 
\[
\proj{V(\bneg{\rl})}{\toloct{\schof{\ded}}}=\addlt{V(\bneg{\rl''})}{x_{i}}{s_{i}}.
\]
By choice of $\rl$ and $V$, we have $V(\bneg{\rl})\cap M=\emptyset$.
Hence, $\addlt{V(\bneg{\rl''})}{x_{i}}{s_{i}}\cap M=\emptyset$. Then,
because the facts in $\addlt{V(\bneg{\rl''})}{x_{i}}{s_{i}}$ are
over $\toloct{\schof{\ded}}$ and have location specifier $x_{i}$
and timestamp $s_{i}$, we can apply Claim \ref{claim:not-in-M-not-in-duc-ver2}
to know that $V(\bneg{\rl''})\cap D_{i}=\emptyset$, as desired.\end{proof}

\tline

\begin{claim}\label{claim:M-and-duc}Let $i\in\Nat$. We have $\projlt M{x_{i}}{s_{i}}=\addlt{(D_{i})}{x_{i}}{s_{i}}$.
Intuitively, this means that the operational deductive fixpoint $D_{i}$
during transition $i$, corresponding to step $s_{i}$ of node $x_{i}$,
is represented by $M$ in an exact way.\end{claim}

\begin{proof}

Recall the notations from Section~\ref{sub:model-to-run-defs-notations}.
Let $n$ denote the largest stratum number of the deductive rules
of $\ded$. We show by induction on $k=0,1,\ldots,n$ that 
\[
\shprojlt{\strat Mk}{x_{i}}{s_{i}}=\shaddlt{\strat{D_{i}}k}{x_{i}}{s_{i}}.
\]
This will give us $\shprojlt{\strat Mn}{x_{i}}{s_{i}}=\shaddlt{\strat{D_{i}}n}{x_{i}}{s_{i}}=\addlt{(D_{i})}{x_{i}}{s_{i}}.$
Moreover, Claim~\ref{claim:M-largest-stratum} says that $\shprojlt{\strat Mn}{x_{i}}{s_{i}}=\projlt M{x_{i}}{s_{i}}$,
and thus we obtain $\projlt M{x_{i}}{s_{i}}=\addlt{(D_{i})}{x_{i}}{s_{i}}$,
as desired.

\paragraph*{Base case ($k=0$)}

By definition, 
\[
\strat M0=\up M\cup\duc M0.
\]
But since there are no deductive ground rules in $\grded$ with stratum
$0$, we have $\duc M0=\emptyset$. Hence, 
\begin{eqnarray}
\shprojlt{\strat M0}{x_{i}}{s_{i}} & = & \shprojlt{\up M}{x_{i}}{s_{i}}\nonumber \\
 & = & \shprojlt{\proj M{\toloct{\edb{\ded}}}}{x_{i}}{s_{i}}\cup\projlt{\ind M}{x_{i}}{s_{i}}\cup\projlt{\deliv M}{x_{i}}{s_{i}}.\label{eq:expression-for-M-0}
\end{eqnarray}
Using Claim~\ref{claim:state-is-edb-ind} and Claim~\ref{claim:deliv-is-untag},
we can rewrite expression (\ref{eq:expression-for-M-0}) to the desired
equality: 
\begin{eqnarray*}
\shprojlt{\strat M0}{x_{i}}{s_{i}} & = & \addlt{\cnfs_{i}(x_{i})}{x_{i}}{s_{i}}\cup\addlt{\untag{m_{i}}}{x_{i}}{s_{i}}\\
 & = & \shaddlt{\cnfs_{i}(x_{i})\cup\untag{m_{i}}}{x_{i}}{s_{i}}\\
 & = & \shaddlt{\strat{D_{i}}0}{x_{i}}{s_{i}}.
\end{eqnarray*}

\paragraph*{Induction hypothesis}

For the induction hypothesis, we assume for a stratum number $k\geq1$
that 
\[
\shprojlt{\strat M{k-1}}{x_{i}}{s_{i}}=\shaddlt{\strat{D_{i}}{k-1}}{x_{i}}{s_{i}}.
\]

\paragraph*{Inductive step}

We show that 
\[
\shprojlt{\strat Mk}{x_{i}}{s_{i}}=\shaddlt{\strat{D_{i}}k}{x_{i}}{s_{i}}.
\]
We show both inclusions separately, in Claims~\ref{claim:strat-k-M-in-D}
and \ref{claim:strat-k-D-in-M}. \end{proof}

\tline

\begin{claim}\label{claim:state-is-edb-ind}Let $i\in\Nat$. We have
$\addlt{\cnfs_{i}(x_{i})}{x_{i}}{s_{i}}=\shprojlt{\proj M{\toloct{\edb{\ded}}}}{x_{i}}{s_{i}}\cup\projlt{\ind M}{x_{i}}{s_{i}}$.\end{claim}

\begin{proof}By definition, 
\[
\cnfs_{i}(x_{i})=\droplt{\left(\shprojlt{\proj M{\toloct{\edb{\ded}}}}{x_{i}}s\cup\projlt{\ind M}{x_{i}}s\right)},
\]
where $s=\locM{i,x_{i}}$. Using Claim~\ref{claim:localM}, we have
$s=s_{i}$. Therefore, 
\[
\addlt{\cnfs_{i}(x_{i})}{x_{i}}{s_{i}}=\shprojlt{\proj M{\toloct{\edb{\ded}}}}{x_{i}}{s_{i}}\cup\projlt{\ind M}{x_{i}}{s_{i}}.
\]
\end{proof}

\tline

\begin{claim}\label{claim:cand-implies-chosen}For each fact $\cand_{R}(x,s,y,u,\bar{a})\in M$,
there is a timestamp $u'\in\Nat$ such that $\chosen_{R}(x,s,y,u',\bar{a})\in M$,
with possibly $u'=u$.\end{claim}

\begin{proof}Towards a proof by contradiction, suppose there is no
such timestamp $u'$. Now, because $\cand_{R}(x,s,y,u,\bar{a})\in M$,
the following ground rule, which is of the form (\ref{eq:chosen}),
can not be in $\grded$, because otherwise $\chosen_{R}(x,s,y,u,\bar{a})\in M$,
which is assumed not to be possible:
\[
\chosen_{R}(x,s,y,u,\bar{a})\gets\cand_{R}(x,s,y,u,\bar{a}).
\]
Because rules of the form (\ref{eq:chosen}) contain a negative $\other_{\ldots}$-atom
in their body, the absence of the above ground rule from $\grded$
implies $\other_{R}(x,s,y,u,\bar{a})\in M$. This $\other_{R}$-fact
must be derived by a ground rule of the form (\ref{eq:other}):
\[
\other_{R}(x,s,y,u,\bar{a})\gets\cand_{R}(x,s,y,u,\bar{a}),\,\chosen_{R}(x,s,y,u',\bar{a}),\, u\neq u'.
\]
But this implies that $\chosen_{R}(x,s,y,u',\bar{a})\in M$, which
is a contradiction. \end{proof}

\tline

\begin{claim}\label{claim:M-largest-stratum}Let $i\in\Nat$. Let
$n$ denote the largest stratum number of the deductive rules of $\ded$.
We have $\shprojlt{\strat Mn}{x_{i}}{s_{i}}=\projlt M{x_{i}}{s_{i}}$.\end{claim}

\begin{proof}

First, since $\strat Mn\subseteq M$, we immediately have $\shprojlt{\strat Mn}{x_{i}}{s_{i}}\subseteq\projlt M{x_{i}}{s_{i}}$. 

Now, let $\fc\in\projlt M{x_{i}}{s_{i}}$. We show $\fc\in\shprojlt{\strat Mn}{x_{i}}{s_{i}}$.
Since $\fc$ has location specifier $x_{i}$ and timestamp $s_{i}$,
we are left to show $\fc\in\strat Mn$. We have the following cases:
\begin{itemize}
\item Suppose $\fc\in\proj M{\toloct{\edb{\ded}}}$. Then $\fc\in\up M\subseteq\strat Mn$.
\item Suppose $\fc\in\proj M{\toloct{\idb{\ded}}}$. Then there is an active
ground rule $\grl\in\grded$ with $\head{\grl}=\fc$. As seen in Section
\ref{sub:model-to-run-defs-notations}, rule $\grl$ can be of three
types: deductive, inductive and delivery. The last two cases would
respectively imply $\fc\in\ind M$ and $\fc\in\deliv M$, giving $\fc\in\up M\subseteq\strat Mn$.
In the deductive case, rule $\grl$ has a stratum number no larger
than $n$, and hence $\fc\in\duc Mn\subseteq\strat Mn$.
\end{itemize}
\end{proof}

\tline

\begin{claim}\label{claim:deliv-is-untag}Let $i\in\Nat$. We have
$\projlt{\deliv M}{x_{i}}{s_{i}}=\addlt{\untag{m_{i}}}{x_{i}}{s_{i}}$.
\end{claim}

\begin{proof}

Let $\fc\in\projlt{\deliv M}{x_{i}}{s_{i}}$. We show $\fc\in\addlt{\untag{m_{i}}}{x_{i}}{s_{i}}$.
Denote $\fc=R(x_{i},s_{i},\bar{a})$. By definition of $\deliv M$,
there is an active delivery rule $\grl\in\grded$ that derives $\fc$:
\[
R(x_{i},s_{i},\bar{a})\gets\chosen_{R}(y,t,x_{i},s_{i},\bar{a}).
\]
Because this rule is active, we have $\chosen_{R}(y,t,x_{i},s_{i},\bar{a})\in M$.
Now, by definition of $x_{i}$ and $s_{i}$, we have $\globM{x_{i},s_{i}}=i$.
Hence, $\pair{\globM{y,t}}{\, R(\bar{a})}\in m_{i}$ and thus $R(\bar{a})\in\untag{m_{i}}$.
Finally, we obtain $\fc=R(x_{i},s_{i},\bar{a})\in\addlt{\untag{m_{i}}}{x_{i}}{s_{i}}$,
as desired.

Let $\fc\in\addlt{\untag{m_{i}}}{x_{i}}{s_{i}}$. We show $\fc\in\projlt{\deliv M}{x_{i}}{s_{i}}$.
Denote $\fc=R(x_{i},s_{i},\bar{a})$. We have $R(\bar{a})\in\untag{m_{i}}$.
Thus, there is some tag $j\in\Nat$ such that $\pair j{R(\bar{a})}\in m_{i}$.
By definition of $m_{i}$, there are values $y\in\nw$, $t\in\Nat$,
$z\in\nw$ and $u\in\Nat$ such that 
\[
\chosen_{R}(y,t,z,u,\bar{a})\in M,
\]
where $\globM{y,t}=j$ and $\globM{z,u}=i$. Here, $\globM{z,u}=i$
implies $z=x_{i}$ and $u=s_{i}$. Hence, $\chosen_{R}(y,t,x_{i},s_{i},\bar{a})\in M$.
Now, the following ground rule $\grl$ is in $\grded$ because (delivery)
rules of the form (\ref{eq:deliv}) are always positive:
\[
R(x_{i},s_{i},\bar{a})\gets\chosen_{R}(y,t,x_{i},s_{i},\bar{a}).
\]
This rule derives $\fc=R(x_{i},s_{i},\bar{a})\in M$ because its body-fact
is in $M$. Hence, $\fc\in\projlt{\deliv M}{x_{i}}{s_{i}}$, as desired.\end{proof}

\tline

\begin{claim}\label{claim:strat-k-M-in-D}Let $i\in\Nat$. Let $k$
be a stratum number (thus $k\geq1$). Suppose that 
\[
\shprojlt{\strat M{k-1}}{x_{i}}{s_{i}}=\shaddlt{\strat{D_{i}}{k-1}}{x_{i}}{s_{i}}.
\]
We have 
\[
\shprojlt{\strat Mk}{x_{i}}{s_{i}}\subseteq\shaddlt{\strat{D_{i}}k}{x_{i}}{s_{i}}.
\]
\end{claim}

\begin{proof}

We consider the fixpoint computation of $M$, i.e., $M=\bigcup_{l\in\Nat}M_{l}$
with $M_{0}=\decl H$ and $M_{l}=T(M_{l-1})$ for each $l\geq1$,
where $T$ is the immediate consequence operator of $\grded$. By
the semantics of operator $T$, we have $M_{l-1}\subseteq M_{l}$. 

We show by induction on $l=0$, $1$, $2$, $\ldots$, that 
\[
\shprojlt{M_{l}\cap\strat Mk}{x_{i}}{s_{i}}\subseteq\shaddlt{\strat{D_{i}}k}{x_{i}}{s_{i}}.
\]
This will imply that 
\[
\projlt{\left(\left(\bigcup_{l\in\Nat}M_{l}\right)\cap\strat Mk\right)}{x_{i}}{s_{i}}\subseteq\shaddlt{\strat{D_{i}}k}{x_{i}}{s_{i}}.
\]
Hence, we obtain, as desired
\[
\shprojlt{M\cap\strat Mk}{x_{i}}{s_{i}}=\shprojlt{\strat Mk}{x_{i}}{s_{i}}\subseteq\shaddlt{\strat{D_{i}}k}{x_{i}}{s_{i}}.
\]

Before we start with the induction, recall from Section \ref{sub:model-to-run-defs-notations}
that 
\begin{eqnarray*}
\strat Mk & = & \up M\cup\duc Mk\\
 & = & \proj M{\toloct{\edb{\ded}}}\cup\ind M\cup\deliv M\cup\duc Mk.
\end{eqnarray*}

\paragraph{Base case ($l=0$)}

We have $M_{0}=\decl H$. Thus $M_{0}$ contains no facts derived
by deductive, inductive or delivery ground rules. Therefore,
\begin{eqnarray*}
M_{0}\cap\strat Mk & = & \proj M{\toloct{\edb{\ded}}}.
\end{eqnarray*}
Hence, 
\begin{eqnarray*}
\shprojlt{M_{0}\cap\strat Mk}{x_{i}}{s_{i}} & \subseteq & \shprojlt{\up M}{x_{i}}{s_{i}}\\
 & \subseteq & \shprojlt{\strat M{k-1}}{x_{i}}{s_{i}}.
\end{eqnarray*}
And by using the given equality $\shprojlt{\strat M{k-1}}{x_{i}}{s_{i}}=\shaddlt{\strat{D_{i}}{k-1}}{x_{i}}{s_{i}}$,
we obtain, as desired:
\begin{eqnarray*}
\shprojlt{M_{0}\cap\strat Mk}{x_{i}}{s_{i}} & \subseteq & \shaddlt{\strat{D_{i}}{k-1}}{x_{i}}{s_{i}}\\
 & \subseteq & \shaddlt{\strat{D_{i}}k}{x_{i}}{s_{i}}.
\end{eqnarray*}

\paragraph*{Induction hypothesis}

Let $l\geq1$. We assume
\[
\shprojlt{M_{l-1}\cap\strat Mk}{x_{i}}{s_{i}}\subseteq\shaddlt{\strat{D_{i}}k}{x_{i}}{s_{i}}.
\]

\paragraph*{Inductive step}

We show 
\[
\shprojlt{M_{l}\cap\strat Mk}{x_{i}}{s_{i}}\subseteq\shaddlt{\strat{D_{i}}k}{x_{i}}{s_{i}}.
\]
Let $\fc\in\shprojlt{M_{l}\cap\strat Mk}{x_{i}}{s_{i}}$. If $\fc\in M_{l-1}$
then $\fc\in\shprojlt{M_{l-1}\cap\strat Mk}{x_{i}}{s_{i}}$ and the
induction hypothesis can be immediately applied. Now suppose that
$\fc\in M_{l}\setminus M_{l-1}$. Then there is a ground rule $\grl\in\grded$
with $\head{\grl}=\fc$ that is active on $M_{l-1}$. We have $\bpos{\grl}\subseteq M_{l-1}$.
As we have seen in Section \ref{sub:model-to-run-defs-notations},
rule $\grl$ can be of three types: deductive, inductive or a delivery.
If $\grl$ is an inductive rule or a delivery rule then 
\begin{eqnarray*}
\fc & \in & \projlt{\ind M}{x_{i}}{s_{i}}\cup\projlt{\deliv M}{x_{i}}{s_{i}}\\
 & \subseteq & \shprojlt{\up M}{x_{i}}{s_{i}}\subseteq\shprojlt{\strat M{k-1}}{x_{i}}{s_{i}}\\
 & = & \shaddlt{\strat{D_{i}}{k-1}}{x_{i}}{s_{i}}\subseteq\shaddlt{\strat{D_{i}}k}{x_{i}}{s_{i}}.
\end{eqnarray*}
Now suppose $\grl$ is deductive. If $\grl$ has stratum less than
or equal to $k-1$, then $\fc\in\shprojlt{\strat M{k-1}}{x_{i}}{s_{i}}$.
In that case, the given equality $\shprojlt{\strat M{k-1}}{x_{i}}{s_{i}}=\shaddlt{\strat{D_{i}}{k-1}}{x_{i}}{s_{i}}$
gives $\fc\in\shaddlt{\strat{D_{i}}{k-1}}{x_{i}}{s_{i}}\subseteq\shaddlt{\strat{D_{i}}k}{x_{i}}{s_{i}}$,
as desired. Now suppose that $\grl$ has stratum $k$. Because $\grl\in\grded$,
there is a rule $\rl\in\pure{\ded}$ and valuation $V$ so that $\grl$
is obtained from $\rl$ by applying valuation $V$ and subsequently
removing the negative (ground) body literals, and so that $V(\bneg{\rl})\cap M=\emptyset$.
Let $\rl'\in\ded$ be the original deductive rule on which $\rl$
is based. Thus $\rl'\in\deduc{\ded}$ (see Section~\ref{sub:subprograms}).
By construction of $\rl$ out of $\rl'$, valuation $V$ can also
be applied to rule $\rl'$. We now show that $V$ is satisfying for
$\rl'$ during the computation of $D_{i}$, in stratum $k$. Since
$V(\head{\rl})=\head{\grl}=\fc$, this results in the derivation of
$V(\head{\rl'})=\droplt{\fc}\in\strat{D_{i}}k$ and thus $\fc\in\shaddlt{\strat{D_{i}}k}{x_{i}}{s_{i}}$,
as desired. It is sufficient to show $V(\bpos{\rl'})\subseteq\strat{D_{i}}k$
and $V(\bneg{\rl'})\cap\strat{D_{i}}{k-1}=\emptyset$ because by the
syntactic stratification, if $\rl'$ uses relations positively then
those relations are in stratum $k$ or lower, and if $\rl'$ uses
relations negatively then those relations are in a stratum strictly
lower than $k$. 
\begin{itemize}
\item We show $V(\bpos{\rl'})\subseteq\strat{D_{i}}k$. First, by the relationship
between $\rl$ and $\rl'$, and because valuation $V$ assigns $x_{i}$
and $s_{i}$ to respectively the body location variable and body timestamp
variable of $\rl$, we have $\bpos{\grl}=V(\bpos{\rl})=\addlt{V(\bpos{\rl'})}{x_{i}}{s_{i}}$.
By choice of $\grl$, we already know $\bpos{\grl}\subseteq M_{l-1}$.
If we could show $\bpos{\grl}\subseteq\strat Mk$ then $\bpos{\grl}\subseteq\shprojlt{M_{l-1}\cap\strat Mk}{x_{i}}{s_{i}}$,
to which the induction hypothesis can be applied to obtain $\bpos{\grl}=\addlt{V(\bpos{\rl'})}{x_{i}}{s_{i}}\subseteq\shaddlt{\strat{D_{i}}k}{x_{i}}{s_{i}}$,
resulting in $V(\bpos{\rl'})\subseteq\strat{D_{i}}k$, as desired.

Now we show $\bpos{\grl}\subseteq\strat Mk$. Let $\fcB\in\bpos{\grl}$.
If $\fcB\in\up M$ then we immediately have $\fcB\in\strat Mk$. Now
suppose that $\fcB\notin\up M$. Since $\bpos{\grl}\subseteq\projlt M{x_{i}}{s_{i}}$,
we have $\fcB\in\projlt M{x_{i}}{s_{i}}\setminus\up M$. Then Claim
\ref{claim:M-largest-stratum} implies there is an active deductive
ground rule $\grl'\in\grded$ with $\head{\grl'}=\fcB$. But we are
working with a syntactic stratification, and thus the stratum of $\grl'$
can not be higher than the stratum of $\grl$, which is $k$. Hence
$\fcB\in\duc Mk\subseteq\strat Mk$.

\item We show $V(\bneg{\rl'})\cap\strat{D_{i}}{k-1}=\emptyset$. By choice
of $\rl$ and $V$, we have $V(\bneg{\rl})\cap M=\emptyset$. So,
\[
V(\bneg{\rl})\cap\shprojlt{\strat M{k-1}}{x_{i}}{s_{i}}=\emptyset.
\]
By applying the given equality $\shprojlt{\strat M{k-1}}{x_{i}}{s_{i}}=\shaddlt{\strat{D_{i}}{k-1}}{x_{i}}{s_{i}}$,
we then have $V(\bneg{\rl})\cap\shaddlt{\strat{D_{i}}{k-1}}{x_{i}}{s_{i}}=\emptyset$.
By the relationship between $\rl$ and $\rl'$, we have $V(\bneg{\rl})=\addlt{V(\bneg{\rl'})}{x_{i}}{s_{i}}$.
Thus $V(\bneg{\rl'})\cap\strat{D_{i}}{k-1}=\emptyset$, as desired. 
\end{itemize}
\end{proof}

\tline

\begin{claim}\label{claim:strat-k-D-in-M}Let $i\in\Nat$. Let $k$
be a stratum number (thus $k\geq1$). Suppose that 
\[
\shprojlt{\strat M{k-1}}{x_{i}}{s_{i}}=\shaddlt{\strat{D_{i}}{k-1}}{x_{i}}{s_{i}}.
\]
We have 
\[
\shaddlt{\strat{D_{i}}k}{x_{i}}{s_{i}}\subseteq\shprojlt{\strat Mk}{x_{i}}{s_{i}}.
\]
\end{claim}

\begin{proof}Recall that the semantics of stratum $k$ in $\deduc{\ded}$
is that of semi-positive $\datalogneg$, with input $\strat{D_{i}}{k-1}$.
So, we can consider $\strat{D_{i}}k$ to be a fixpoint, i.e., as the
set $\bigcup_{l\in\Nat}A_{l}$ with $A_{0}=\strat{D_{i}}{k-1}$ and
$A_{l}=T(A_{l-1})$ for each $l\geq1$, where $T$ is the immediate
consequence operator of stratum $k$ in $\deduc{\ded}$. We show by
induction on $l=0$, $1$, $2$, etc, that 
\[
\shaddlt{A_{l}}{x_{i}}{s_{i}}\subseteq\shprojlt{\strat Mk}{x_{i}}{s_{i}}.
\]
This then gives us the desired result.

\paragraph*{Base case ($l=0$)}

We have $A_{0}=\strat{D_{i}}{k-1}$. By applying the given equality,
we obtain 
\[
\shaddlt{A_{0}}{x_{i}}{s_{i}}=\shaddlt{\strat{D_{i}}{k-1}}{x_{i}}{s_{i}}=\shprojlt{\strat M{k-1}}{x_{i}}{s_{i}}\subseteq\shprojlt{\strat Mk}{x_{i}}{s_{i}}.
\]

\paragraph*{Induction hypothesis}

Let $l\geq1$. We assume 
\[
\shaddlt{A_{l-1}}{x_{i}}{s_{i}}\subseteq\shprojlt{\strat Mk}{x_{i}}{s_{i}}.
\]

\paragraph*{Inductive step}

Let $\fc\in A_{l}$. We show $\addlt{\fc}{x_{i}}{s_{i}}\in\shprojlt{\strat Mk}{x_{i}}{s_{i}}$.
If $\fc\in A_{l-1}$ then the induction hypothesis can be applied
to obtain the desired result. Now suppose $\fc\in A_{l}\setminus A_{l-1}$.
Let $\rl\in\deduc{\ded}$ and $V$ be respectively a rule with stratum
$k$ and a valuation that together have derived $\fc\in A_{l}$. Let
$\rl'\in\pure{\ded}$ be the rule obtained from $\rl$ by applying
transformation (\ref{eq:pure-duc}). Let $V'$ be the extension of
$V$ to assign $x_{i}$ and $s_{i}$ respectively to the body location
and timestamp variable of $\rl'$, which are also both used in the
head of $\rl'$. Let $\grl$ be the ground rule obtained from $\rl'$
by applying valuation $V'$ and by subsequently removing all negative
body literals. We show $\grl\in\grded$ and $\bpos{\grl}\subseteq M$,
which then implies 
\[
\head{\grl}=V'(\head{\rl'})=\addlt{V(\head{\rl})}{x_{i}}{s_{i}}=\addlt{\fc}{x_{i}}{s_{i}}\in M.
\]
Moreover, because $\rl$ (and thus $\rl'$) has stratum $k$, rule
$\grl$ is an active deductive ground rule with stratum $k$, and
thus $\addlt{\fc}{x_{i}}{s_{i}}\in\shprojlt{\duc Mk}{x_{i}}{s_{i}}\subseteq\shprojlt{\strat Mk}{x_{i}}{s_{i}}$,
as desired.
\begin{itemize}
\item To show $\grl\in\grded$, we require $V'(\bneg{\rl'})\cap M=\emptyset$.
Because $V$ is satifying for $\rl$, and because negation is only
applied to lower strata, we have 
\[
V(\bneg{\rl})\cap\strat{D_{i}}{k-1}=\emptyset.
\]
Thus
\[
\addlt{V(\bneg{\rl})}{x_{i}}{s_{i}}\cap\shaddlt{\strat{D_{i}}{k-1}}{x_{i}}{s_{i}}=\emptyset.
\]
By the relationship between $\rl$ and $\rl'$, we have $\addlt{V(\bneg{\rl})}{x_{i}}{s_{i}}=V'(\bneg{\rl'})$,
which gives us 
\[
V'(\bneg{\rl'})\cap\shaddlt{\strat{D_{i}}{k-1}}{x_{i}}{s_{i}}=\emptyset.
\]
And by using the given equality $\shprojlt{\strat M{k-1}}{x_{i}}{s_{i}}=\shaddlt{\strat{D_{i}}{k-1}}{x_{i}}{s_{i}}$,
we have 
\[
V'(\bneg{\rl'})\cap\shprojlt{\strat M{k-1}}{x_{i}}{s_{i}}=\emptyset.
\]
Now, for the last step, we work towards a contradiction: suppose that
there is a fact $\fcB\in V'(\bneg{\rl'})\cap M$. From the construction
of $\rl'$, we know that $\fcB$ is over $\toloct{\schof{\ded}}$
and has location specifier $x_{i}$ and timestamp $s_{i}$. 

\begin{itemize}
\item If $\fcB$ is over $\toloct{\edb{\ded}}$ then $\fcB\in\shprojlt{\proj M{\toloct{\edb{\ded}}}}{x_{i}}{s_{i}}$.
Thus $\fcB\in\shprojlt{\up M}{x_{i}}{s_{i}}\subseteq\shprojlt{\strat M{k-1}}{x_{i}}{s_{i}}$,
which is a contradiction.
\item If $\fcB$ is over $\toloct{\idb{\ded}}$ then there is an active
ground rule $\grl'\in\grded$ with $\head{\grl'}=\fcB$. As seen in
Section~\ref{sub:model-to-run-defs-notations}, rule $\grl'$ is
either deductive, inductive or a delivery. The last two cases would
imply that $\fcB\in\shprojlt{\ind M\cup\deliv M}{x_{i}}{s_{i}}\subseteq\shprojlt{\up M}{x_{i}}{s_{i}}$,
which gives a contradiction like in the previous case. Now suppose
that $\grl'$ is deductive. Because the predicate of $\fcB$ is used
negatively in $\rl'$ and thus negatively in $\rl$, the syntactic
stratification assigns a smaller stratum number to $\grl'$ than the
stratum number of $\grl$, which is $k$. Hence, $\fcB\in\shprojlt{\strat M{k-1}}{x_{i}}{s_{i}}$,
which is again a contradiction.
\end{itemize}

We conclude that $V'(\bneg{\rl'})\cap M=\emptyset$.

\item We show $\bpos{\grl}\subseteq M$. Because $V$ is satisfying for
$\rl$, we have 
\[
V(\bpos{\rl})\subseteq A_{l-1}.
\]
By the relationship between $\rl$ and $\rl'$ (and $\grl$), we have
$\addlt{V(\bpos{\rl})}{x_{i}}{s_{i}}=V'(\bpos{\rl'})=\bpos{\grl}$.
Thus 
\[
\bpos{\grl}\subseteq\shaddlt{A_{l-1}}{x_{i}}{s_{i}}.
\]
By now applying the induction hypothesis, we obtain, as desired: 
\[
\bpos{\grl}\subseteq\shprojlt{\strat Mk}{x_{i}}{s_{i}}\subseteq M.
\]

\end{itemize}
\end{proof}

\tline

\end{appendix}

\end{document}